%% file: main.tex
\documentclass[sigconf, nonacm]{acmart}
\usepackage{algorithm}
\usepackage{algorithmic}
\usepackage{amsmath}
\usepackage{subfigure}
\newcommand\vldbdoi{XX.XX/XXX.XX}
\newcommand\vldbpages{XXX-XXX}
\newcommand\vldbvolume{14}
\newcommand\vldbissue{1}
\newcommand\vldbyear{2020}
\newcommand\vldbauthors{\authors}
\newcommand\vldbtitle{\shorttitle} 
\newcommand\vldbavailabilityurl{https://github.com/mshubhankar/DP-DataGeneration-MissingData}
\newcommand\vldbpagestyle{plain}

\input{commands}

\newif\ifpaper
\paperfalse   

\begin{document}

\title{Differentially Private Data Generation with Missing Data}

\author{Shubhankar Mohapatra, Jianqiao Zong, Florian Kerschbaum, Xi He}
\affiliation{%
  \institution{University of Waterloo}
  }
\email{shubhankar.mohapatra, jianqiao.zong, florian.kerschbaum, xi.he@uwaterloo.ca}

\begin{abstract}

Despite several works that succeed in generating synthetic data with differential privacy (DP) guarantees, they are inadequate for generating high-quality synthetic data when the input data has missing values. In this work, we formalize the problems of DP synthetic data with missing values and propose three effective adaptive strategies that significantly improve the utility of the synthetic data on four real-world datasets with different types and levels of missing data and privacy requirements. We also identify the relationship between privacy impact for the complete ground truth data and incomplete data for these DP synthetic data generation algorithms. We model the missing mechanisms as a sampling process to obtain tighter upper bounds for the privacy guarantees to the ground truth data. Overall, this study contributes to a better understanding of the challenges and opportunities for using private synthetic data generation algorithms in the presence of missing data.
\end{abstract}

\maketitle

\ifpaper
\pagestyle{\vldbpagestyle}
\begingroup\small\noindent\raggedright\textbf{PVLDB Reference Format:}\\
\vldbauthors. \vldbtitle. PVLDB, \vldbvolume(\vldbissue): \vldbpages, \vldbyear.\\
\href{https://doi.org/\vldbdoi}{doi:\vldbdoi}
\endgroup
\begingroup
\renewcommand\thefootnote{}\footnote{\noindent
This work is licensed under the Creative Commons BY-NC-ND 4.0 International License. Visit \url{https://creativecommons.org/licenses/by-nc-nd/4.0/} to view a copy of this license. For any use beyond those covered by this license, obtain permission by emailing \href{mailto:info@vldb.org}{info@vldb.org}. Copyright is held by the owner/author(s). Publication rights licensed to the VLDB Endowment. \\
\raggedright Proceedings of the VLDB Endowment, Vol. \vldbvolume, No. \vldbissue\ %
ISSN 2150-8097. \\
\href{https://doi.org/\vldbdoi}{doi:\vldbdoi} \\
}\addtocounter{footnote}{-1}\endgroup

\ifdefempty{\vldbavailabilityurl}{}{
\vspace{.3cm}
\begingroup\small\noindent\raggedright\textbf{PVLDB Artifact Availability:}\\
The source code, data, and/or other artifacts have been made available at \url{\vldbavailabilityurl}.
\endgroup
}
\else
\fi

\input{intro}
\input{prelim}
\input{problem}
\input{dpmiss}

\input{amplification}
\input{experiments}
\input{related}
\input{conclusion}



\bibliographystyle{ACM-Reference-Format}
\bibliography{ref}

\end{document}

%% file: commands.tex
\newcommand{\revision}[1]{{{#1}}}

\newcommand{\stitle}[1]{\smallskip \noindent{\bf #1}}
\newcommand{\eat}[1]{}

\newcommand{\groundtruth}{\ensuremath{\bar{D}}}
\newcommand{\incompletedata}{\ensuremath{D}}
\newcommand{\syntheticdata}{\ensuremath{D^*}}
\DeclareMathOperator*{\E}{\mathbb{E}}
\newcommand{\attrset}{\mathcal{A}}

\usepackage{enumitem}

\newtheorem{problem}{Problem}

\newcommand{\squishlist}{
	\begin{list}{$\bullet$}
		{
			\setlength{\itemsep}{0pt}
			\setlength{\parsep}{3pt}
			\setlength{\topsep}{3pt}
			\setlength{\partopsep}{0pt}
			\setlength{\leftmargin}{1.5em}
			\setlength{\labelwidth}{1em}
			\setlength{\labelsep}{0.5em} } }

	\newcommand{\squishend}{
\end{list}  }

%% file: intro.tex
\section{Introduction}\label{sec:intro}

Our world as we see it today revolves around private data regarding 
our medical, financial, and social information. It is sometimes imperative to 
query such data for research and advancement of science~\cite{kass2003use, andreou2013should}. Many 
industries also use statistics from private data to improve their products and user experience. However, reckless data sharing for data-driven applications and research causes great privacy concerns~\cite{user_concern, ibm_report} and penalties~\cite{GDPR}.  As a response, differential privacy (DP)~\cite{DBLP:conf/icalp/Dwork06} has emerged as a standard data privacy guarantee which has now been adopted by government agencies~\cite{DBLP:conf/kdd/Abowd18, Hawes2020Implementing} and companies~\cite{Erlingsson14Rappor, Greenberg16Apples, DBLP:journals/pvldb/JohnsonNS18}. Informally, DP guarantees that the output distribution of an algorithm is similar with or without a particular individual in the dataset. 
A privacy budget is set to limit the total privacy loss and each query (e.g., releasing  statistics~\cite{DBLP:conf/kdd/Abowd18, onthemap}, building prediction models~\cite{DBLP:conf/iclr/PapernotSMRTE18, DBLP:conf/ccs/AbadiCGMMT016}, and answering SQL queries~\cite{DBLP:conf/sigmod/0002HIM19, DBLP:journals/pvldb/JohnsonNS18, DBLP:journals/pvldb/KotsogiannisTHF19, DBLP:conf/sigmod/McSherry09}) consumes part of the privacy budget, and once that budget is used up, no more queries can be answered directly. An alternative way is to generate a synthetic dataset using the privacy budget. The synthetic dataset, once generated, can be made public, and the analyst can use it for any number of downstream tasks 
~\cite{DBLP:journals/corr/abs-1812-02274, DBLP:conf/kdd/ChenXZX15, DBLP:conf/iclr/JordonYS19a, DBLP:conf/sigmod/ZhangCPSX14, DBLP:conf/pods/BarakCDKMT07}.

Despite a number of work~\cite{DBLP:conf/sigmod/ZhangCPSX14, ge2021kamino, phan2016differential, DBLP:conf/iclr/JordonYS19a} that succeed in generating synthetic data with DP guarantees, they only look at a simple scenario where the input data has no missing values. Several prior studies~\cite{little2019statistical,getz2023performance} have reported on the prevalence of missing data in various fields. 

For instance, a study of 9 publicly available healthcare datasets commonly used in machine learning research found that the proportion of missing values ranged from 0.2\% to 78.6\%~\cite{getz2023performance}. The presence of missing data 
can be attributed to multiple reasons, such as 
human errors~\cite{agarwal2021causal} and privacy regulations like GDPR~\cite{GDPR} which allow people the ``right to forget'' where one may ask their data to be deleted completely~\cite{politou2018forgetting}.

In our work, we ask how missing data will affect the quality of the synthetic data generated by mechanisms that offer DP guarantees. Our preliminary study shows that existing DP mechanisms have 4\%-18.5\% decay in the F1-score of downstream ML tasks on the synthetic dataset generated from a 
dataset with 10\% missing values as compared to that when generated from the complete dataset. The decay varies depending on the types of missing mechanisms and the types of data generation processes.

\revision{In our work, we formally define the research problem of generating synthetic data for sensitive data with missing values using DP. We consider a missing mechanism that takes complete ground truth data and outputs data with missing values. Under this setup, we can offer DP to either \emph{the incomplete data} or \emph{the ground truth data}. For each privacy guarantee, we study how to handle missing data in the synthetic data generation process.}
\revision{While it is desired to have a unified comprehensive approach to solve the problem, we identify multiple challenges in achieving such a solution. First, there is no straightforward winning algorithm for DP synthetic data generation even without missing data as demonstrated by prior benchmarking work~\cite{tao2021benchmarking}. Second, the design space for a solution is huge for dealing with missing data. We explore several techniques, including the vanilla approach that uses complete rows only~\cite{nakagawa2008missing}, common imputation techniques (e.g., statistical methods, machine learning methods)~\cite{little2019statistical, rubin1976inference, nordholt1998imputation, lakshminarayan1996imputation}, and other differentially private imputation methods~\cite{DBLP:journals/corr/abs-2206-15063,ge2021kamino}. We show that all these methods have their limitations, such as discarding too many rows or incurring high privacy costs, leading to poor-quality synthetic data. Therefore in our work we present a comprehensive list of possible approaches to solve the problem including two vanilla approaches and three adaptive approaches, as a contribution to understanding the design space for this new problem. }

In addition, we show the relationship between the privacy guarantee for these DP mechanisms that they offer for the incomplete data and that for the ground truth data. To do so, we model the missing mechanism as a sampling process and obtain a tighter upper bound for the privacy loss to the ground truth data via sampling amplification techniques~\cite{balle2018privacy, steinke2022composition}. Unlike prior work for sampling amplification that considers a random subset, we make use of the randomness due to missing values to amplify the privacy for ground truth data. The major contributions of our work are as follows:

\squishlist
\item We are the first to formalize and study the problems of DP synthetic data with missing data. Our results show that existing algorithms have a decrease of 5-23\% in utility with 
$\leq$ 5\% missing values and 
     a decrease of 10-190\% with 
     $\leq$ 20\% missing values. 

     \item We develop three novel adaptive approaches, each tailored to an existing category of DP mechanisms, that seamlessly combine dealing with missing data along with the learning process. Our evaluation shows that they improve the utility of the synthetic datasets by up to 15-72\%. These simple yet effective approaches sometimes even achieve the same utility as the synthetic data generated from the no-missing ground truth data.
     \item We differentiate the privacy guarantees for both incomplete and ground truth data. Our analysis develops sufficient conditions that when satisfied, the algorithms for the incomplete data can be used to achieve privacy for the ground truth data. 
     \item We are the first to apply amplification due to missing mechanisms and tighten the privacy bound for ground truth data. The amplified ground truth privacy is 0.1-0.65x the privacy achieved for the incomplete data with 10-50\% missing values.
\squishend 

%% file: prelim.tex
\section{Preliminaries}\label{sec:preliminaries}

We consider a database relation $R = \{A_1, \cdots, A_k\}$ with $k$ attributes, and a database instance $D$ consisting of $n$ rows. 
We use $D_i$ to refer to the $i$th row of $D$, and $D_{ij}$ to refer to the $j$th attribute of row $D_i$. We also use $S_{:i}$ to denote all elements from $1$ to $i$ in a sequence $S$.   

\subsection{Missing Data}\label{sec:preliminaries:missing}

For missing data, we define a missingness indicator matrix $M=[\ldots,m_{ij},\ldots]$ of size $n\times k$, where $i \in [1,n] , j \in [1, k]$ and, shorthand $m_i$ to point to $i^{th}$ row of $M$. Each cell of $M$ has one-to-one relation with $D$ such that, $m_{ij} = 1$ if $D_{ij}$ is missing and $m_{ij} = 0$ otherwise.

Missing data is classified into different types using missing mechanisms. A missing mechanism $M_\phi: \mathcal{D} \rightarrow \mathcal{D}$ takes as input the ground truth dataset $\bar{D}$ and outputs an incomplete dataset $D$. It is parameterized by $\Phi$, a set of probabilities, which refers to the set of probabilities that control the unknown missing data process. Three missing types can be defined using $\Phi$ and the conditional distribution of missing indicator $m_i$ given the dataset $D_i$
~\cite{little2019statistical, rubin1976inference}.

\stitle{Missing completely at random (MCAR)}  assumes the probability of missingness is completely independent of the data. Under MCAR, any two rows of the dataset, regardless of their values, for the same attribute have the same probability of having a missing value. Hence, the parameter set $\Phi$ consists of $\{\phi_j |j\in [1, k]$\}, where $\phi_j$ is the probability of any row having a missing value for the $j$th column.

For  $j\in [1,k]$ and $i\in [1,n]$,
$         \Pr[m_{ij}|D] = \Pr[m_{ij}]= \phi_j   $.

Hence, the probability of a row having no missing values is  $\prod_{j=1}^k(1-\phi_j)$.

\stitle{Missing at random (MAR)} captures the scenario when the probability of missingness is independent of the missing values given the observed data. In other words, under MAR how likely a value is to be missing can be estimated based on the non-missing data. Consider examples, 1) Young people have missing IQ (because they haven't taken an IQ test yet), and MAR models the same probability of missing IQ attribute for rows of the same age, regardless of their IQ values; 2) Businessmen are less likely to share their income, and MAR models the same probability for income values for rows that have an occupation as `Business.'  MAR, therefore, is parameterized by a set of conditional probabilities $\Phi$ where each $\phi \in |\Phi|$ maps relationships between observed and missing values in the dataset. 
\revision{For $\phi_x\in \Phi$ that models the $j$th column's missingness, and for $i\in [1,n]$,
$\Pr[m_{ij}|D_{(0)}] = 
\Pr[m_{ij}|D_{(0)},D_{ij}] = \phi_x$,
    where $D_{(0)}$ refers to the observed attributes in the dataset.}

\stitle{Missing not at random (MNAR)}  captures the scenario when given all the observed information, the probability of missingness depends on \revision{any other unobserved missing values in the dataset}. Consider examples, \revision{1) Students with missing attendance also have missing scores, and the probability of the missing scores depending on missing attendance} 2) People who smoke don't want to mention they smoke. Here MNAR models the probability of missing smokers based on the attribute of smoking. With MNAR missingness, $\Phi$ consists of a set of conditional probabilities that map the probability of an attribute to be missing given its own value. 

Datasets often contain various missing data types, but identifying them without domain knowledge remains challenging~\cite{little1988test, little1994class}. While statistical tests exist for MCAR cases (e.g., Little's MCAR test, pattern mixture models)~\cite{little1988test, little1994class}, identifying MAR and MNAR remains unsolved due to complex interactions between observed and unobserved variables~\cite{van2018flexible}. Despite binary indicator tests based on modeling for MAR, these lack conclusiveness and are sensitive to analysis assumptions. For private datasets, these tests also need to be computed privately by using a portion of the privacy budget.

\subsection{Differential Privacy}\label{sec:preliminaries:dp}

Differential privacy (DP)~\cite{DworkMNS06, DBLP:conf/eurocrypt/DworkKMMN06} is used as our measure of privacy.

\begin{definition}[Differential Privacy (DP)~\cite{DBLP:journals/fttcs/DworkR14, DBLP:conf/eurocrypt/DworkKMMN06}]
A randomized algorithm $M$ achieves ($\epsilon,\delta$)-DP if for all $Z \subseteq$ Range($\mathcal{M}$) and for two neighboring databases $D,D' \in \mathcal{D}$ that differ in one row:
$$\Pr[\mathcal{M}(D) \in Z]\le e^\epsilon \Pr[\mathcal{M}(D') \in Z] + \delta.$$
\end{definition}

The privacy cost is measured by the parameters ($\epsilon,\delta$), often referred to also as the privacy budget. The smaller the privacy parameters, the stronger the offered privacy. 

Gaussian mechanism~\cite{DBLP:journals/fttcs/DworkR14} and Laplace mechanism~\cite{DworkMNS06} are two such widely used DP algorithms. 
Given a function $f:\mathcal{D} \rightarrow \mathbb{R}^d$, the Gaussian mechanism adds noise sampled from a Gaussian distribution $\mathcal{N}(0,S_f^2\sigma^2)$ to each component of the query output, where $\sigma$ is the noise scale and $S_f$
is the $L_2$ sensitivity of function $f$, 
which is defined as
$S_f = \max_{D,D'\text{differ in a row}} ||f(D) - f(D')||_2$. 
For $\epsilon\in (0,1)$, if $\sigma \geq \sqrt{2\ln(1.25/\delta)}/\epsilon$, 
 the Gaussian mechanism satisfies $(\epsilon,\delta)$-DP.
Laplace mechanism works similarly but with the noise from the Laplace distribution and the $L_1$ sensitivity.
Both these mechanisms have been applied to answer counting queries~\cite{DBLP:journals/vldb/LiMHMR15} and is widely used in estimating low dimensional statistics about the dataset. 

Complex DP algorithms can be built from these basic algorithms following two important properties of DP:
1) Post-processing~\cite{DBLP:conf/eurocrypt/DworkKMMN06} states that for any function $g$ defined over the output of the mechanism $\mathcal{M}$, if $\mathcal{M}$ satisfies ($\epsilon,\delta$)-DP, so does $g(\mathcal{M})$;
2) Composability~\cite{DBLP:conf/icalp/Dwork06} states that if $\mathcal{M}_1$, $\mathcal{M}_2$, $\cdots$, $\mathcal{M}_k$ satisfy ($\epsilon_1,\delta_1$)-, ($\epsilon_1,\delta_1$)-, $\cdots$, ($\epsilon_k,\delta_k$)-DP, 
then sequentially applying these mechanisms 
satisfies ($\sum_{i=1}^k\epsilon_i, \sum_{i=1}^k\delta_i$)-DP.

\subsection{DP Synthetic Data Generation}\label{sec:preliminaries:dpsyn}

A common DP study is to generate synthetic data given a fixed privacy budget. The synthetic data, once generated, can be made public and all queries on this dataset come for free due to the post-processing property of DP. There are three main approaches for DP synthetic data generation~\cite{tao2021benchmarking}:

\stitle{Statistical approaches}
rely on estimating low-dimensional statistics about the dataset such as marginals~\cite{DBLP:journals/tkde/XiaoWG11, DBLP:conf/sigmod/QardajiYL14, mckenna2022aim}. These approaches can be made better by finding the correlation between attributes. Techniques for improvement include probabilistic models~\cite{DBLP:books/daglib/0023091}, Bayesian models~\cite{DBLP:conf/sigmod/ZhangCPSX14, DBLP:journals/pvldb/LiXZJ14, DBLP:conf/ssdbm/PingSH17} and undirected graphs~\cite{DBLP:conf/kdd/ChenXZX15, DBLP:conf/icml/McKennaSM19}. Statistical approaches capture the underlying distribution of the correlated independent attributes very well but fail to imbibe complex relationships between multiple attributes.

\stitle{Deep learning approaches} \revision{are promising for generating synthetic data~\cite{DBLP:conf/cvpr/ShrivastavaPTSW17, Chawla2019DeepfakesH,DBLP:conf/icassp/Gupta19}, particularly with autoencoders and generative adversarial networks (GAN). Autoencoders map data into a low-dimensional feature space and sample synthetic data from the low-dimensional space. GANs utilize a generator to produce fake examples and a discriminator to distinguish real from fake. DPSGD~\cite{WM10, BST14, song2013stochastic, DBLP:conf/ccs/AbadiCGMMT016} is commonly used for privacy. Various private approaches for autoencoders \cite{phan2016differential, abay2018privacy, acs2018differentially} and GANs \cite{DBLP:conf/sec/FrigerioOGD19, DBLP:conf/iclr/JordonYS19a, DBLP:journals/corr/abs-2001-09700, DBLP:journals/corr/abs-1802-06739} have been proposed, effective on image data but challenged with tabular data due to poor encoding. Conditional GANs \cite{xu2019modeling} and private versions \cite{torkzadehmahani2019dp} address this by sampling based on conditional probabilities of categorical attributes.}

\stitle{Mixed approaches}
are inspired by both the above approaches and try to preserve both low-dimensional statistics and high-level information. Some techniques include leveraging the dimensionality reduction via random orthonormal (RON) projection, the Gaussian generative model~\cite{chanyaswad2019ron}, combining denial constraints and attribute-wise embedding models~\cite{ge2021kamino} and Gretel.ai statistics~\cite{noruzmangretel}.

%% file: problem.tex
\section{Problem Statement}\label{sec:prob}
Consider a private dataset behind a privacy firewall with $n$ rows and $k$ attributes. A trusted curator aims to generate a synthetic version of the same size with an end-to-end ($\epsilon, \delta)$-DP guarantee while preserving maximum utility\revision{. Real-world data collected by curators may contain missing values, a scenario overlooked in prior work. We address this by formalizing two versions of the problem based on privacy considerations.} First, we offer a DP guarantee for the incomplete dataset held by the data curator.

\begin{problem}\label{prob1}[Privacy for Incomplete Data] 
Consider collecting data from a ground truth data $\groundtruth$ of $n$ rows owned by $n$ individuals, a missing mechanism $M_{\Phi}:\mathcal{D}\rightarrow \mathcal{D}$ is involved that takes in $\groundtruth$  and outputs a dataset $\incompletedata$ of $n$ rows but with missing values.  A trusted data curator uses this dataset $\incompletedata$ as input and aims to generate a synthetic data $\syntheticdata$ of $n$ rows with a mechanism $M:\mathcal{D}\rightarrow \mathcal{D}$ such that $\syntheticdata$ \revision{minimizes $d(f(\incompletedata)-f(\syntheticdata))$, where $f: \mathcal{D} \rightarrow \mathbb{R}^l$ is any utility metric function that the user is interested in, $d(\cdot,\cdot)$ is a distance metric and \emph{$M$ offers $(\epsilon,\delta)$-DP to the input data $\incompletedata$}}. 
\end{problem}

\ifpaper
\else
In Section~\ref{sec:dpincomplete}, we delineate multiple options for the data curator to deal with incomplete data, put forward challenges that come with them, and discuss which option might be the best and when. The first option is a na\"ive adaptation of prior work for DP data generation by simply discarding the rows with missing values~\cite{rubin1976inference}. We refer to this approach as \emph{complete row only}. This approach can fail in many cases, as detailed in Section~\ref{sec:vanilla}. For instance, if all rows have some missing values, then there will be no input data for the data generation methods. 
A second approach is to impute the missing parts with inferred values from the observed data. We denote this approach as \emph{imputation first approach}. However, as our data is private, the imputation process needs to be privatized as well and the additional incurred privacy cost must be accounted for in the privacy budget. 
In section~\ref{sec:vanilla}, we explore the privacy costs of imputation and show how they can be expensive in practice.  
Rather than having separate processes for imputation and synthetic data generation, we can integrate these two processes into one.  This line of thought motivates us to a new approach, which we call \emph{adaptive recourse approach}. In Section~\ref{sec:adaptive}, we improve upon three categories of DP generation approaches and demonstrate their effectiveness in generating synthetic data from incomplete data. These strategies use no extra privacy budget and solely improve by observing available data in the dataset.
\fi

The incomplete data $\incompletedata$ can be modeled as a sample generated from a complete ground truth dataset $\groundtruth$ via a missing mechanism $M_\Phi$. If the privacy goal is to protect the ground truth dataset $\groundtruth$ with DP guarantee, how will the problem and the solution be different? We formalize the second problem as follows.

\begin{problem}\label{prob2} [Privacy for Ground Truth Data] Consider the same setup as Problem~\ref{prob1}.

The trusted data curator uses the incomplete dataset $\incompletedata$ as input and aims to generate a synthetic data $\syntheticdata$ of $n$ rows with a mechanism $M:\mathcal{D}\rightarrow \mathcal{D}$ such that 
\revision{minimizes $d(f(\groundtruth)-f(\syntheticdata))$, where $f: \mathcal{D} \rightarrow \mathbb{R}^l$ is any utility metric function that the user is interested in, $d(\cdot,\cdot)$ is a distance metric} and \emph{$M\circ M_{\Phi}$ offers $(\bar{\epsilon},\bar{\delta})$-DP} \emph{to the ground truth data $\groundtruth$}. 

\end{problem}

The problem mentioned above differs from our initial one only in the final aspect: instead of aiming for DP for the observed incomplete data, we target DP for the ground truth data. The missing mechanism limits the information available for synthetic data generation. Although initially appearing similar, ensuring privacy for incomplete data doesn't necessarily guarantee privacy for ground truth data. In Section~\ref{sec:amplification}, we delve into their relationship, showcasing scenarios where privacy for incomplete data may or may not extend to ground truth data. We also explore how the missing mechanism can serve as a sampling mechanism to enhance privacy and improve the utility of synthetic data.
\ifpaper
While the paper contains numerous theorems and lemmas, some proofs are available in the extended version \cite{full_paper}.
\else
\fi

%% file: dpmiss.tex
\section{Privacy for Incomplete Data}\label{sec:dpincomplete}

This section examines Problem~\ref{prob1} and explores methods for generating synthetic data from an incomplete private dataset. The section starts by discussing two vanilla methods that are found to be ineffective in the DP context and instead recommends adaptive recourse methods that are novel solutions that address both issues and produce better-quality synthetic data.

\input{complete}

\input{impute}
\input{adaptiverecourse}

%% file: complete.tex
\subsection{Vanilla Approaches} \label{sec:vanilla}
Complete row only and imputation first are two traditional methods for handling missing data. These methods either involve discarding rows with missing values or filling up missing information with values inferred from the observed data.

\stitle{Complete Row Only Approach.} This approach is effective when the missingness is completely at random (MCAR) since the distribution of each attribute remains the same after removing missing rows. However, for other types of missingness, such as missing at random (MAR) and missing not at random (MNAR), the complete row only approach can lead to biased results. Hence, a standard synthetic data generation algorithm that learns directly from the remaining complete rows will result in a biased data distribution that is different from the ground truth data. 

\ifpaper
We run an experiment with the Adult dataset~\cite{uci_repo} where we compare the difference
between the original dataset and the generated synthetic datasets from four different DP synthetic data generation algorithms and show that such an approach affects MAR and MNAR much more than it does for MCAR. This is because MAR and MNAR introduce inherent bias to the estimated distribution of attributes. For space constraints, we defer this experiment to the full paper~\cite{full_paper}. Besides the potential bias issue for the complete row only approach, the number of complete rows remaining can be very small. For example, the ground truth Adult dataset which has 32k rows reduces to $\approx5$k complete rows with 20\% MAR and $\approx 1$k complete rows with 20\% MCAR/MNAR missing mechanism respectively. Our results in Section~\ref{sec:exp} show that the number of complete rows plays a vital role in the performance of the synthetic data generation algorithms.

\else
\begin{figure}[htb]
    \centering
    \includegraphics[width=0.8\linewidth]{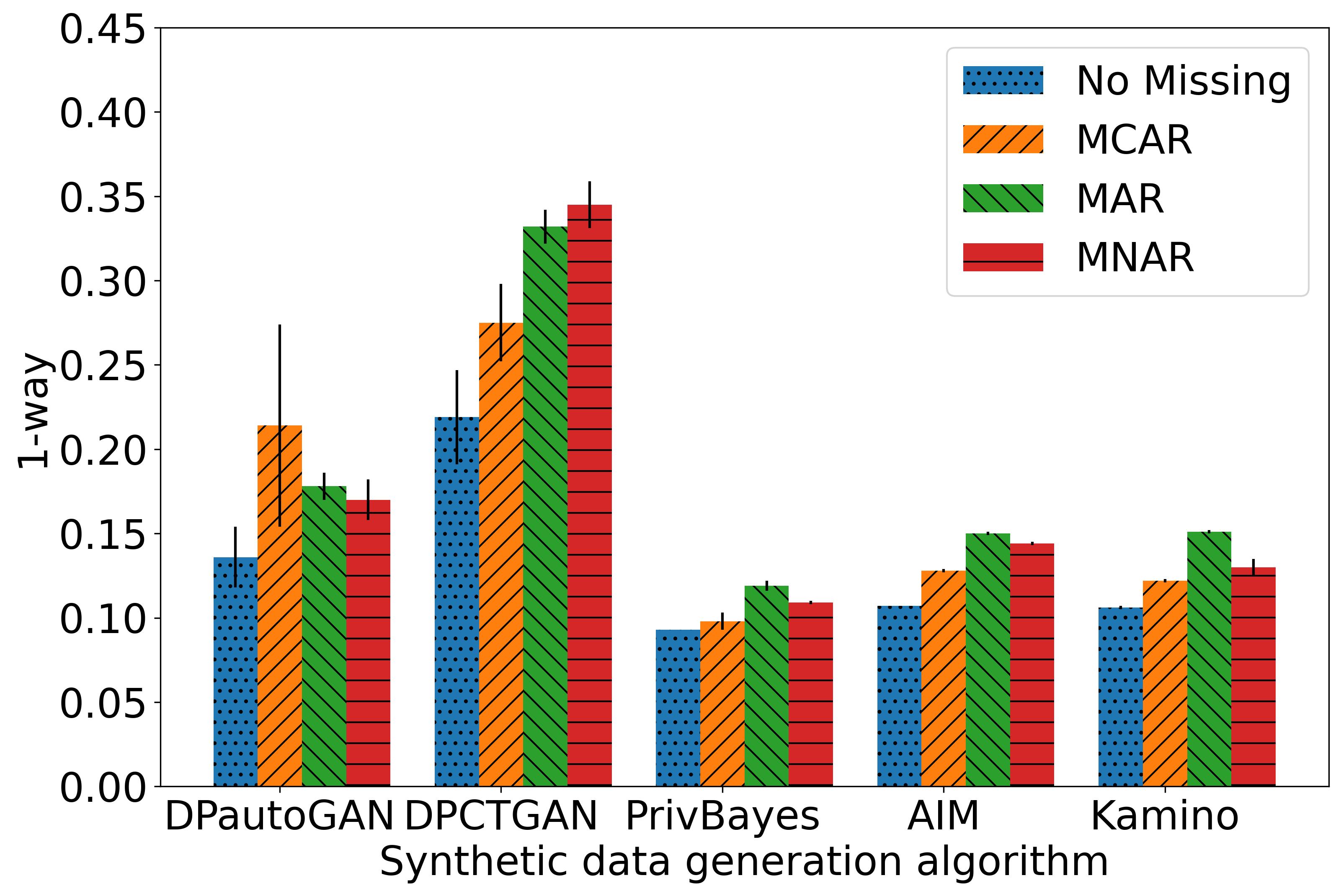}
    \caption{Complete row only approach results in poor results for MAR and MNAR missing mechanism.
    }
    \label{fig:compelte_row_bias}
\end{figure}

Figure~\ref{fig:compelte_row_bias} illustrates the performance of five different DP synthetic data generation algorithms using the complete row only approach for the Adult dataset with various missing mechanisms. The results indicate that the complete row only approach performs poorly for missing at random (MAR) and missing not at random (MNAR) mechanisms, which introduce bias to the estimated distribution of attributes. This bias can directly affect statistical approaches such as PrivBayes which rely heavily on empirical estimation of marginals. The 1-way distance between the generated synthetic dataset and the original dataset confirms this observation.

Besides the potential bias issue for the complete row only approach, the number of complete rows remaining can be very small. For synthetic data generation methods involving large deep learning models such as GAN, feeding the training process with a small number of complete rows will result in a poor data generation model. This is because the learning process does not converge or/and the noise added for achieving DP overshadows the signals of the training samples. For example, the ground truth Adult dataset which has 32k rows reduces to $\approx5$k complete rows with 20\% MAR and $\approx 1$k complete rows with 20\% MCAR/MNAR missing mechanism respectively. Our results show that the number of complete rows plays a vital role in the performance of the synthetic data generation algorithms. We discuss this in detail in Section~\ref{sec:exp} where we study several prior work approaches and evaluate them on the different missing data scenarios. 
\fi

%% file: impute.tex
\stitle{Imputation First Approach.} 
Imputation is vastly used in practice where the missing data are filled up with values inferred from the observed data.
\ifpaper
There are multiple ways to impute missing values in a dataset that include statistical~\cite{little2019statistical, rubin1976inference}, hot/cold deck methods~\cite{andridge2010review, nordholt1998imputation} and ML based imputation~\cite{jerez2010missing, lakshminarayan1996imputation}.
\else
There are multiple ways to impute missing values in the dataset, including: 
\begin{enumerate}[leftmargin=1.5em]
    \item Statistical methods: Each attribute of the dataset is modeled separately using statistical methods such as mean, median, and mode. The model is then used to fill up the missing values of the attribute~\cite{little2019statistical, rubin1976inference}. For example, in  Figure~\ref{fig:worst_impute}, for the left  tables, we use the median of the observed values of column A, to fill up the missing value of the 4th row. 
    
    \item Hot and cold deck methods: This imputation technique replaces every missing value with another value from the same dataset (Hot deck) or from a proxy dataset (Cold deck)~\cite{andridge2010review, nordholt1998imputation}. 
    
    The missing cells of the incomplete row are then filled up from the closest similar row. Similarity metric like cosine distance or $\ell_2$ distance can be used. 
    
    \item ML imputation: Machine learning (ML) based approaches are common for missing data imputation~\cite{jerez2010missing, lakshminarayan1996imputation}. An ML model is trained to predict the missing values of an attribute based on other non-missing attributes in the dataset as training features. 
\end{enumerate}
\fi

In our private dataset setting, imputations must also be conducted privately. \revision{A straightforward yet ineffective method involves randomly selecting values from the missing attribute's domain, as illustrated later  (Figure~\ref{fig:exp5}).} We skip analysis of the cold deck imputation as finding another similar dataset is impractical for private datasets. 

\revision{DP imputation can be approached in two manners. The first involves splitting the privacy budget, allocating a portion for imputation and the remainder for synthetic data generation. However, this approach is challenging due to budget allocation and choice of imputation algorithm. Additionally, some imputation techniques such as the hot deck imputation that are row specific (replicates the missing value in a row based on some other observed value of a different user) cannot be performed in the DP setting. Randomizing this row to achieve DP introduces too many errors to the dataset. }

The second way is to formulate imputation as a transformation of the dataset and calculate the associated privacy cost as an end-to-end algorithm. We use the notion of stability (Theorem~\ref{thm:stability}) to calculate the privacy costs of these transformations.

\revision{\begin{theorem}~\cite{DBLP:conf/sigmod/McSherry09}
We say a transformation $T(\cdot)$ is $c$-stable, if the distance between $T(D)$ and $T(D')$ is at most $c$ times the distance between $D$ and $D'$. The composite mechanism $\mathcal{M} \circ T$ then becomes $(c \cdot \epsilon, \delta)$-DP, for any mechanism $\mathcal{M}$ which is $(\epsilon, \delta)$-DP.
\label{thm:stability}
\end{theorem}}
 
\begin{lemma}\label{lemma:statimp}
Consider a transformation $T_A(\cdot)$ for imputing attribute $A$,
which takes in the incomplete dataset $D$ as part of the input  and outputs a dataset $D'$ with the complete values for attribute $A$. Then, the stability of \hspace{0.05cm}$T_A(\cdot)$ is
$c= m_A + 1$
where $m_A$ refers to the number of missing values for the attribute $A$.
\end{lemma}
\ifpaper
This lemma basically considers the transformation using the row in which the neighboring databases differ. If applying a sequence of imputation functions over the attributes of a dataset $(\cdots, T_{A_i}, \cdots \}$, the difference in the resulted datasets can be very large (up to the data size at a worst case) when the input dataset differs in a single row. Note that these results hold even if imputation functions for two attributes are not the same. 
\else
\proof
As the neighboring databases $D$ and $D'$ differ by a row, $T_A(\cdot)$ uses two different values $x$ and $x'$ to impute the missing values in $D$ and $D'$ respectively. As there are $m_{A(1)}$  rows in both $D$ and $D'$ that have missing values for attribute $A$, the resulted imputed databases, $T(D)$ and $T(D')$ have $m_{A(1)} +1$ number of rows (include the row that $D$ and $D'$ differ). 
\qed

Using the above Lemma we can see that applying a sequence of imputation functions over the attributes of a dataset $(\cdots, T_{A_i}, \cdots \}$, the difference in the resulted datasets can be very large when the input dataset differs in a single row. Note that these results hold even if imputation functions for two attributes are not the same. 
\fi

\begin{theorem}
The composite mechanism $\mathcal{M}\circ T$ on a dataset $D$ with $n$ rows is $n\epsilon$-DP, where $\mathcal{M}$ is a $\epsilon$-DP mechanism, and $T$ is a sequence of imputation functions performed to each attribute of $D$. 
\end{theorem}
\ifpaper
\else
\proof
Consider a worst-case scenario: $D$ and $D'$ differ in a single row that does not have any missing values, and the rest of the rows have only one attribute with missing values, i.e., $\sum_{i} m_{A_i} = n-1$. As $T$ uses the complete row to impute all the missing values, all rows will be affected and the overall cost of $M \circ T$ will be $n\epsilon$-DP. \qed
\fi 
\ifpaper
Our full paper~\cite{full_paper} consists of two worst case illustrations on a toy dataset for the above Theorem.

\else
\begin{example}
In Figure~\ref{fig:worst_impute}, we illustrate two worst-case toy examples for a dataset with two columns using mean and median imputations. The gray color indicates missing values and dotted lines denote differing rows between the top and the bottom datasets for each example. The missing values in both columns in the left and right examples are filled up with mean and median functions respectively. After applying imputations on the neighboring datasets $D$ and $D'$, the number of rows in column $B$ in both examples starts to differ by $4$ rows and $5$ respectively including the imputed values and the differing row. In such a scenario, one needs to pay $4\epsilon$ or $5\epsilon$ privacy cost to ensure DP to the incomplete data. 
\end{example}
\begin{figure}
    \centering
    \includegraphics[width=0.9\linewidth]{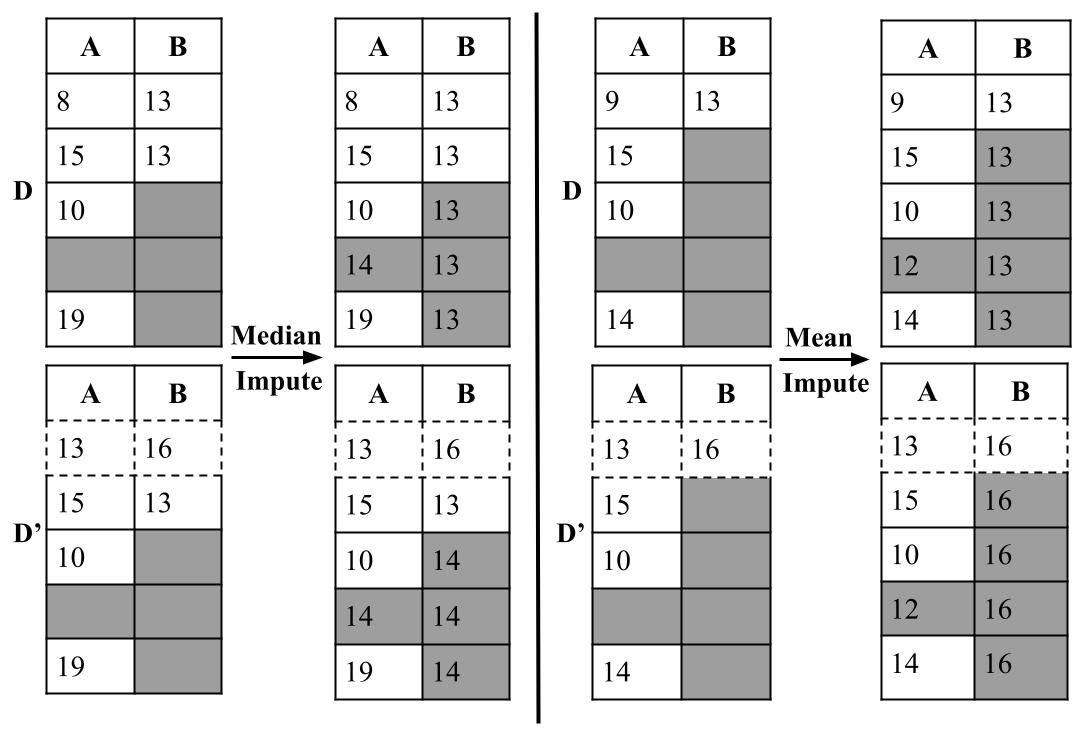}
    \caption{Illustration of worst case statistical imputations.}
    \label{fig:worst_impute}
\end{figure}
\fi

%% file: adaptiverecourse.tex
\subsection{Adaptive Recourse Approach}\label{sec:adaptive}
Both aforementioned approaches suffer from inefficiencies in data or privacy budget utilization. The complete row approach discards partial rows, wasting potential learning data, while the imputation-first method incurs high privacy costs. These challenges prompt us to modify existing synthetic data generation techniques to optimize both dataset information and privacy budget usage, which we term \emph{adaptive recourse}. \revision{The concept involves employing the privacy-preserving learning aspect of synthetic data generation for both imputation and synthetic data sampling. This offers two advantages: the privacy budget is dedicated solely to learning a single model, and the imputation process generates more comprehensive training examples, enhancing model utility. We have selected three representative DP data generation models—generative adversarial networks (GAN), partial marginal observation, and column-wise data generation—as proof-of-concept for these adaptive approaches, which can extend to other existing or new DP models.}

\input{misgan}

\begin{algorithm}[t]
    \begin{algorithmic}[1]
        \REQUIRE Incomplete dataset $D$, Attributes $\mathcal{A}$, Privacy budget $\epsilon_1, \epsilon_2$
        \STATE Initialize Bayesian network $B$ of degree $k$ and $V = \phi$
        \STATE Sample $X_1$ from $\mathcal{A}$ and add $(X_1, \phi)$ to $B$; add $X_1$ to $V$
        \FOR{$i = 2 \dots |\mathcal{A}|$}
            \STATE Initialize $\Omega = 0$
            \STATE For each $X \in \mathcal{A}\setminus V$ and each $\Psi \in {\binom{V}{k}}$; add $(X, \Psi)$ to $\Omega$
            \STATE Select a pair $(X_i, \Psi_i)$ from $\Omega$ with maximal mutual information in complete rows for attributes $X_i$ in D using exponential mechanism of budget $\epsilon_1/|\mathcal{A}|$
            \STATE Add $(X_i, \Psi_i)$ to $B$; add $X_i$ to $V$
        \ENDFOR
        \STATE Initialize synthetic dataset $D^*$
        \FOR{$i = 1 \dots |\mathcal{A}|$}
            \STATE Compute distribution from non-missing values $\Pr[X_i, \Psi_i]$ from the complete rows of $X_i$ in $D$ 
            \STATE Learn $\Pr^*[X_i, \Psi_i]$ with Laplace mechanism at budget $\epsilon_2$
            \STATE Set negative values to $0$ and normalize
            \STATE Sample from $\Pr^*[X_i, \Psi_i]$ and add to $D^*$
        \ENDFOR
        \STATE Return $D^*$
        \caption{PrivBayes Enhanced (PrivBayesE)} 
        \label{algo:privbayese}
    \end{algorithmic}
\end{algorithm}

\stitle{Partial marginal observation-based adaptive recourse.}
This approach can be applied to algorithms that use low dimensional marginal queries~\revision{\cite{DBLP:conf/sigmod/ZhangCPSX14, DBLP:conf/icml/McKennaSM19, DBLP:conf/sigmod/QardajiYL14, DBLP:journals/tkde/XiaoWG11, mckenna2022aim, DBLP:books/daglib/0023091, DBLP:conf/ssdbm/PingSH17, DBLP:journals/pvldb/LiXZJ14}}. Instead of discarding all the partially missing rows, only the rows with missing cells in the queried attributes are removed. Such a strategy is most helpful when only a subset of attributes have missing data. For example, with MAR missing mechanism, partial marginal observation can be learned from all the non-missing columns.

\revision{\emph{Algorithm overview:}
We extend PrivBayes~\cite{DBLP:conf/sigmod/ZhangCPSX14} using this strategy and call it PrivBayes enhanced or PrivBayesE in short (Algorithm~\ref{algo:privbayese}). Both PrivBayes and PrivBayesE learn Bayesian networks of degree $k$ to know the correlated columns. The network $B$ is initialized by adding the first attribute in the attribute list $\mathcal{A}$. The vertices that been discovered so far are stored in the list $V$ (Lines 1-2). Next the algorithm loops over each attribute in $\mathcal{A}$ (Line 3) and generates $|V|$ choose $k$ sets appended with every attribute seen so far $\mathcal{A}\setminus V$ and stores in a list $\Omega$ (Lines 4-5). The mutual information values for each pair in $\Omega$ is computed and the best one chosen using exponential mechanism using privacy budget $\epsilon_1/\mathcal{A}$ is added to the network $B$ (Lines 6-7). PrivBayes/PrivBayesE then generate a synthetic dataset $D^*$ using this network $B$ in sequence (Line 10). Marginals are computed for each attribute $X_i$ with its most correlated attributes $\Psi_i$ using privacy budget $\epsilon/2$ and added to $D^*$ (Line 12-14).
}

\revision{
\emph{Highlights:}
In PrivBayesE, modifications are made to the dataset generation process. Each time a marginal query is made, PrivBayesE learns from all non-missing information of the attribute(s) (Line 11). This improves upon the complete row approach by discarding missing rows on a smaller set of attributes, rather than from the entire dataset. This is particularly advantageous for scenarios like missing completely at random (MCAR), where analyzing more data aids in better estimating the true distribution of marginals, and missing at random (MAR), where some marginals are completely available, allowing estimation based on complete data. The privacy analysis of PrivBayesE mirrors that of PrivBayes, as PrivBayesE does not introduce additional queries to the dataset.
}

\revision{\emph{Flexibility:} This enhancement can be applied as a wrapper to any method that makes use of low-dimensional marginals to generate synthetic data. In our paper, we choose PrivBayes as our baseline as it is the most simple and fundamental marginal based approach. 
\ifpaper
 In the full paper~\cite{full_paper}, we also include an additional experiment to enhance another marginal based approach AIM~\cite{mckenna2022aim}. 
\else
\fi
}

\stitle{Column-wise data generation-based adaptive recourse}. 
This approach can be applied to any algorithm that uses column-wise intermediate models to learn the data distribution. In such algorithms, a sequence of attributes is decided, and starting with the second attribute in sequence, a model is learned to predict the current attribute using previously learned ones.

\revision{
\emph{Algorithm overview:}} We extend Kamino~\cite{ge2021kamino} using this strategy and call in Kamino impute or KaminoI in short (Algorithm ~\ref{algo:kaminoi}). \revision{ Kamino/KaminoI starts with deciding a sequence of attributes based on a given denial constraints $\Psi$ (Line 1). The distribution of the first attribute in the computed sequence is learnt using all the non-missing cells (Line 3). This computed distribution is noised (Line 4) and values are sampled to populate the synthetic dataset (Line 5). For each new attribute $Y$ in the sequence, Kamino/KaminoI learns a private intermediate model which uses all previously visited attributes $X$ to predict the new attribute $Y$ (Lines 7-8). This intermediate model is used to generate the values for the attribute $Y$ in the synthetic data given the sampled values for $X$ (Line 10).
}

\revision{
\emph{Highlights:}
 KaminoI includes an additional imputation step to impute values of attribute $Y$ (Line 9), utilizing the same intermediate model that was trained to predict attribute $Y$ given attribute set $X$. This enhancement ensures that no missing values are discarded; instead, they are used to train intermediate models. It is worth noting that the sequence $S$ significantly influences KaminoI's imputation process. In Kamino, the sequence is generated considering input constraints $\Psi$. However, if an attribute with many missing values is positioned early in sequence $S$, its imputation may be less effective. To optimize imputation, attributes not in constraints $\Psi$ can be ordered based on decreasing percentage of missing values. If available, clues from the missing mechanism can also help determine the sequence. For instance, with the MAR mechanism predicting missing IQ based on age, the age attribute can precede the age column in $S$. To ensure fair comparison, the same sequence as Kamino is used for KaminoI in the experimental section.
}

\begin{algorithm}[t]
	\begin{algorithmic}[1]
        \REQUIRE Incomplete dataset $D$, Attributes $\mathcal{A}$, Constraints $\Psi$, Privacy budget $\epsilon_1, \epsilon_2$
        \STATE Build sequence $S$ of attributes $\mathcal{A}$ using constraints $\Psi$
        \STATE Initialize synthetic dataset $D^*$
        \STATE Compute distribution of first attribute $H = \Pr[S_1]$ using all non-missing values
        \STATE Generate DP $H^*$ by adding Gaussian noise of budget $\epsilon_1$
        \STATE Sample from $H^*$ to populate $D^*[S_1]$
        \FOR{$i=2 \dots |\mathcal{A}|$}
            \STATE Load training features $X = S_{:j}$, and target label $Y = S_j$
            \STATE Train model $M = \theta(X,Y)$ privately with budget $\frac{\epsilon_2}{|\mathcal{A}|-1}$ 
            \STATE Impute missing values in dataset $D[S_j]$ using $M$
            \STATE Predict synthetic values $\Omega = M(D^*[S_{:j}])$ and fill $D^*[S_j] = \Omega$
        \ENDFOR
        \STATE Return $D^*$
        \caption{Kamino Impute (KaminoI)} 
        \label{algo:kaminoi}
    \end{algorithmic}
\end{algorithm}

\revision{
\emph{Flexibility:} 
This enhancement is applicable to any method that iterates over the columns of the dataset. Such an algorithmic architecture allows for value imputation as learning progresses. Each time a model predicts the next attribute, it can also impute missing values using the same model, incurring no additional privacy costs. This strategy is particularly effective when missing data correlates with other attributes in the dataset, such as in the missing at random (MAR) scenario. While there's only one known private approach employing this strategy \cite{ge2021kamino}, several non-private approaches exist \cite{xu2018synthesizing, riverssynthetic}. Although PrivBayesE can use learned distributions to impute missing values for all visited attributes (Line 12 in Algoirthm~\ref{algo:privbayese}), it only utilizes low-way marginals compared to KaminoI's larger models, which benefit more from imputation. Hence, we did not include this imputation step in PrivBayesE.}

%% file: misgan.tex

\stitle{GAN-based adaptive recourse.} 
In non-private literature, several approaches use the GAN framework to deal with missing data~\cite{yoon2018gain, li2019misgan, luo2019e2gan, xu2020scigans}. \revision{We choose to privatize an approach called MisGAN~\cite{li2019misgan}, which allows us to simultaneously demonstrate the power of learning the data distribution and the missingness pattern for GAN-based algorithms. We call its DP version, DP-MisGAN, as shown in Algorithm~\ref{algo:misgan}. We first describe the high-level architecture of MisGAN and DP-MisGAN, and then we highlight the enhancement in this approach.}

\revision{
\emph{Algorithm Overview:}  MisGAN/DP-MisGAN trains two generator-discriminator pairs --- one for learning the data distribution and the other for learning the missingness pattern. The training spans $E$ epochs, with each epoch sampling $|D|/B$ sized subsets from the dataset $D$ without replacement (Line 4). 
Each subset $S_t$ is processed with real data $x_{data}$ (Line 6) and its corresponding missing mask $x_{mask}$(Line 7). The missing mask $x_{mask}$, computed from the missing indicator matrix $M$, marks where data is missing as $1$ and $0$ otherwise. Missing values in real data are replaced with $0$s (Line 8). Two fake examples $y_{data}$ and $y_{mask}$ are generated by passing random Gaussian noise through data and mask generators (Line 9). These generators are updated using gradient descent, with discriminators learning true distributions in one phase (Line 10) and generators updating within specified $T_G$ intervals in the second phase (Line 11). In each generator interval, two fake samples $y_{data}$ and $y_{mask}$ are again generated similarly from the two generators (Line 12) and gradients from the discriminator is computed (Line 13). Note that the non-private MisGAN uses these gradients directly to update the parameters of the two generators (Line 16) and releases both generators in the end. Finally, after the training is completed for $E$ epochs, the discriminators are thrown away, and the privately learned data generator is used to sample synthetic data (Line 20)
}.

\begin{algorithm}[t]
	\begin{algorithmic}[1]
        \REQUIRE Incomplete dataset $D$, noise scale $\sigma$, epochs $E$, learning rates $\eta_D$ and $\eta_G$, generator interval $T_G$, batch size $B$, missing indicator matrix $M$
        \STATE Initialize data generator $\theta^D_G$ and discriminator $\theta^D_D$
        \STATE Initialize mask generator $\theta^M_G$ and discriminator $\theta^M_D$
        \FOR{$i$ in [$1,\dots, E$]}
            \STATE Subsample 
            dataset $D$
                    into $\{S_k\}^{k=(|D|/B)}_{k=1}$ subsets
            \FOR{$t$ in [$1,\dots, |D|/B$]}
    		\STATE Set real data $x_{data} = S_t$
                \STATE Sample real mask $x_{mask}$ from missing indicator $M(S_t)$
                \STATE Fill missing values in $x_{data}$ with $0$
                \STATE Generate fake data and mask \\ 
                $y_{data} = \theta^D_G(z)$ where $z \sim \mathcal{N}(0,1)$ \\
                $y_{mask} = \theta^M_G(z)$ where $z \sim \mathcal{N}(0,1)$
                \STATE Update $\theta^D_D = \theta^D_D - \eta_D \nabla_{\theta^D_D}[\frac{1}{B} (\theta^D_D; x_{data}) - \frac{1}{B} (\theta^D_D; y_{data})]$ \\
                $\theta^M_D = \theta^M_D - \eta_D \nabla_{\theta^M_D}[\frac{1}{B} (\theta^M_D; x_{mask}) - \frac{1}{B} (\theta^M_D; y_{mask})]$
                \IF{$t$ in interval of $T_G$}
                    \STATE Generate fake data and mask \\ 
                    $y_{data} = \theta^D_G(z)$ where $z \sim \mathcal{N}(0,1)$ \\
                    $y_{mask} = \theta^M_G(z)$ where $z \sim \mathcal{N}(0,1)$
                    \STATE Compute data and mask gradient \\
                    $g_d =  \nabla_{\theta^D_G}(\theta^D_D; y_{data})$,
                    $g_m =  \nabla_{\theta^M_G}(\theta^M_D; y_{mask})$
                    \STATE Clip each gradient $\bar{g_i}_d = {g_i}_d / \max(1, \|{g_i}_d\|_2)$ \\
                    $\bar{g_i}_m = {g_i}_m / \max(1, \|{g_m}\|_2)$
                    \STATE Compute noisy gradients $\tilde{g_i}_d = \bar{g_i}_d + 2\sigma\mathcal{N}(0, 1)$ \\ 
                     $\tilde{g_i}_m = \bar{g_i}_m + 2\sigma\mathcal{N}(0, 1)$
                    \STATE Update generators $\theta^D_G = \theta^D_G - \eta_D \frac{1}{B} \tilde{g_i}_d$ \\
                    $\theta^M_G = \theta^M_G - \eta_D \frac{1}{B} \tilde{g_i}_d$
                \ENDIF
            \ENDFOR
        \ENDFOR
        \STATE Generate and return synthetic dataset from generator $\theta^D_G$
        \caption{DP-MISGAN} 
        \label{algo:misgan}
    \end{algorithmic}
\end{algorithm}

\revision{\emph{Highlights:} In our DP-MisGAN, we clip and then noise the gradients learned for the generators using the \emph{sampled Gaussian mechanism} (SGM)~\cite{DBLP:journals/corr/abs-1908-10530} to ensure privacy (Lines 14-15). This step is different from prior efforts for DPCTGAN~\cite{fang2022dp} and DPautoGAN~\cite{tantipongpipat2021differentially}, which privatize non-private optimizer of the discriminators using DPSGD~\cite{WM10, BST14, song2013stochastic, DBLP:conf/ccs/AbadiCGMMT016}. This divergence in technique leads to notable difference in the algorithm's utility: while the discriminators become noisy due to the introduction of noise, the generator's gradient calculation, which relies on discriminator weights (Line 13), is also affected.} 
\revision{We observe that we only need to publish the generators, not the discriminators. Hence, we resort to the gradient sanitization (GS) approach~\cite{chen2020gs} to perturb only the gradients of the generators (Lines 14-15), without affecting the utility of the discriminator.}
\revision{In addition, the GS approach also allows us to skip the hyperparameter tuning for the gradient clip $C$, which can be vastly detrimental if set wrong~\cite{mohapatra2021role}.  Unlike the standard gradient clipping to bound the sensitivity of the gradient norm by $C$, i.e., $g/max(1, \|g\|_2/C)$, we clip the gradient by $g/max(1,\|g\|_2)$ (Line 15) by considering a Wasserstein-Gan (WGAN)~\cite{arjovsky2017wasserstein} framework with an additional gradient penalty term in the loss function of the discriminator that enforces the $\ell_2$-norm of the discriminator gradients to be naturally close to $1$.
}
\ifpaper
We show a tight bound for the privacy loss of the training procedure calculated using R\'enyi-DP (RDP) in the full paper~\cite{full_paper}. 
\else
A tight bound for the privacy loss of the training procedure can be calculated using R\'enyi-DP (RDP).

\begin{definition}[R\'enyi-DP~\cite{mironov2017renyi}]
    A randomized algorithm $M$ with domain $\mathcal{D}$ is $(\alpha, \epsilon)$-RDP at order $\alpha>1$, for any pair of neighbouring databases $D,D' \in \mathcal{D}$ that differ in one tuple. Let $P_D$ and $P_{D'}$ be the output probability density of $\mathcal{M}(D)$ and $\mathcal{M}(D')$ respectively. Then, it holds that: $\frac{1}{\alpha - 1} \log\E_{x\sim \mathcal{M}(D')} \left( \frac{P_{D}(x)}{P_{D^\prime}(x)} \right)^\alpha \leq \epsilon$.
\end{definition}

The post-processing and composability properties of DP also apply to RDP.
Specifically, if a sequence of adaptive mechanisms $\mathcal{M}_1$, $\mathcal{M}_2$, $\cdots$, $\mathcal{M}_k$ satisfy ($\alpha,\epsilon_1$)-, ($\alpha,\epsilon_2$)-, $\cdots$, ($\alpha,\epsilon_k$)-RDP, 
then the composite privacy loss is ($\alpha, \sum_{i=1}^k\epsilon_i$)-RDP. The RDP privacy loss of the \emph{sampled Gaussian mechanism} (SGM)~\cite{DBLP:journals/corr/abs-1908-10530} is given using Lemma~\ref{lemma:sgm}.

\begin{lemma}\label{lemma:sgm}
Given a database $D$ and query $f: \mathcal{D} \rightarrow \mathbb{R}^d$ with sensitivity $S_f$, returning $f(\{x\in D \mid \text{x is sampled with probability r}\})+\mathcal{N}(0,S_f^2\sigma^2 \mathbb{I}^d)$ results in the following RDP cost for an integer moment $\alpha$.
\begin{equation*}
    R_{\sigma, r}(\alpha) = 
    \begin{cases}
     \frac{\alpha}{2(S_f\sigma)^2} &  r = 1\\ 
     \sum_{k=0}^{\alpha} \binom{\alpha}{k}(1-r)^{\alpha-k} r^{k}\exp(\frac{\alpha^2 - \alpha}{2(S_f\sigma)^2}) &  0<r<1
    \end{cases}
\end{equation*}
\end{lemma}
We can now calculate the total RDP cost of DP-MisGAN by composing the cost of each SGM in the training procedure.

\begin{theorem}\label{thm:dpmisgan_rdp}
    The total RDP cost of DP-MisGAN is 
    \begin{equation*}
        R(\alpha) = 2\lceil\frac{T}{T_G}\rceil \sum_{k=0}^{\alpha} \binom{\alpha}{k}(1-\frac{B}{|D|})^{\alpha-k} \frac{B}{|D|}^{k}\exp(\frac{\alpha^2 - \alpha}{8\sigma^2})
    \end{equation*}
\end{theorem}
\begin{proof}
Let the gradient clipping procedure of the discriminator in DP-MisGAN be $f = g/\max(1, \|g\|_2/C)$. The sensitivity $S_f$ can be thus be derived by the reversed triangle inequality. 
\begin{equation*}
    S_f = \max_{D, D'} \| f(D) - f(D') \| \leq 2C \leq 2 
\end{equation*}
The last inequality follows as we set $C = 1$ in DP-MisGAN. Each SGM in DP-MisGAN (Algorithm~\ref{algo:misgan}) has sampling probability of $\frac{B}{|D|}$ and thus the RDP cost of each SGM can be derived using Lemma~\ref{lemma:sgm} as $\sum_{k=0}^{\alpha} \binom{\alpha}{k}(1-\frac{B}{|D|})^{\alpha-k} \frac{B}{|D|}^{k}\exp(\frac{\alpha^2 - \alpha}{8\sigma^2})$.
DP-MisGAN consists of two generators that are updated every generator interval $T_G$ each. Therefore, using the composition property of RDP, we can compose the costs for both the discriminators in DP-MisGAN as $R(\alpha) = 2\lceil\frac{T}{T_G}\rceil \sum_{k=0}^{\alpha} \binom{\alpha}{k}(1-\frac{B}{|D|})^{\alpha-k} \frac{B}{|D|}^{k}\exp(\frac{\alpha^2 - \alpha}{8\sigma^2})$.
\end{proof}
By the tail bound property of RDP~\cite{mironov2017renyi}, we can convert the RDP cost of DP-MisGAN to $(\epsilon,\delta)$-DP, where $\epsilon$ is computed by 
\begin{equation}
\epsilon(\delta) = \min_{\alpha} R(\alpha) + \frac{\log(1/\delta)}{\alpha-1},    
\end{equation}
for a given $\delta$. In practice, given the parameters $\epsilon, \delta, B, T, T_G$, the $\sigma$ is calculated using Theorem~\ref{thm:dpmisgan_rdp} and the order $\alpha$ is usually searched within a range of user input values~\cite{waites2019pyvacy}.

\fi

We expect DP-MisGAN to perform better than the naive GAN approaches because it learns from both the complete rows as well as the incomplete rows of the dataset. Furthermore, as DP-MisGAN learns the missing data pattern of the incomplete dataset, we anticipate that it will capture more information in complex MAR and MNAR missing mechanisms.

\revision{
\emph{Flexibility:} The framework of DP-MisGAN can be used for any GAN method. The core idea of the change is to encompass two discriminator-generator pairs for learning missing data and synthetic data generation. However, to achieve better privacy accounting, it is important to use the GS approach with a WGAN structure and discard the discriminator after training. 
}

%% file: amplification.tex
\section{Privacy for Groundtruth Data}\label{sec:amplification}

In this section, we shift our focus to exploring the privacy implications for the ground truth data, which we approach as a distinct problem that closely relates to Problem~\ref{prob1}. \revision{We first find that the solutions for  Problem~\ref{prob1} do not always offer sufficient privacy for Problem~\ref{prob2}, except when the probability of missing values in a row is independent of the other rows in the dataset (Section~\ref{sec:relateprob1}).}
Furthermore, we demonstrate that certain missing mechanisms, such as MCAR, allow a tighter privacy analysis \revision{for Problem~\ref{prob2} when applying the same solution from Problem~\ref{prob1} in Section~\ref{sec:amplificationmcar}}.

\subsection{Relationship to Problem~\ref{prob1}}\label{sec:relateprob1}

We have proposed multiple synthetic data generation algorithms $M$ which train on the incomplete dataset $D$ and achieve $(\epsilon, \delta)$-DP as solutions to problem~\ref{prob1}. However, this incomplete dataset $D$ results from a missing mechanism $M_\phi$ on the ground truth dataset $\bar{D}$. In problem~\ref{prob2}, we study the same mechanisms $M$ which train on \revision{$D$} but discuss their privacy impact on $\bar{D}$. We do so by combining the missing mechanism $M_\phi$ and the synthetic data generation process $M$ as a composite mechanism $M \circ M_\phi$.

It is important to note that just because $M$ is a DP mechanism for incomplete data $D$, it does not necessarily mean that $M\circ M_\phi$ is DP for the ground truth data $\bar{D}$. 
\begin{figure}
        \centering
        \includegraphics[width=\linewidth]{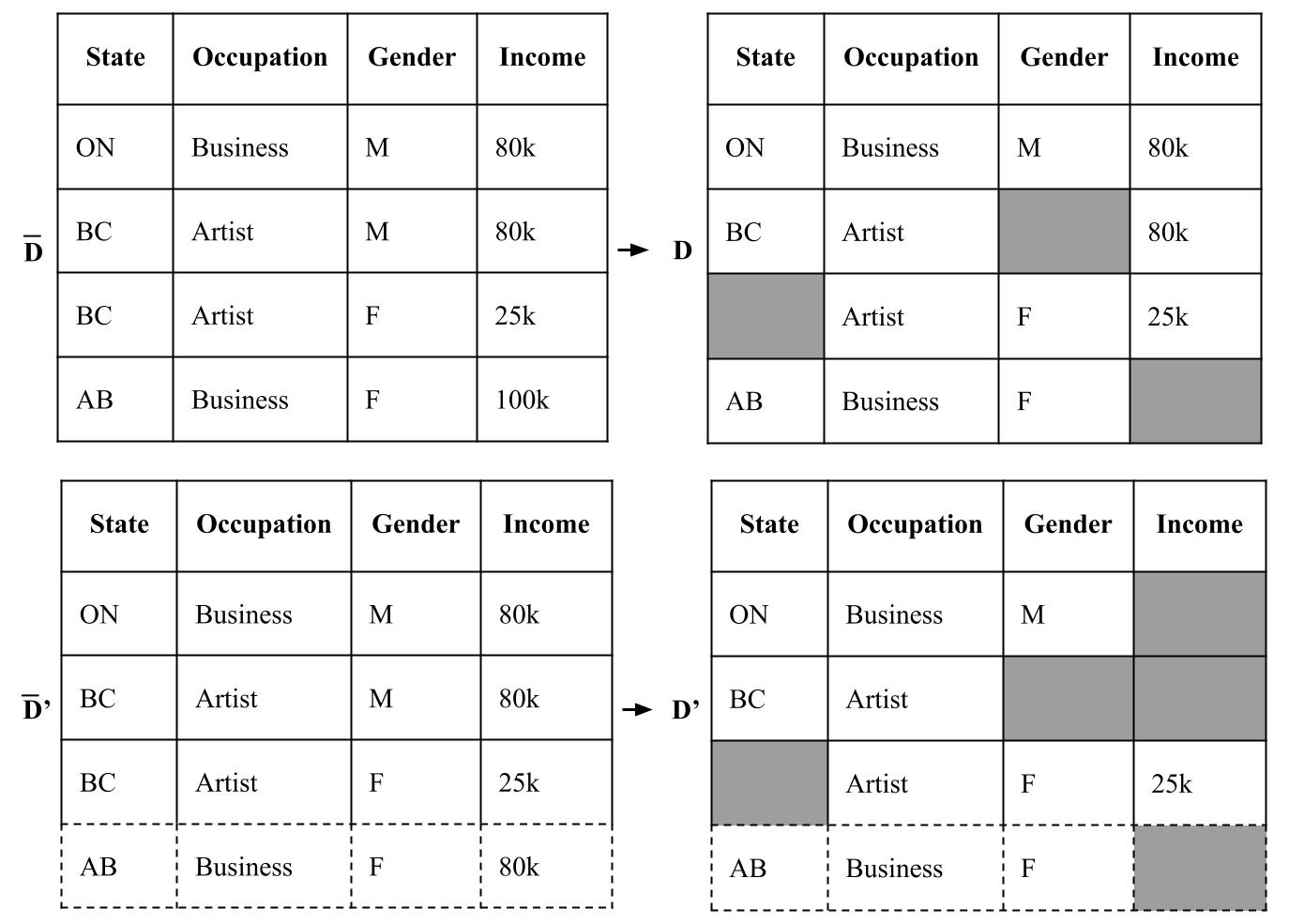}
        \caption{Example of private incomplete datasets generated from two neighbouring ground truth datasets. Gray denotes missing cells and dotted lines represent the differing row.}
        \label{fig:mnarexample}
    \end{figure}

\begin{example}

In Figure~\ref{fig:mnarexample}, consider neighboring ground truth data $\bar{D}$ and $\bar{D'}$ differ in the last row's income value (100k v.s. 80k). Their income columns have an MNAR missing mechanism that hides the highest income value and their corresponding incomplete data $D$ and $D'$ then differ more than one row. This means that an $\epsilon$-DP mechanism for incomplete data cannot guarantee the same level of privacy for the ground truth data.
\end{example}

The example above does not provide a strong privacy guarantee for the ground truth data because the probability of a row having missing values depends on the values of other rows. However, we can show that if $M_\phi$ enforces independent probabilities for each row to have missing values, a strong privacy guarantee applies to the ground truth data.

\begin{theorem}\label{thm:independent_randomness}
Let the missing mechanism $M_{\Phi}$ have independent randomness to hide the values of each row and $D=M_\phi(\bar{D})$. If $M$ achieves  $(\epsilon,\delta)$-DP for $D$, then $M\circ M_{\Phi}$ satisfies ($\bar{\epsilon},\bar{\delta}$)-DP for $\bar{D}$, where $\bar{\epsilon}\leq \epsilon, \bar{\delta}\leq \delta$.
\end{theorem}

\ifpaper

\else
\begin{proof} (sketch) 
As $M_{\Phi}$ has independent randomness to hide values of each row, given the ground truth data $\bar{D}$ and a possible incomplete dataset $D$, 
we have $\Pr[D|\bar{D}] =\prod_l \Pr[D_l|\bar{D}]$, where $D_l$ refers to the value taken by the $l$th row.
Consider neighboring groundtruth datasets $\bar{D}$ and $\bar{D}'$ differ in the $i$th row and any possible output $O$ of $M\circ M_{\Phi}$. Let $\mathcal{D}$ be all possible incomplete datasets that can be outputted by $M_{\Phi}$ from $\bar{D}$ or $\bar{D}'$. We partition $\mathcal{D}$ into $\{\cdots, \mathcal{D}_j, \cdots\}$ such that all datasets with the same row values except the $i$th row are in the same group $\mathcal{D}_j$. Hence, for all $D\in \mathcal{D}_j$, they have the same probability for $\prod_{l\neq i}\Pr_{M_{\Phi}}[D_l|\bar{D}]$.
Now we have

\begin{eqnarray*}
&&    \Pr[O|\bar{D}] \\
&=& \sum_{\mathcal{D}_j}\sum_{D\in \mathcal{D}_j}\Pr_M[O|D]\Pr_{M_{\Phi}}[D|\bar{D}] \\
&=& \sum_{\mathcal{D}_j}\sum_{D\in \mathcal{D}_j}(\Pr_M[O|D] \Pr_{M_{\Phi}}[D_i|\bar{D}] \cdot \prod_{l\neq i}\Pr_{M_{\Phi}}[D_l|\bar{D}]) \\
&=&   \sum_{\mathcal{D}_j} (\prod_{l\neq i}\Pr_{M_{\Phi}}[D_l|\bar{D}] \cdot \sum_{D\in \mathcal{D}_j}\Pr_M[O|D] \Pr_{M_{\Phi}}[D_i|\bar{D}] ) \\
&\leq& \sum_{\mathcal{D}_j} (\prod_{l\neq i}\Pr_{M_{\Phi}}[D_l|\bar{D}'] \cdot \sum_{D'\in \mathcal{D}_j} (e^{\epsilon}\Pr_M[O|D']+\delta) \Pr_{M_{\Phi}}[D_i'|\bar{D}'] ) \\
 &=& e^{\epsilon}\sum_{\mathcal{D}_j}\sum_{D'\in \mathcal{D}_j}\Pr_M[O|D']\Pr_{M_{\Phi}}[D'|\bar{D}] +\delta \sum_{\mathcal{D}_j}\sum_{D'\in \mathcal{D}_j}\Pr_{M_{\Phi}}[D'|\bar{D}'] \\
  &=& e^{\epsilon}\sum_{\mathcal{D}_j}\sum_{D'\in \mathcal{D}_j}\Pr_M[O|D']\Pr_{M_{\Phi}}[D'|\bar{D}] +\delta
\end{eqnarray*}
The inequality above is based on for any neighbors $D$ and $D'$, we have $\Pr_M[O|D] \leq e^{\epsilon}\Pr_M[O|D']+\delta$ and $\sum_{D\in \mathcal{D}_j}\Pr_{M_{\Phi}}[D_i|\bar{D}] =\sum_{D'\in \mathcal{D}_j}\Pr_{M_{\Phi}}[D'_i|\bar{D}'] =1$.
\end{proof}
\fi

Theorem~\ref{thm:independent_randomness} says that the privacy bound for the ground truth dataset is lesser than equal to the bound of the incomplete dataset for a synthetic data generation algorithm if each row in the ground truth dataset has an independent probability of having missing values. 
Next we illustrate how to obtain a tighter privacy bound for the missing completely at random (MCAR) mechanism.

\subsection{Privacy Amplification Due To MCAR} \label{sec:amplificationmcar}

Missing completely at random (MCAR) enforces an independent probability of having missing rows for each attribute in the dataset. We use these probabilities to tighten the privacy bounds for ground truth data when the missing mechanism is MCAR. The technique we developed is inspired by the seminal work of privacy amplification due to sampling~\cite{balle2018privacy}. The premise of privacy amplification by subsampling is that we run a DP algorithm on some random subset of the data (e.g., sampled Gaussian mechanism, DP-SGD). The subset introduces additional uncertainty, which benefits privacy. 
Privacy amplification due to subsampling has been shown to work for many sampling methods (e.g., Poisson sampling, sampling with/without replacement) and for neighboring datasets which may differ with replacement or substitution. Privacy amplification by subsampling theorem~\ref{thm:samplingamplification} makes this intuition precise.

\begin{theorem}[Sampling Amplification Theorem~\cite{balle2018privacy, steinke2022composition}] \label{thm:samplingamplification}
    Consider an algorithm $M: \mathcal{D} \rightarrow \mathcal{D}$ that satisfies $(\epsilon, \delta)$-DP and a sampling mechanism $S(D)$ that samples a random subset $U$ from dataset $D$ of $n$ samples. If $p = max_{i \in [n]} \Pr_U [i \in U]$, then the composite mechanism $M(S(D))$ offers $(\epsilon', \delta')$-DP where $\epsilon' = \log(1+ p(e^\epsilon -1))$, $\delta' = p\delta$. For small values of $\epsilon$, we have $\epsilon' = \log(1+ p(e^\epsilon -1)) \approx p\epsilon$.
\end{theorem}

In our missing data context, we note that for synthetic data generation algorithms that train on incomplete data, many rows are naturally discarded due to the presence of missing cells. We exploit this natural throwing out of rows as a sampling mechanism and show that it can be used to amplify privacy.  Recall from Section~\ref{sec:preliminaries:missing} that MCAR enforces independent probability of having missing cells in the dataset for each attribute $\phi_1, \cdots, \phi_k$. We use these probabilities to propose our amplification results in Proposition~\ref{prop:mcar_amplification}.

\begin{proposition}\label{prop:mcar_amplification}
Consider an MCAR mechanism $M_{\Phi}:\mathcal{D} \rightarrow \mathcal{D}$ with missing probabilities $\{\phi_{1}, \ldots, \phi_{k}\}$ over attributes $\{A_1,\ldots A_k\}$ of the input ground truth data $\bar{D}$ and outputs an incomplete dataset $D$. If an algorithm $M: \mathcal{D} \rightarrow \mathcal{D}$ takes in rows in $D$ which have no missing values on attributes $\mathcal{A}_{M}\subseteq \{A_1,\ldots,A_k\}$, then 
$M\circ M_{\Phi}$ offers $(p\epsilon,p\delta)$-DP to the ground truth data $\bar{D}$ where
$    p_{\mathcal{A}_M} = \prod_{A_i\in \mathcal{A}_M}(1-\phi_{i})$. 
We call $\mathcal{A}_{M}$ an amplification attribute set for $M$ and $p_{\mathcal{A}_M}$ the amplification factor of $\mathcal{A}_{M}$.
\end{proposition}
\proof
A row in MCAR has $\prod_{i=1}^{i=l}(1-\phi_{i})$ probability of having no missing values and plugging in to Theorem~\ref{thm:samplingamplification}. 
\qed

We note three important facts. First, if an algorithm $M$ takes in rows with no missing values over an attribute set $\mathcal{A}_M$, then $M$ also takes in rows with no missing values over an attribute set $\mathcal{A}'_M\subset \mathcal{A}_M$. In other words, if $\mathcal{A}_M$ is an amplification attribute set for $M$, then any subset of $\mathcal{A}_M$ is an amplification attribute set for $M$ with amplification factor greater than that of $\mathcal{A}_M$. Second, when $\mathcal{A}_M=\emptyset$, $p_{\mathcal{A}_M}=1$. Third and more importantly, as the dataset is read only once, each attribute can only be used once as an amplification factor.
We can now use Proposition~\ref{prop:mcar_amplification} and Theorem~\ref{thm:samplingamplification} in conjunction to show the privacy amplifications for the different algorithms we have discussed so far in our paper. 

\stitle{Use case 1: Privacy amplification for complete row only approach.}
Here we show how to apply Proposition~\ref{prop:mcar_amplification} to all complete row only approaches (PrivBayes, Kamino, and GAN-based approaches). As these approaches take as input all attributes, the probability of seeing a row without missing values is $\prod_{i=1}^{i=k}(1-\phi_i)$. The following example illustrates how this probability can be used 

to obtain a tighter privacy bound for the ground truth data. 
\begin{example}[MCAR amplification for complete row only approach]\label{example:usecase1}
Consider the incomplete dataset from Figure~\ref{fig:mnarexample}. Let's assume that the missing data comes from an MCAR mechanism where the missing probabilities are $\phi_{State} = \frac{1}{4}, \phi_{Occupation} = 0, \phi_{Gender} = \frac{1}{4}, \phi_{Income} = \frac{1}{4}$ . Given $4$ DP sub-algorithms $M_1, M_2, M_3, M_4$ that each offer DP guarantee to the incomplete dataset $D$ at budget $\frac{\epsilon}{4}$. $M_1$ computes the marginals of the complete rows over attribute <State>, $M_2$ over <Occupation>, $M_3$ over <Gender> and $M_4$ over <Gender, Income>. As all sub-algorithms take as input only the complete rows, using Proposition~\ref{prop:mcar_amplification}, the amplification is $\prod_i(1-\phi_i) = 0.421$ and using Theorem~\ref{thm:amplificationcomposition} the final privacy is $\bar{\epsilon} = 4 * 0.421\frac{\epsilon}{4} = 0.421\epsilon$. 
\end{example}

\stitle{Use case 2: Privacy amplification for partial marginal observation approach.}
For partial marginal observation methods (e.g. PrivBayesE), calculating the amplification privacy cost is more complex. These methods involve multiple low-dimensional marginals with overlapping attributes. To determine the overall amplification for such algorithms, it is necessary to calculate the amplification for each marginal and carefully compose them. The complexity of this calculation arises from the optimal selection of the amplification attribute set for each marginal, which maximizes amplification while ensuring that each attribute is used only once.  First, we consider a simple case that the amplification factors of all marginals are disjoint. In this scenario, we can compose the total privacy cost using Theorem~\ref{thm:amplificationcomposition} and demonstrate using Example~\ref{example:partialobservation}.

\begin{theorem}\label{thm:amplificationcomposition}
    Consider an MCAR mechanism $M_{\Phi}$, 
and a sequence of $j$ mechanisms $M_1, \ldots, M_j$ with DP guarantees of $\epsilon_1, \ldots, \epsilon_j$ to $D$ and amplification attribute set $\mathcal{A}_{M_1},\ldots, \mathcal{A}_{M_j}$ respectively. If their amplification attribute sets do not overlap, then 
these mechanisms offers DP to the ground truth data $\bar{D}$ at a cost of $\bar{\epsilon}=\sum_{i=1} p_{\mathcal{A}_{M_i} \epsilon_i }$. 
\end{theorem}

\begin{proof}
    As all mechanisms $M_i$ work on disjoint sets of attributes, their amplification attribute sets $\mathcal{A}_i$ are also disjoint. Furthermore as the missing probabilities are always $\leq 1$, we always use all attributes in  $\mathcal{A}_i$ amplify marginal $M_i$. We can then use Theorem~\ref{thm:samplingamplification} to calculate the final amplified privacy cost $\bar{\epsilon}=\sum_{i=1} p_{\mathcal{A}_{M_i} \epsilon_i }$. 
\end{proof}

\begin{example} \label{example:partialobservation}
    Continuing from Example~\ref{example:usecase1}, assume we have the same dataset but use a partial observation algorithm. We consider only the sub-algorithms $M_1, M_2$, and $M_4$ for this example. The marginals for these sub-algorithms do not overlap and allow us to consider all engaging attributes as their amplification attribute set. 
    Hence, by Theorem~\ref{thm:amplificationcomposition}, 
    $M_1$ is amplified using $p_{M_1} = 1-\phi_{state} = \frac{3}{4}$, $M_2$ is amplified using $p_{M_2} = 1-\phi_{occupation} = 1$ and $M_4$ using $p_{M_4} = (1-\phi_{gender})(1- \phi_{income}) = \frac{9}{16}$. The amplified privacy cost would thus be $\frac{3}{4} \frac{\epsilon}{3} + \frac{\epsilon}{3} + \frac{9}{16} \frac{\epsilon}{3} = 0.77\epsilon$. 
\end{example}

The problem however gets more nuanced when two marginals have overlapping attributes. We show this in Example~\ref{example:partialobservation_nuanced} by first showing a na\"ive composition and then an optimized one. 

\begin{example}\label{example:partialobservation_nuanced}
Consider all $4$ sub-algorithms in 
Example~\ref{example:usecase1} and a partial observation algorithm. 
The marginals for sub-algorithms $M_3$ and $M_4$ overlap in the `Gender' attribute with amplification factors $p_{M_3} = (1 - \phi_{Gender}) = \frac{3}{4}$ and $p_{M_4} = (1 - \phi_{Gender}) (1 - \phi_{Income}) = \frac{9}{16}$ respectively. We cannot apply Theorem~\ref{thm:amplificationcomposition} on $M_3$ and $M_4$'s entire attribute set as the corresponding amplification attribute sets would then overlap on the `Gender' attribute. 
A na\"ive solution would be to amplify the DP mechanism with the most amplification and skip the others. In our example, we would amplify only $M_4$ with amplification of $\frac{9}{16}$ and skip amplification for $M_3$. The total amplified privacy cost would thus be $\bar{\epsilon} = \frac{3}{4} \frac{\epsilon}{4} + \frac{\epsilon}{4}  +  \frac{\epsilon}{4} + \frac{9}{16} \frac{\epsilon}{4} = 0.83\epsilon$. 
However, a better bound can be calculated if overlapping mechanisms were grouped together and amplified using the intersecting attribute. 
For instance, both $M_3$ and $M_4$ can 
be amplified by an amplification factor of $\frac{3}{4}$ using the amplification attribute set `Gender', resulting in a total privacy loss of $\frac{3}{4} \frac{\epsilon}{4} + \frac{\epsilon}{4}  +  \frac{3}{4} (\frac{\epsilon}{4} + \frac{\epsilon}{4}) = 0.81\epsilon$ i.e. tighter than $0.83\epsilon$.
\end{example}

In a more general setting, solving this problem requires us to make groups of the mechanisms with overlapping attributes and make sure that each group is amplified using distinct amplification factors. 

\begin{problem}~\label{prob:optimized_amplified_cost}
    Consider an MCAR mechanism $M_{\Phi}$ and 
 a sequence of $j$ mechanisms $M_1, \ldots, M_j$ with DP guarantees of $\epsilon_1, \ldots, \epsilon_j$ to $D$, where $M_i$ computes a marginal over attributes $\mathcal{A}_i$. We would like to find 
amplification attribute sets $\{\mathcal{A}_{M_1}\subseteq \mathcal{A}_1,\ldots, \mathcal{A}_{M_j}\subseteq \mathcal{A}_k\}$ and their corresponding amplification factor $p_1,\ldots, p_j$
for $M_1,...M_j$, that gives the smallest DP cost to the ground truth data $\bar{D}$. 
\end{problem}

One way to solve the above problem is by creating valid partitions of marginals and assigning each group in the partition an amplification attribute set such that all groups have disjoint attribute sets and all marginals from the same group are amplified using their own amplification attribute set. 

\begin{definition}[Valid partition] \label{def:valid_partition}
    Given DP mechanisms $M_1,\ldots,M_j$ for computing marginals over attribute sets $\mathcal{A}_1,\ldots, \mathcal{A}_j$, a partition of these mechanisms $P = \{G_1, \ldots G_i\}$ is considered valid if it satisfies these conditions: 
       (1) All mechanisms in the same partition are amplified with the same set of amplification attribute set and with the same amplification factors; and 
        (2) The amplification attribute sets of all partitions are disjoint. 
    The privacy cost for a valid partition is $\bar{\epsilon}=\sum_{G_l\in P}p_{\mathcal{A}_{G_l}}\sum_{M_j\in G_l} \epsilon_j$, where $\mathcal{A}_{G_l}$
    is the amplification attribute set for the mechanisms grouped into $G_l$ and 
    $p_{\mathcal{A}_{G_l}}$ is the corresponding amplification factor.
\end{definition}

\revision{A valid partition ensures each group's amplification attribute set is disjoint, ensuring each attribute is considered only once. To solve Problem~\ref{prob:optimized_amplified_cost}, we select the partition with the least privacy cost that is also valid. Thus, all DP sub-mechanisms are amplified using the best disjoint attribute set. However, enumerating all valid partitions is intractable due to the exponential number of possibilities\footnote{The total number of partitions of a set is given by the \href{https://en.wikipedia.org/wiki/Bell_number}{Bell number}.}. If the cardinality of $\mathcal{A}$ and the number of DP mechanisms $j$ are small, then one can enumerate all possible solutions and choose the best one. However, for datasets with a large number of attributes, we show an initial pruning method to trim away bad solutions using Lemma~\ref{lemma:atleastone}.}

\begin{lemma} \label{lemma:atleastone}
A valid partition $P$ solution to Problem~\ref{prob:optimized_amplified_cost} should have a non-empty amplification attribute set for all group $G \in P$.
\end{lemma} 
\ifpaper
\else
 \proof
Suppose there are a sequence of DP sub-mechanisms $M_1, \ldots, M_j$ with privacy cost of $\epsilon_1, \ldots, \epsilon_j$. A valid partition $P$ over $M_1, \ldots, M_j$ consists of a group $G$ that has an empty amplification attribute set. Then the amplification factor of group $G$ will be $p_G = 1$ and the overall privacy cost of this partition can be calculated using Definition~\ref{def:valid_partition}.
\begin{eqnarray*}
    \bar{\epsilon}_P &=&  \sum_{G_l\in P}p_{\mathcal{A}_{G_l}}\sum_{M_j\in G_l} \epsilon_j \\
    &=& \sum_{G_l\in \{P - G\}}p_{\mathcal{A}_{G_l}}\sum_{\mathcal{A}_j\in G_l} \epsilon_j +  \sum_{M_j\in G} \epsilon_j \\
    &=& L +  X
\end{eqnarray*}, where $L = \sum_{G_l\in \{P - G\}}p_{\mathcal{A}_{G_l}}\sum_{\mathcal{A}_j\in G_l} \epsilon_j$ and $X = \sum_{M_j\in G} \epsilon_j$. 
We can always define another valid partition $P'$ that has the groups as $P$ but group $G$ has amplification attribute set $A$. Therefore, the privacy cost of partition $P'$ can be calculated as:  
\begin{eqnarray*}
    \bar{\epsilon}_{P'} &=&  \sum_{G_l\in P'}p_{\mathcal{A}_{G_l}}\sum_{M_j\in G_l} \epsilon_j \\
    &=& \sum_{G_l\in \{P' - G\}}p_{\mathcal{A}_{G_l}}\sum_{\mathcal{A}_j\in G_l} \epsilon_j +  p_A \sum_{M_j\in G} \epsilon_j \\
    &=& L' +  p_{A}X
\end{eqnarray*}, where $L' = \sum_{G_l\in \{P' - G\}}p_{\mathcal{A}_{G_l}}\sum_{\mathcal{A}_j\in G_l} \epsilon_j$. 
Note that as amplification attribute set $A$ is used for $G$ in $P'$,  the privacy cost of the other groups can at most be increased by $p_A$. Assuming the worst case, $L' = p_{A} L$ and that the privacy cost of $P$ is always higher than that of $P'$. 
\begin{equation*}
\begin{split}
     L + X - L' - p_A X \geq 0\\
     L + X - p_A L - p_A X  \geq 0 \\
     L (1 - p_A) + X (1 - p_A) \geq 0 \\
     (L + X) (1-p_A) \geq 0
\end{split}   
\end{equation*}
The above inequality if always true as the first term is always positive ($L$ and $X$ are privacy costs) and the $p_A \leq 1$ as all missing percentages are $\leq 1$.  \qed

\fi

For large datasets where we are left with multiple partitions even after pruning, we use a brute force search as described in Algorithm~\ref{algo:mcar_amplification}. In Line 1, we enumerate all possible disjoint amplification attribute sets that we can make from $\mathcal{A}$ and store in a variable $\mathcal{P}_\mathcal{A}$. Then, we loop through each possible disjoint attribute set $P_\mathcal{A} \in \mathcal{P}_\mathcal{A}$, and calculate the cost of each amplification attribute set $\mathcal{A}_l \in P_\mathcal{A}$ in Line 2-3 using Proposition~\ref{prop:mcar_amplification}. We initialize the cost of the partition $c_{P_\mathcal{A}} = 0$ in Line 4. Then we loop through each pair of marginal $S_i$ and its corresponding DP mechanism ($M_i$, $\epsilon_i$) and find the candidate amplification attribute sets that are contained in the marginal $S_i$ in Line 6. If no such candidate set is valid for $S_i$, then we can prune the entire partition $P_\mathcal{A}$ using Lemma~\ref{lemma:atleastone} and loop back to Line 2. Otherwise, in Line 7-8, we find the candidate attribute set which has the best amplification cost for $M_i$ and add its corresponding cost the final cost $c_{P_\mathcal{A}}$. Finally, the partition with the minimum sum cost is returned in Line 10. In Example~\ref{example:partitioned_amplification} we show how Algorithm~\ref{algo:mcar_amplification} can be used to find the valid partition for our running example with the lowest privacy cost.

\begin{example} \label{example:partitioned_amplification}
    Consider the same setup in Example~\ref{example:partialobservation_nuanced}. 
    There are total of 4 attributes and Algorithm~\ref{algo:mcar_amplification}  starts by enumerating all 15 possible disjoint amplification attribute sets --- \{State | Occupation | Gender | Income\}, \{State Occupation | Gender | Income\}, ... , \{State Occupation Gender | Income\}, ... , \{State Occupation Gender Income\}. We then iterate through each of these disjoint sets. Let's consider the disjoint sets $\mathcal{A}_1:\{State\}, \mathcal{A}_2: \{Occupation\}, \mathcal{A}_3: \{Gender\}, \mathcal{A}_4: \{Income\}$. For each attribute set we calculate its corresponding amplification factor, $c_{\mathcal{A}_1}: (1-\phi_{State}) = 3/4, c_{\mathcal{A}_2}: (1-\phi_{Occupation}) = 0, c_{\mathcal{A}_3}: (1-\phi_{Gender}) = 3/4, c_{\mathcal{A}_4}: (1-\phi_{Income}) = 3/4$. We then iterate through all marginals and choose the best amplification attribute for each marginal. Therefore, $M_1$ is amplified using $\mathcal{A}_1$, $M_2$ using $\mathcal{A}_2$, $M_3$ and $M_4$ both using $\mathcal{A}_3$. The final privacy cost therefore is $\bar{\epsilon} = \frac{3}{4} \frac{\epsilon}{4} + \frac{\epsilon}{4}  +  \frac{3}{4} \frac{\epsilon}{4} + \frac{3}{4} \frac{\epsilon}{4} = 0.81\epsilon$. This partition also happens to be the best partition among the 15 partitions.    
\end{example}

\begin{algorithm}[t]
	\begin{algorithmic}[1]
        \REQUIRE Marginals $\mathcal{S}=\{S_1,\ldots, S_j\}$ over attributes $\mathcal{A}$, DP mechanisms $\mathcal{M}=\{(M_1,\epsilon_1),\ldots,(M_j,\epsilon_j)\}$ for $\mathcal{S}$, Missing probabilities $\Phi=\{\Phi_A|A\in\attrset\}$ for MCAR 
        \STATE Find all possible partitions of $\attrset$ and store in $\mathcal{P}_{\attrset}$ 
        \FOR{each attribute partition $P_{\attrset}$ in $\mathcal{P}_{\attrset}$} 
        \STATE Calculate amplification factor   $c_{\mathcal{A}_l} = \prod_{A\in \mathcal{A}_l} (1-\phi_A)$ for $\mathcal{A}_l\in P_{\attrset}$
        \STATE{Initialize disjoint set cost $c_{P_{\attrset}} = 0$}
        \FOR{each $S_i\in \mathcal{S}$ with its corresponding ($(M_i,\epsilon_i)$)}
        \STATE{Skip $P_{\attrset}$ if $\{\attrset_l\subseteq S_i | \attrset_l\in P_{\attrset}\}=\emptyset$}
        \STATE{Find  the best amplification attribute set for $M_i$:\\ $\attrset_{l*} \gets argmin_{\attrset_l\in P_{\attrset}\wedge \attrset_l\subseteq S_i} c_{\attrset_l}$}
        \STATE{Add the best amplified cost $c_{P_\attrset} = c_{P_\attrset} + \epsilon_i \cdot  c_{\attrset_{l*}}$}
        \ENDFOR
        \ENDFOR
        \STATE{Return the attribute partition $P_{\attrset}$ with minimum cost $c_{P_{\attrset}}$}
        \caption{Optimal amplified privacy cost} 
        \label{algo:mcar_amplification}
    \end{algorithmic}
\end{algorithm}

\stitle{Use case 3: Privacy amplification for column-wise imputation algorithms.}
\revision{Column-wise imputation algorithms (e.g., KaminoI) learn attributes sequentially in a predefined sequence $S$ over ${A_1, \ldots, A_k}$. The first attribute $S_1$ is learned using its observed distribution, while the rest $S_2, \ldots, S_k$ are learned using intermediate models $M_2, \ldots, M_k$ with privacy costs $\epsilon_2, \ldots, \epsilon_k$. At each $i^{th}$ iteration, attribute $S_i$ is learned using intermediate model $M_i$, taking previously learned attributes $S_{:i}$ as feature input. After training, $M_i$ is used to sample the synthetic dataset and impute incomplete values of $D[S_i]$. Thus, for the next model $M_{i+1}$, previously learned attributes are either complete or imputed. The total number of complete rows fed to an intermediate model $M_i$ depends solely on the complete values in $D[S_i]$, leading to amplified privacy for ground truth data $\bar{\epsilon_i} = p_{S_i}\epsilon_i$, where $p_{S_i} = 1 - \phi_i$. However, we note that for training model $M_i$, every attribute $S_{:i-1}$ is considered for amplification using the same probability as $S_i$. Since each attribute can be amplified only once, privacy for ground truth is calculated accordingly.}

\begin{theorem}\label{thm:columnamplification_mcar}
    Consider an MCAR mechanism $M_{\Phi}$ and  a sequence of attributes $S$ and $k$ mechanisms $M_1, \ldots, M_k$ with privacy cost of $\epsilon_1, \ldots, \epsilon_k$ to $D$, where $M_i$ is an intermediate model that trains $A_i$ as target and $A_{:i}$ as features. If model $M_i$ is used for imputation of attribute $A_i$, then the overall process
  offers DP to the ground truth data $\bar{D}$ at a cost of $\bar{\epsilon} = p_{S_j}\epsilon_j + \sum_{i=1, i\neq j}^k \epsilon_i$, where $j = \max\limits_{i=1}^{k} (1- \phi_i)$. 
\end{theorem}

\ifpaper 
\else
\begin{proof}
    As all cells of an attribute are imputed at previous iterations, the number of complete rows fed to a model at the $i$th iteration depends on the missing cells of the $S_i$ attribute. The amplification for the $i$th model can therefore be calculated at $\epsilon_i p_{S_i}$, where $p_{S_i} = 1- \phi_i$ is the amplification factor of attribute $S_i$. To calculate the overall privacy of the algorithm, we would like to compose the amplified privacy of all sub models. However, the model $M_i$ trains on every attribute $S_{:i-1}$ and each attribute uses the same probability as $S_i$. Furthermore as each attribute can only be used once, we can only amplify a single model. Hence, we choose to amplify the model with the maximum amplification $M_j$ where $j = \max\limits_{i=1}^{k} (1- \phi_i)$ and the other models are left un-amplified.
\end{proof}
\fi

\begin{example}
   \label{example:usecase3}
   Consider the incomplete dataset from Figure~\ref{fig:mnarexample}. Lets assume we have a column based imputation algorithm learns this dataset. Each attribute is learnt using the equal privacy budget $\frac{\epsilon}{4}$. The attributes `State', `Gender' and `Income' has an amplification factor $\frac{3}{4}$ whereas the attribute `State' has a factor of $1$. Therefore, the amplified privacy budget is $ \bar{\epsilon} = \frac{3}{4} \frac{\epsilon}{4} + 3 \frac{\epsilon}{4} = 0.9375\epsilon$.
\end{example}

\stitle{Discussion on MAR and MNAR.}
\revision{MAR and MNAR mechanisms model probabilities of missing data conditioned upon non-missing values from other columns (MAR) or missing values from the same column in the dataset (MNAR). Unfortunately, verifying the presence of MAR and MNAR missing types in the dataset lacks a foolproof method, as it may rely on observed and unobserved variables and their interactions (as discussed in Section~\ref{sec:preliminaries}). Due to the conditional nature of these missing types, achieving privacy amplification with MAR or MNAR mechanisms is challenging. Each row's probability of having missing values may vary due to conditional dependence on dataset values. While maximum probabilities or noisy upper bounds can be used for privacy amplification if these probabilities are known or privately learned, they are often associated with multiple patterns and complex relationships. For instance, in the simple example depicted in Figure~\ref{fig:mnarexample}, rows with the highest income value (80k) are missing. Here, the probability of a row having a missing value depends directly on the probability of a row having an income of 80k. To calculate missing probabilities associated with MNAR, not only the highest value in the column but also its probability of occurrence needs to be determined. Without distribution assumptions, such calculations cannot be achieved with DP \cite{bun2015differentially}. Amplification extensions for MAR and MNAR remain a topic for future research.}

%% file: experiments.tex
\section{Evaluation}\label{sec:exp}

\ifpaper
\else
\begin{figure*} [htb]
    \centering
    \begin{minipage}{\textwidth}
    \subfigure[1-way ($\downarrow$)]{\includegraphics[width=.32\textwidth]{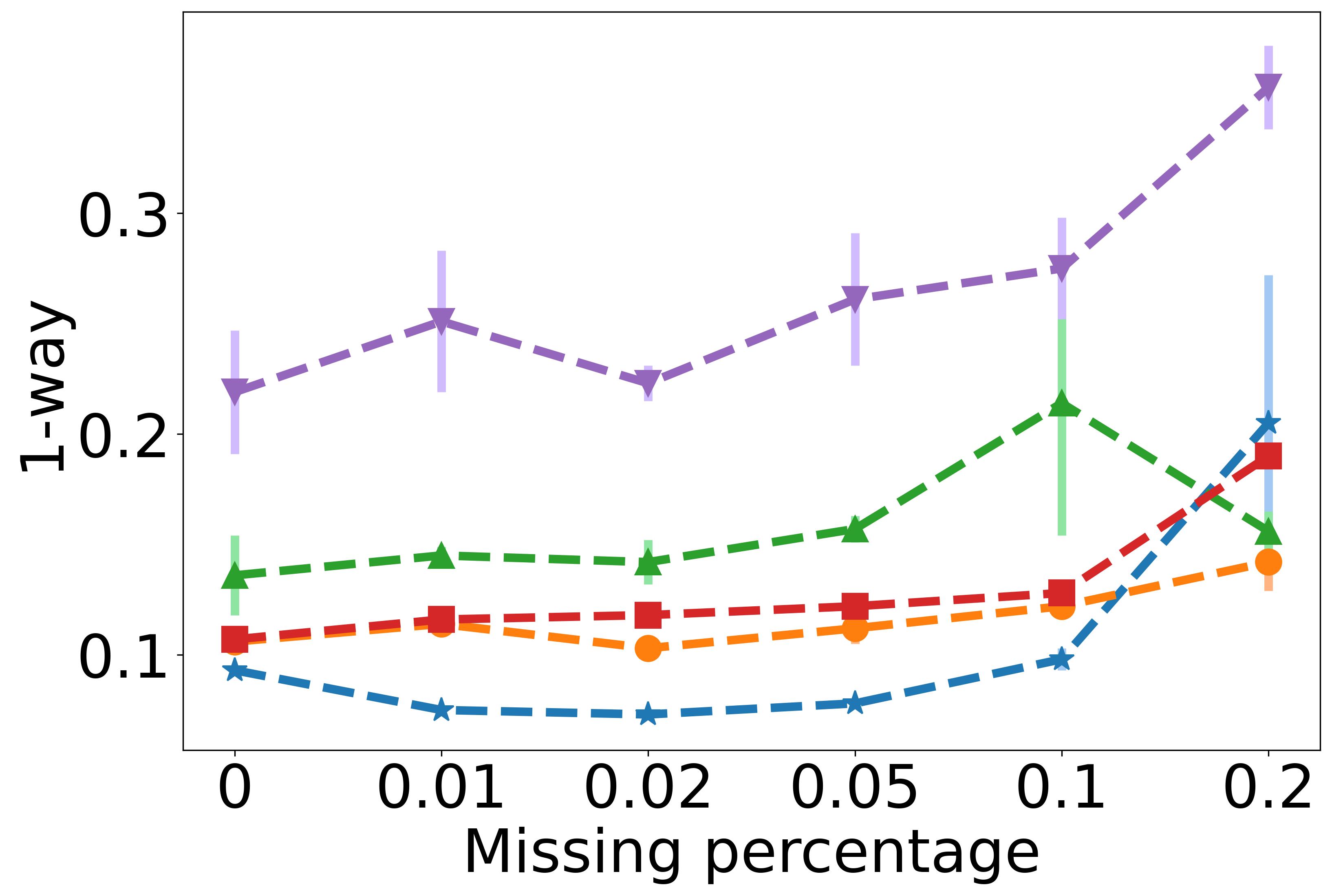}}
    \subfigure[2-way ($\downarrow$)]{\includegraphics[width=.32\textwidth]{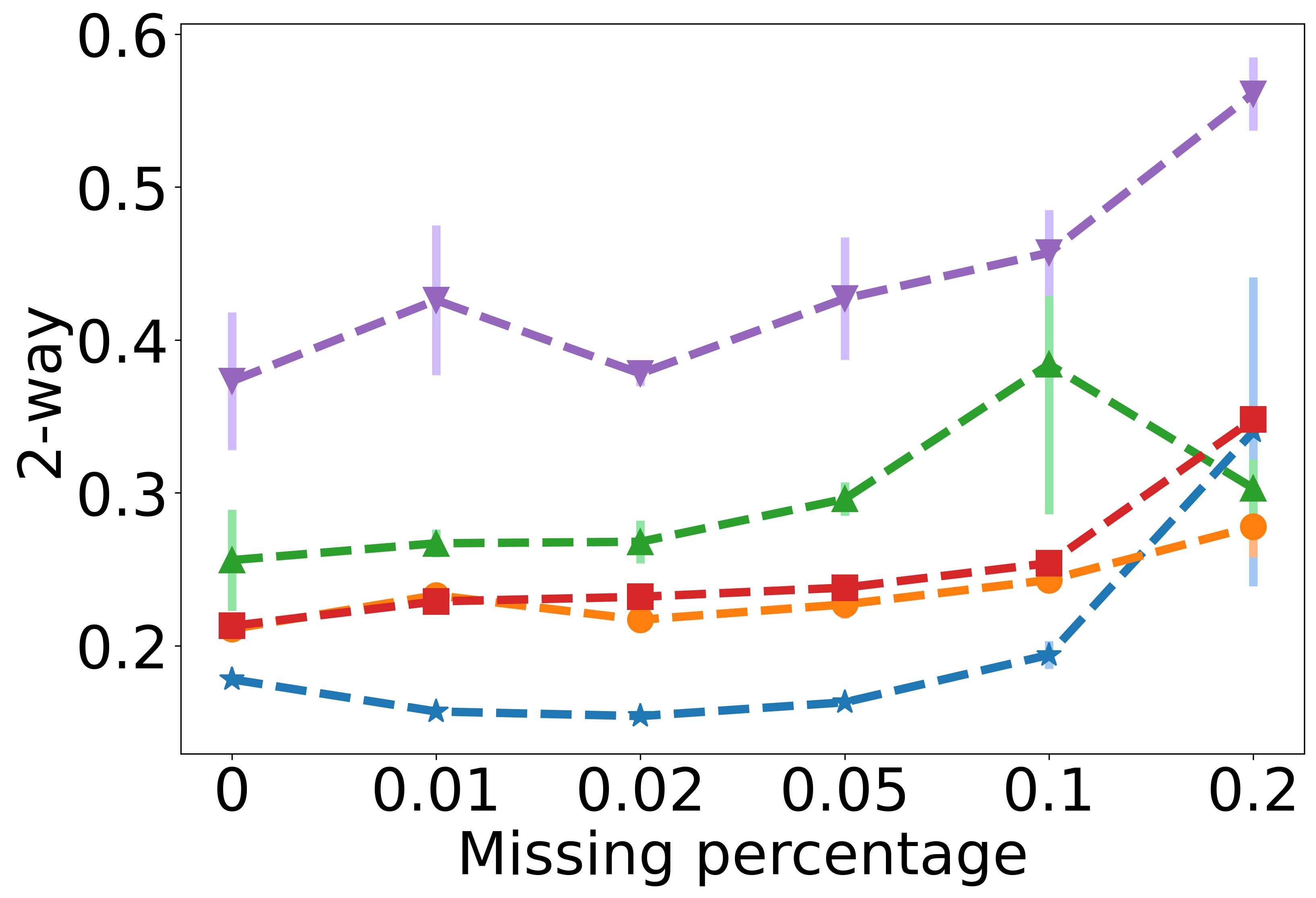}}
    \subfigure[F1-score ($\uparrow$)]{\includegraphics[width=.32\textwidth]{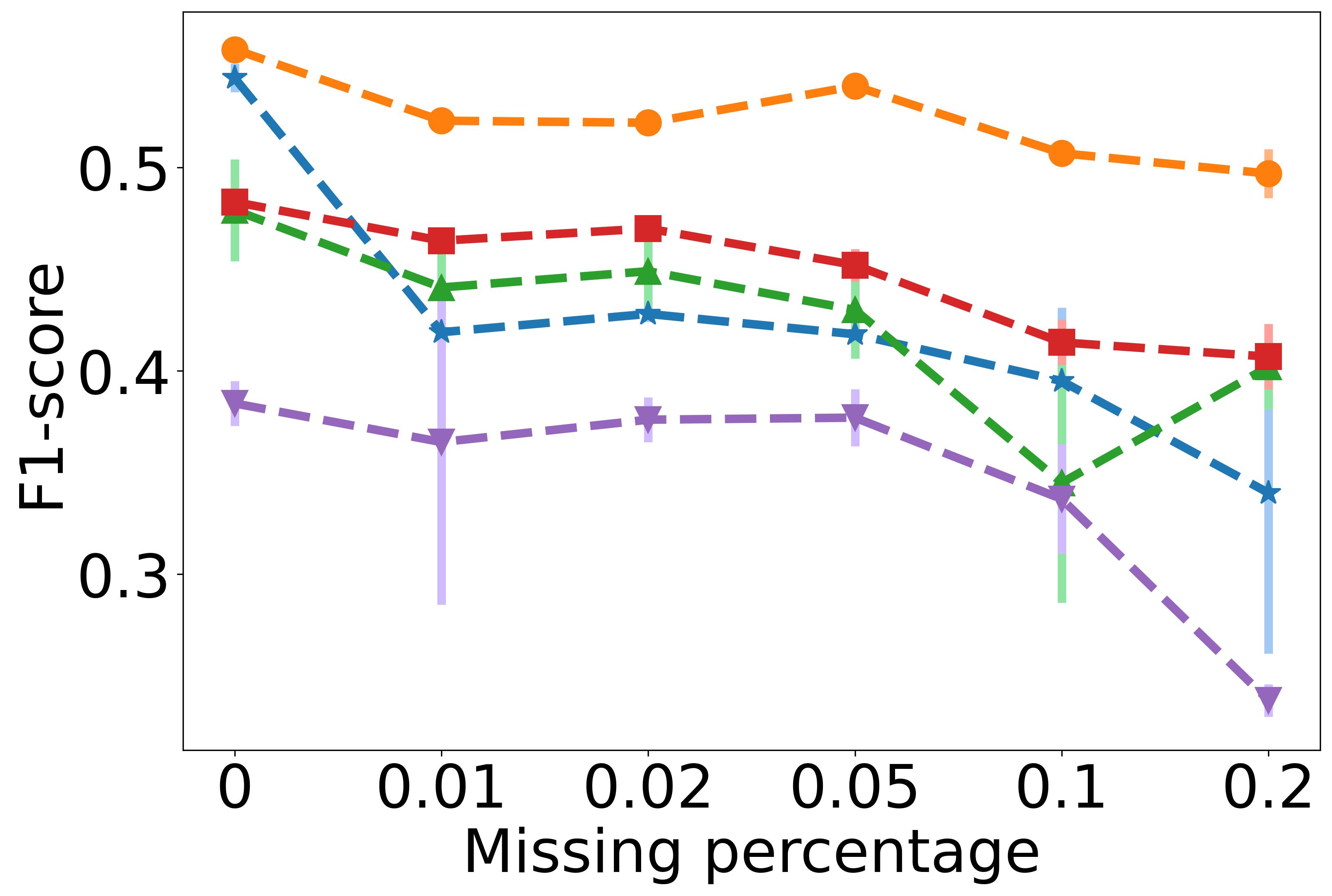}}
    \centering \subfigure{\includegraphics[width=0.75\textwidth]{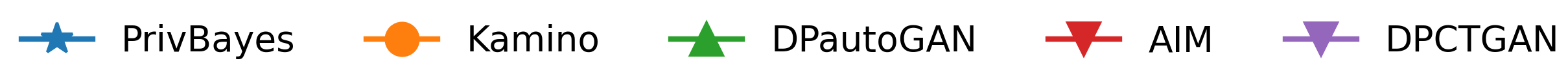}}
    \vspace{-0.3cm}
    \caption{The effect of missing data on DP synthetic data generation algorithms.}
    \label{fig:exp1}
    \end{minipage}%
\end{figure*}
\fi


We thoroughly  experiment with DP synthetic data generation algorithms on missing data. First, we demonstrate how existing DP methods are affected by varying amount of missing data. Next, we evaluate the effectiveness of our proposed adaptive recourse methods and analyze the impact of varying missing data percentages, missing mechanisms, and privacy budgets for each method. Finally, we show how missingness amplifies privacy for ground truth data.
\subsection{Experimental Setup}

\ifpaper
\else
\begin{table}[t]
    \centering
    \caption{Dataset Characteristics}
    \begin{tabular}{|c|c|c|c|}
        \hline
        Dataset & Cardinality & \#Numerical Attr & \#Categorical Attr  \\
        \hline
        \hline
         Adult & 32561 & 5 & 10 \\ 
         Bank & 45211 & 3 & 14 \\
         BR2000 & 38000 & 3 & 11 \\
         National & 15012 & 6 & 14 \\
         \hline
    \end{tabular}
     
    \label{tab:datasets}
\end{table}
\fi

\begin{figure*}
\centering
\subfigure[Adult 2-way ($\downarrow$)]{\includegraphics[width=.24\textwidth]{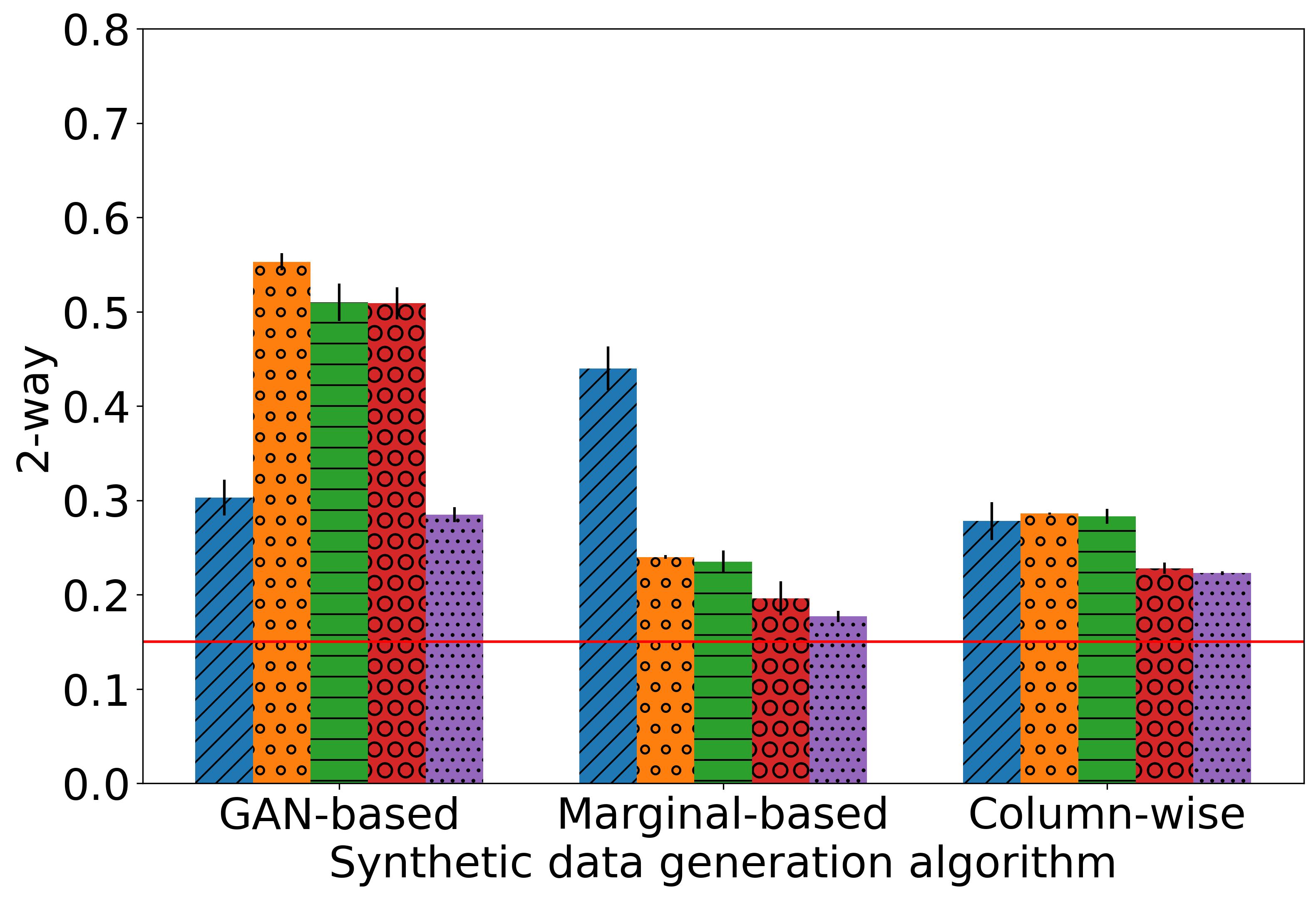}}
\subfigure[Bank 2-way ($\downarrow$)]{\includegraphics[width=.24\textwidth]{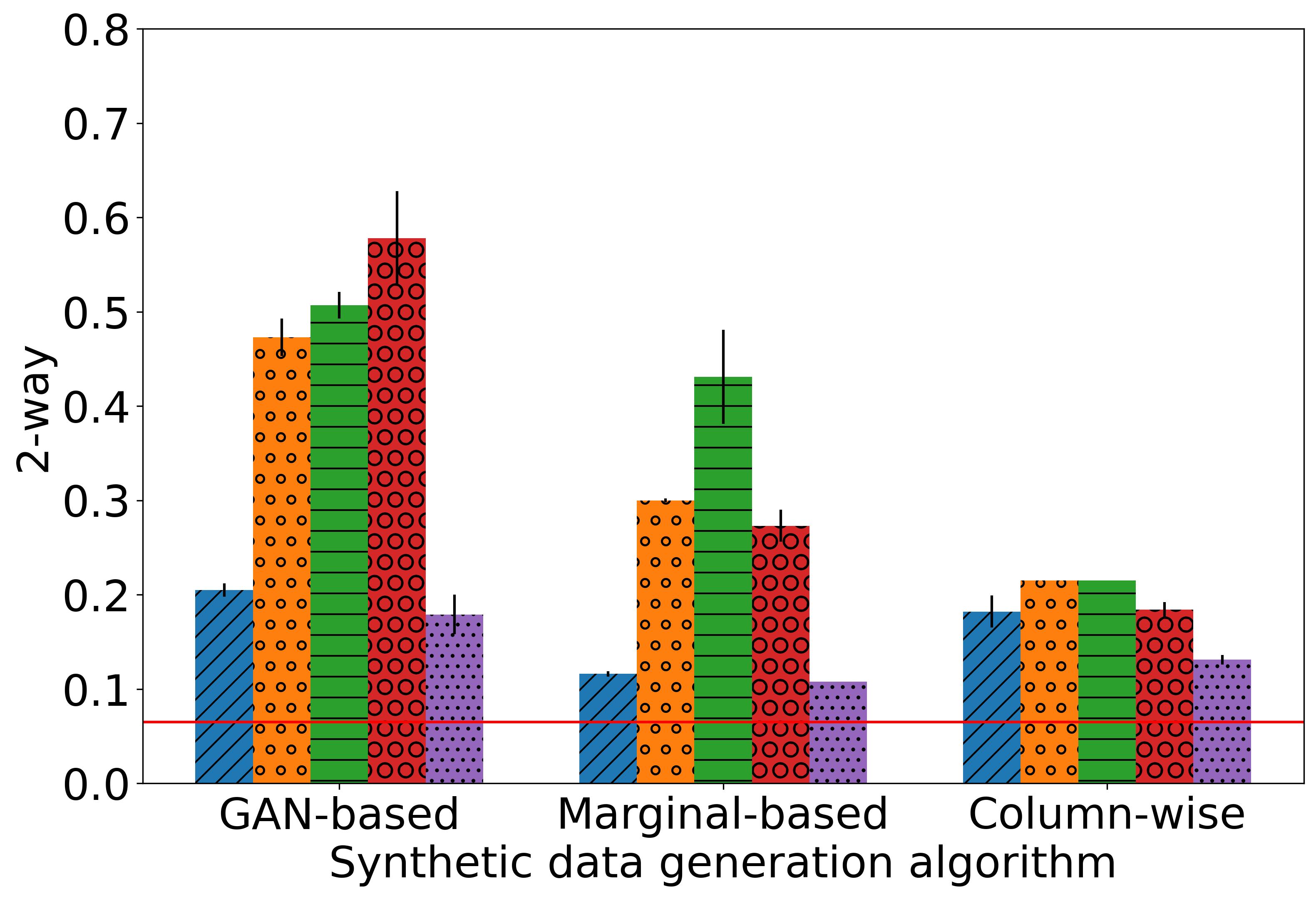}}
\subfigure[BR2000 2-way ($\downarrow$)]{\includegraphics[width=.24\textwidth]{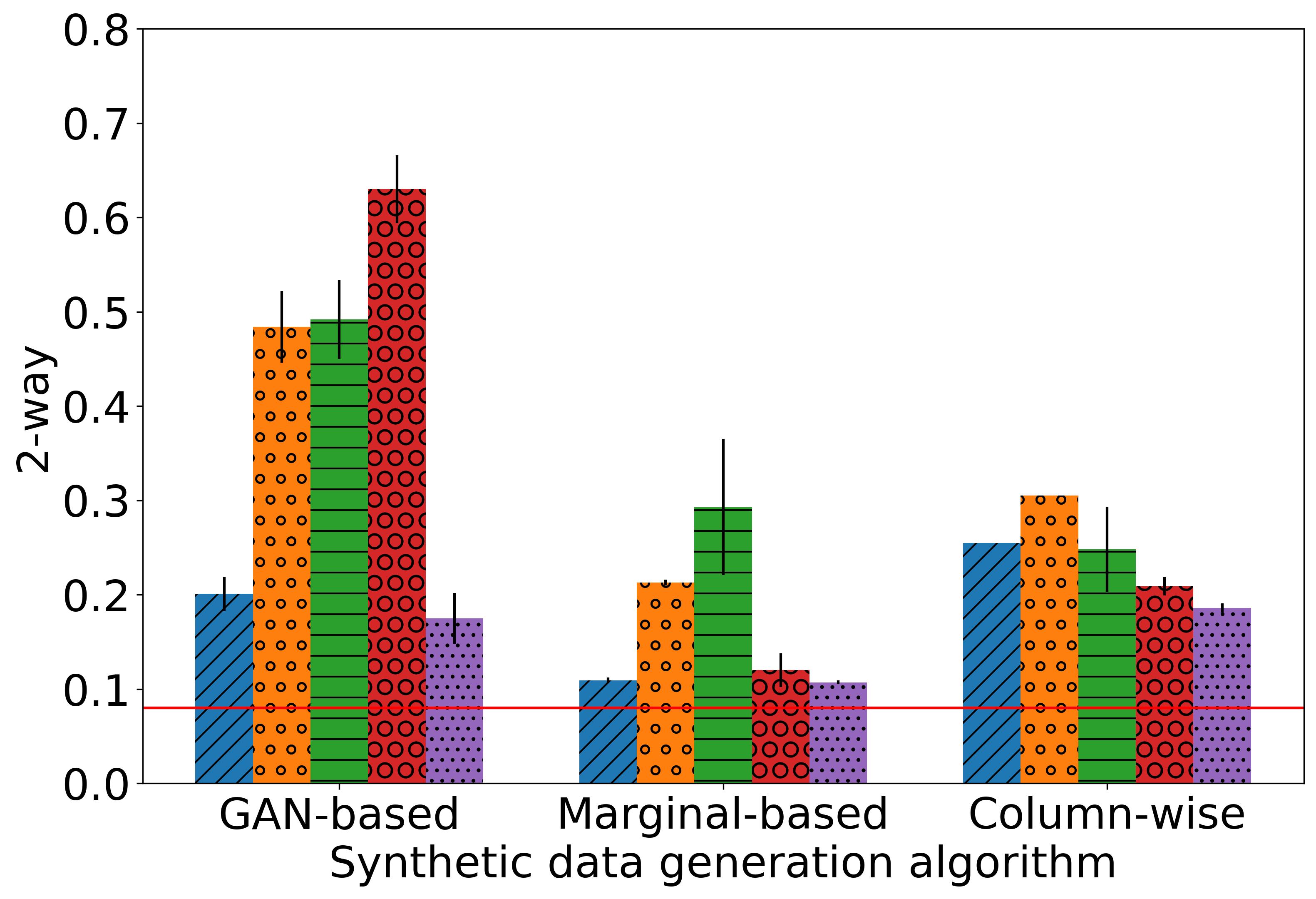}}
\subfigure[National 2-way ($\downarrow$)]{\includegraphics[width=.24\textwidth]{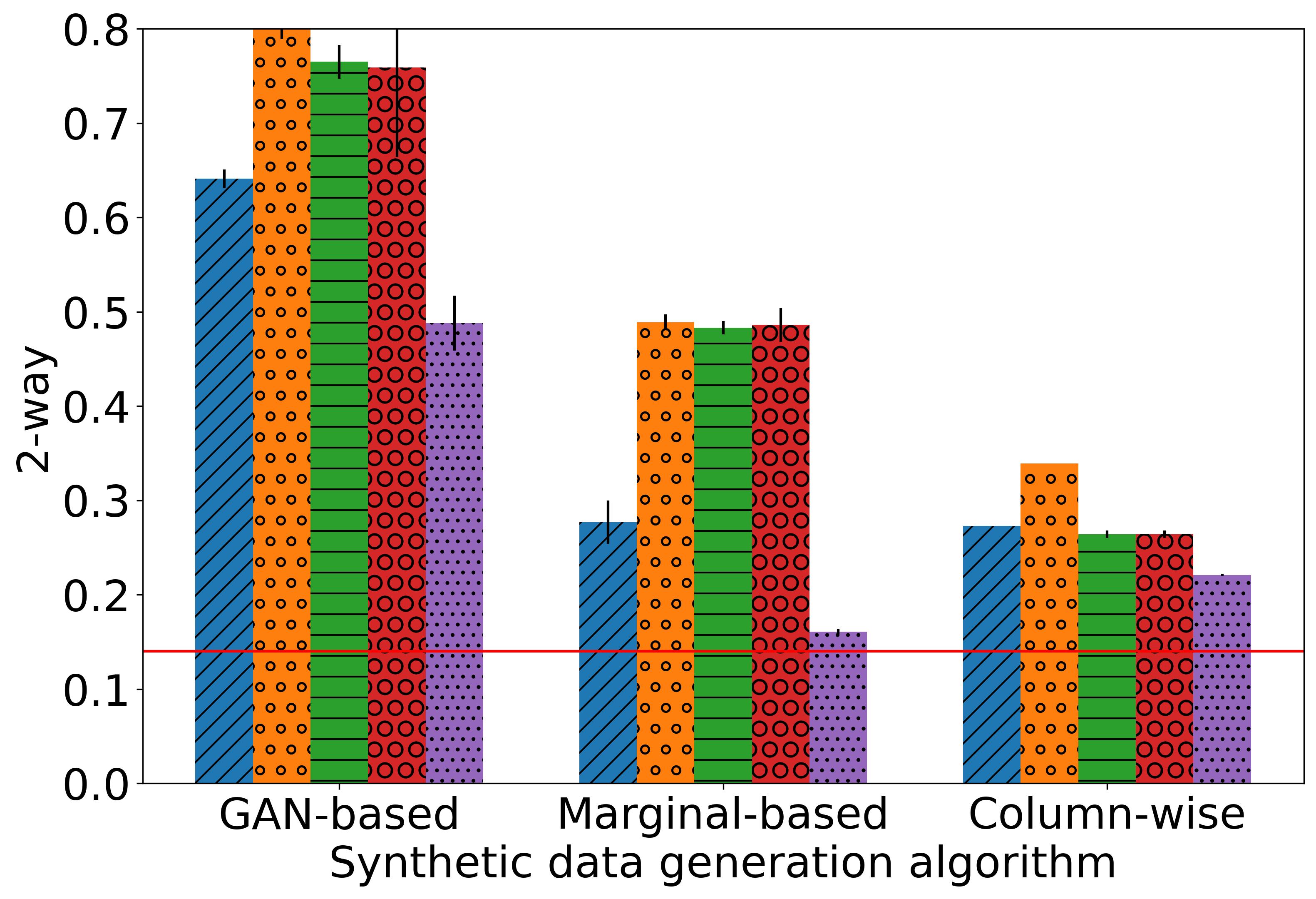}}

\subfigure{\includegraphics[width=0.8\textwidth]{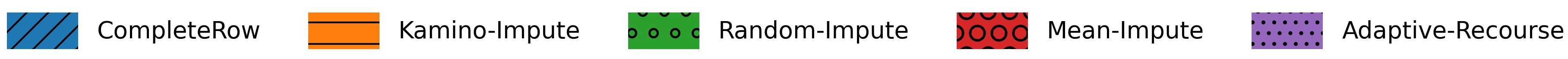}}

\vspace{-3mm}
\caption{\revision{Comparing all strategies to deal with missing data. Adaptive recourse results in the best performance, followed by the complete row approach. The red line represents the best no-missing baseline.} }
\label{fig:exp5}
\end{figure*}

\stitle{Datasets.}
\ifpaper
We run our experiments on four tabular datasets: (i) Adult dataset~\cite{uci_repo} (32561 rows), (ii) Bank dataset~\cite{bank} (45211 rows), (iii) BR2000~\cite{DBLP:conf/sigmod/ZhangCPSX14} (38,000 rows), (iv) National dataset~\cite{national} (15012 rows). Each dataset has a combination of numerical and categorical columns which are pre-processed according to the synthetic data generation algorithm as discussed in their respective research paper -- numerical attributes are discretized into 10 uniform bins or scaled between 0 to 1 and the categorical attributes are encoded using one-hot/ordinal encoding. We implement a pipeline that can generate different categories of missing data using the approach from Muzellec et al.~\cite{muzellec2020missing}. We run experiments on all the three types of missingness for every dataset and go up to 30\% missing values except the national dataset, where we stop at 20\% due to the lower number of rows in the original dataset. 
\else
We run our experiments on four tabular datasets as described in Table~\ref{tab:datasets}: (i) Adult dataset~\cite{uci_repo}, which contains information about 32561 individuals from the 1994 US Census (ii) Bank dataset~\cite{bank}, which has 45211 rows about direct marketing campaigns of a Portuguese banking institution (iii) BR2000~\cite{DBLP:conf/sigmod/ZhangCPSX14}, which consists of 38,000 census records collected from Brazil in the year 2000, (iv) National dataset~\cite{national} from NIST Diverse Community Excerpts which contains information about 15012 individuals US census data. Each dataset has a combination of numerical and categorical columns which are pre-processed according to the synthetic data generation algorithm as discussed in their respective research paper -- numerical attributes are discretized into 10 uniform bins or scaled between 0 to 1 and the categorical attributes are encoded using one-hot or ordinal encoding. 
\fi

\stitle{Baselines.} 
\revision{We consider several existing differentially private data generation methods that do not consider missing data:  PrivBayes~\cite{DBLP:conf/sigmod/ZhangCPSX14} and AIM~\cite{mckenna2022aim} from statistical approaches, DPCTGAN~\cite{fang2022dp} and DPautoGAN~\cite{tantipongpipat2021differentially} from deep learning techniques, and Kamino~\cite{ge2021kamino}, which is a mixed approach. These methods have been 
published in well-known conferences with their code readily available. For each method, we construct the following baselines to deal with missing data as discussed in Section~\ref{sec:vanilla}. The first baseline referred by its original name combines the original method with \underline{the complete row-only approach}. Next, we construct baselines for \underline{the imputation first approach}, which first imputes the data and then generates the synthetic data. We initially considered differentially private data imputation using k-means~\cite{DBLP:journals/tissec/CliftonHMM22} and OLS regression~\cite{DBLP:journals/corr/abs-2206-15063}, but they fail in working with categorical columns and require a large privacy budget for the imputation process itself. Hence, we adopt the following imputation methods: (i) random imputation, which fills missing values randomly from the attribute domain; (ii) mean imputation, replacing missing numerical attribute values with the mean and sampling from the probability distribution for categorical attributes; and (iii) Kamino imputation, which leverages intermediate models from Kamino for imputation and shares similar model as  Holoclean~\cite{aimnet,holocleancode}. Random impute requires no privacy budget, but the other two require splitting of the privacy budget for imputation and for the data generation algorithm.
}

\begin{figure*}
\centering

\subfigure[Adult 2-way ($\downarrow$)]{\includegraphics[width=.24\textwidth]{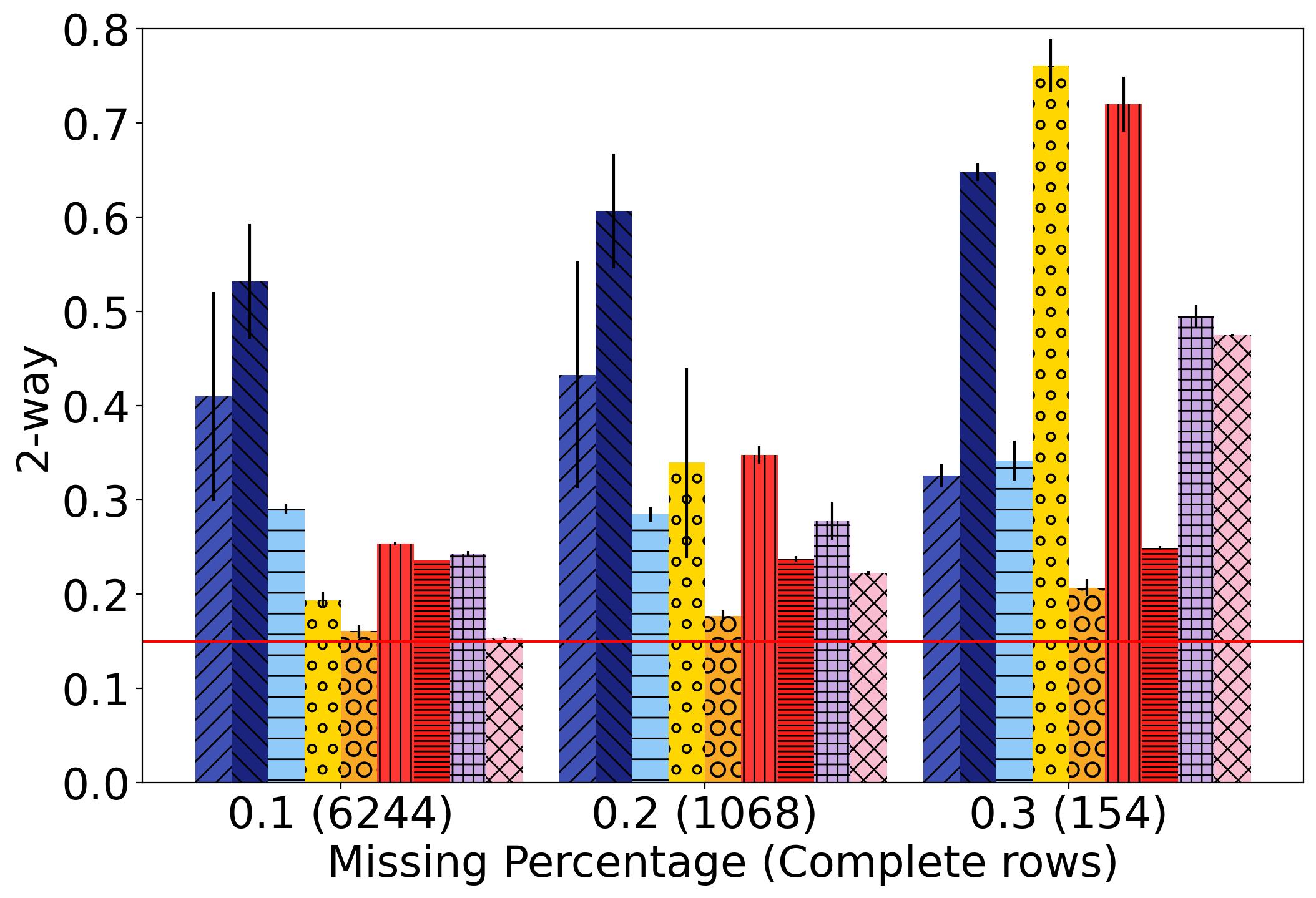}}
\subfigure[Bank 2-way ($\downarrow$)]{\includegraphics[width=.24\textwidth]{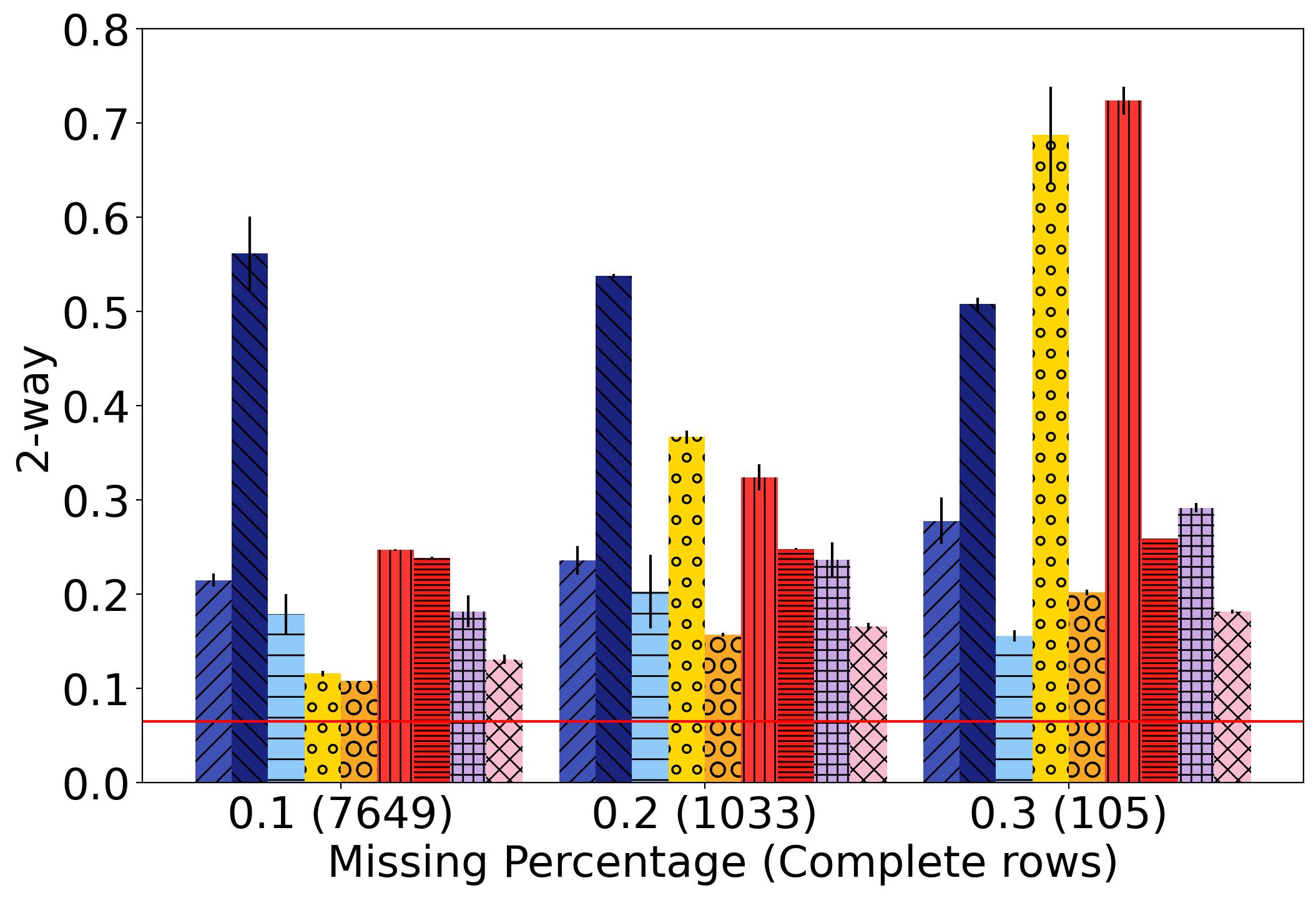}}
\subfigure[BR2000 2-way ($\downarrow$)]{\includegraphics[width=.24\textwidth]{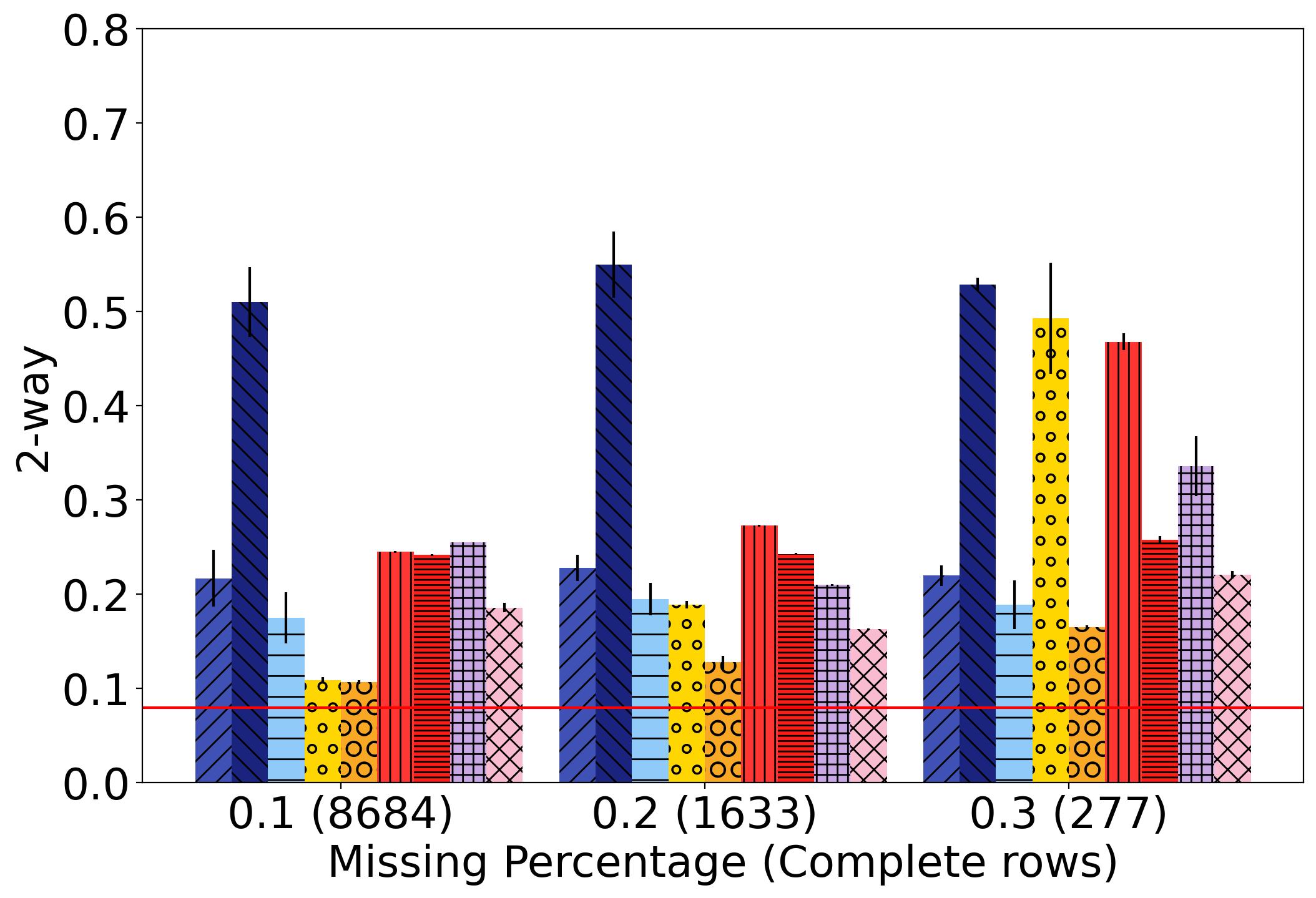}}
\subfigure[National 2-way ($\downarrow$)]{\includegraphics[width=.24\textwidth]{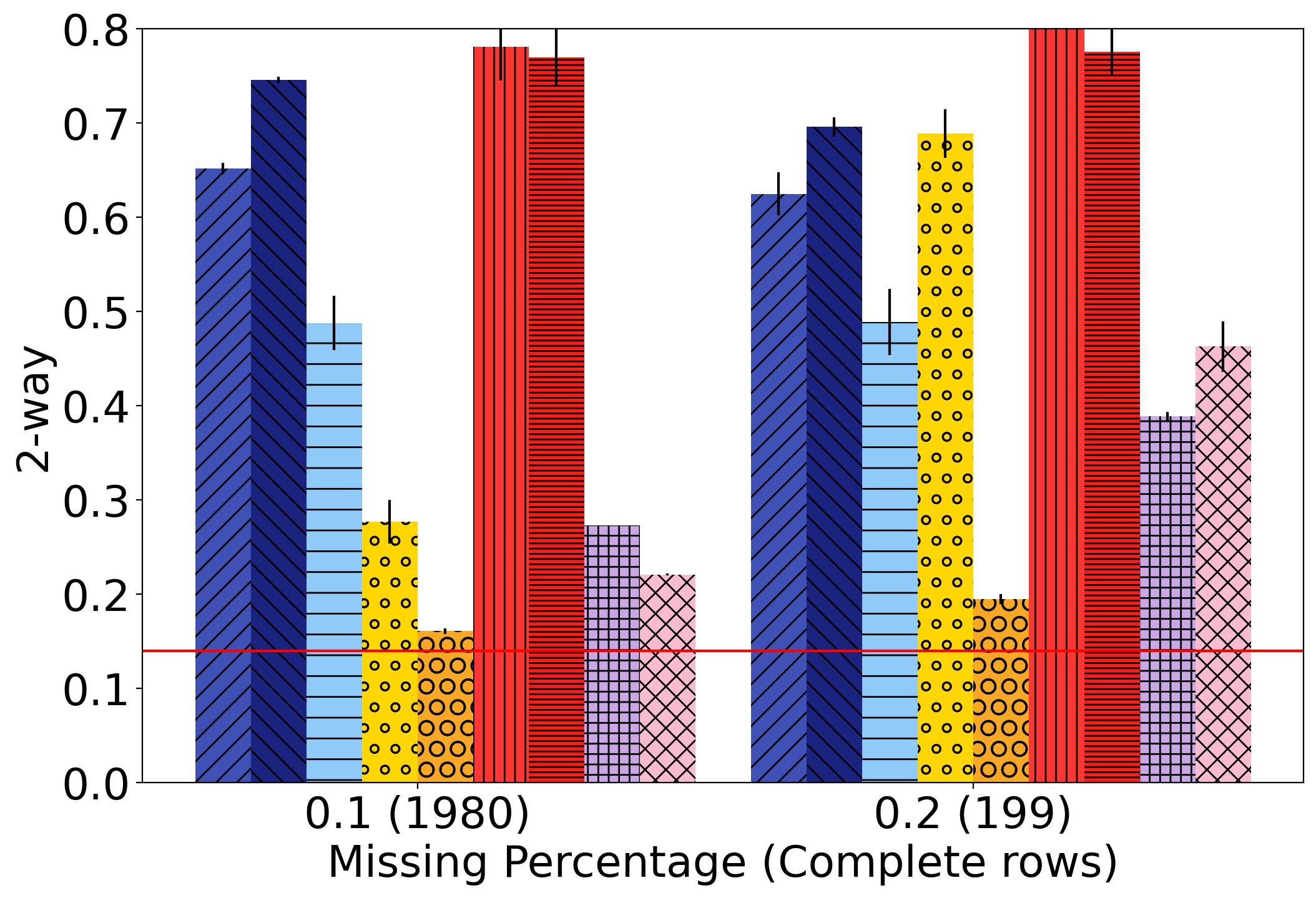}}

\subfigure[Adult F1-score ($\uparrow$)]{\includegraphics[width=.24\textwidth]{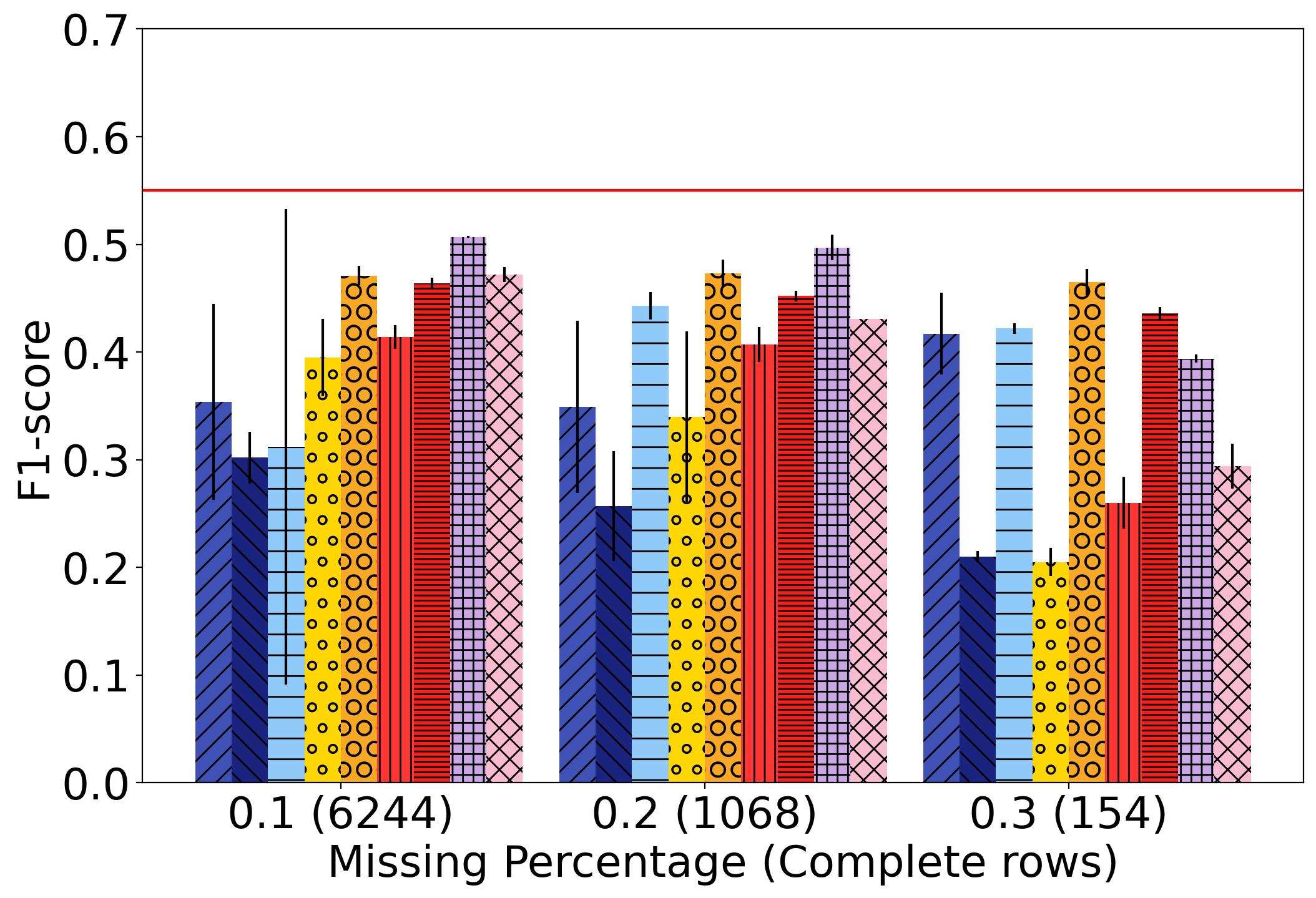}}
\subfigure[Bank F1-score ($\uparrow$)]{\includegraphics[width=.24\textwidth]{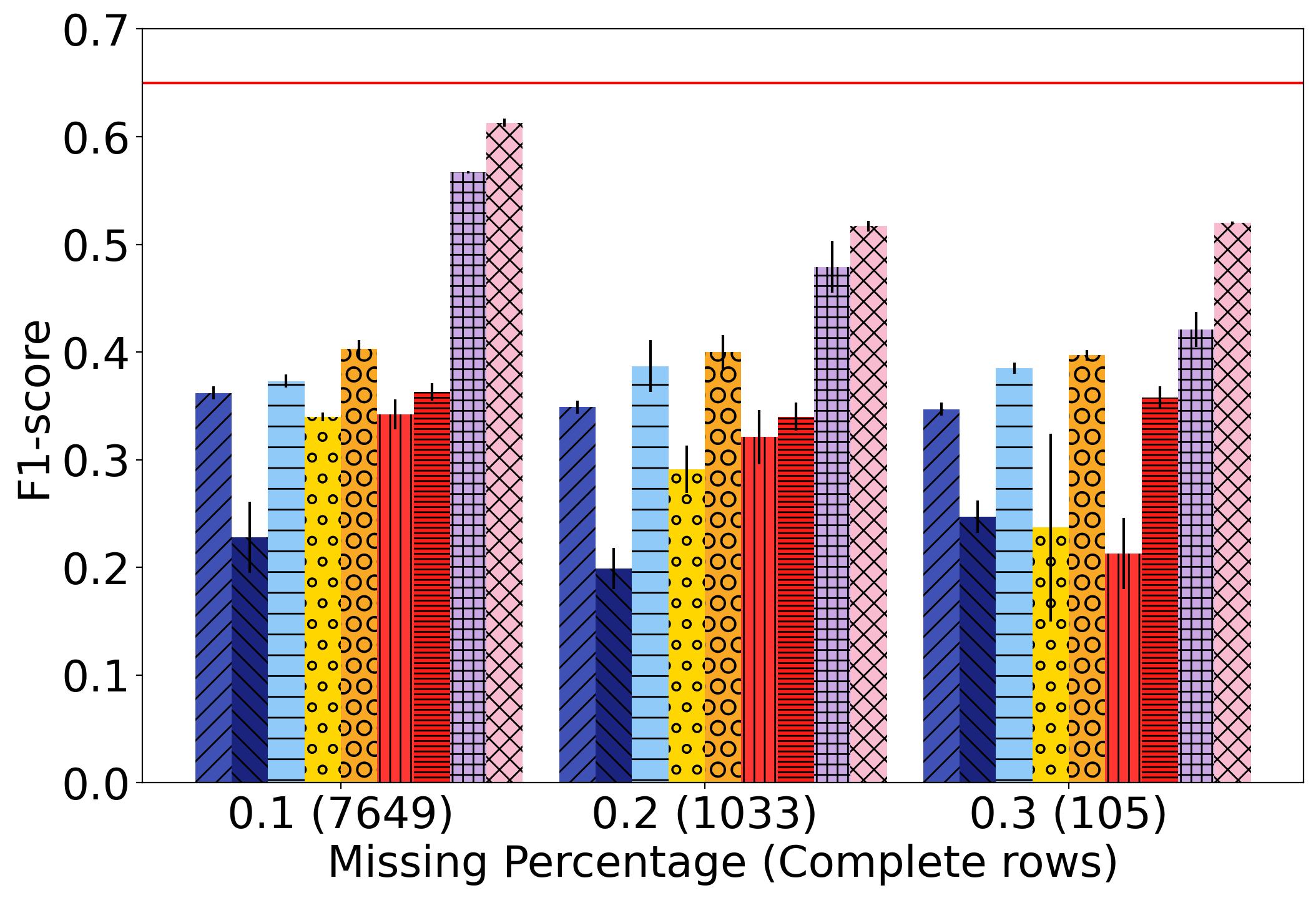}}
\subfigure[BR2000 F1-score ($\uparrow$)]{\includegraphics[width=.24\textwidth]{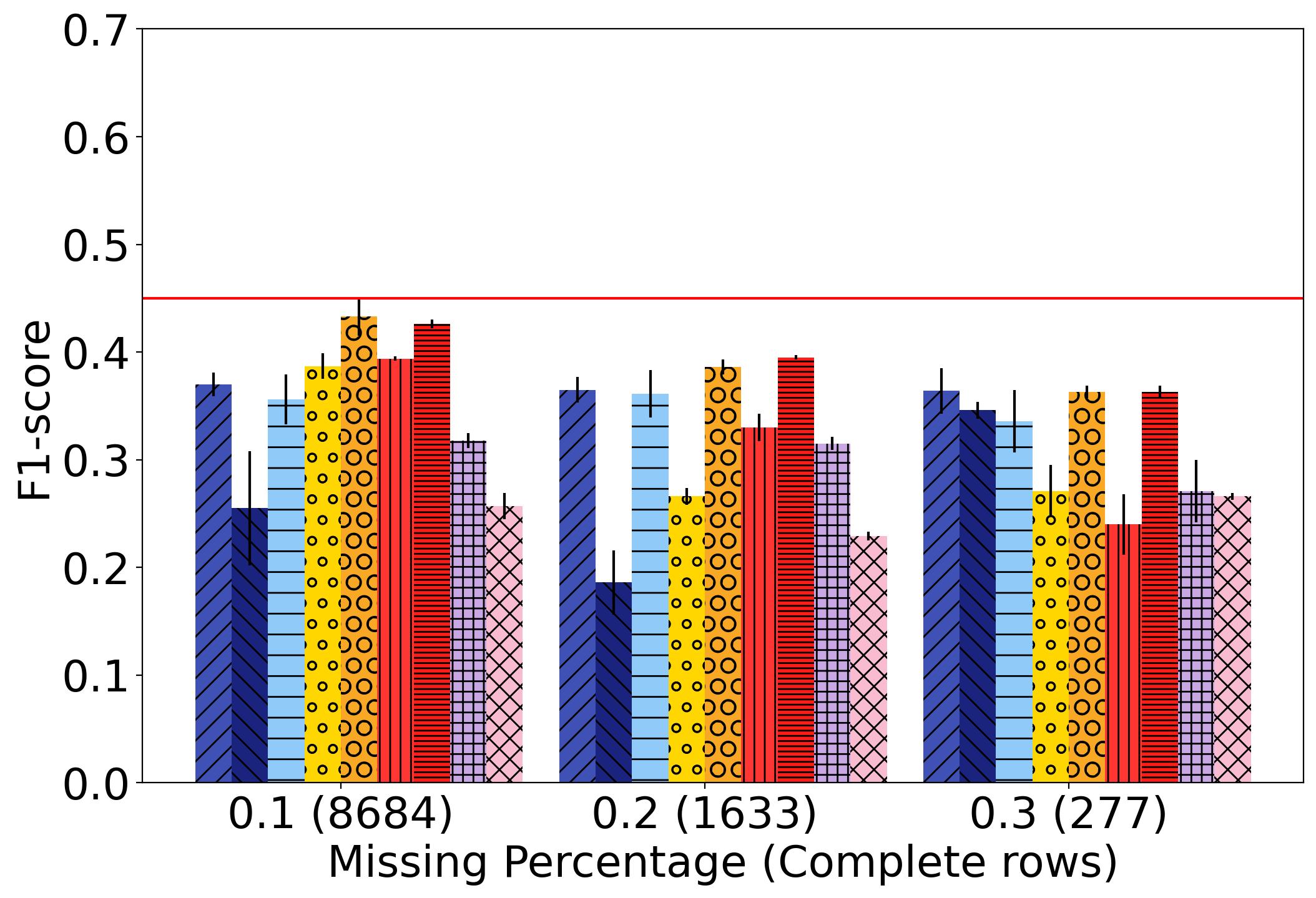}}
\subfigure[National F1-score ($\uparrow$)]{\includegraphics[width=.24\textwidth]{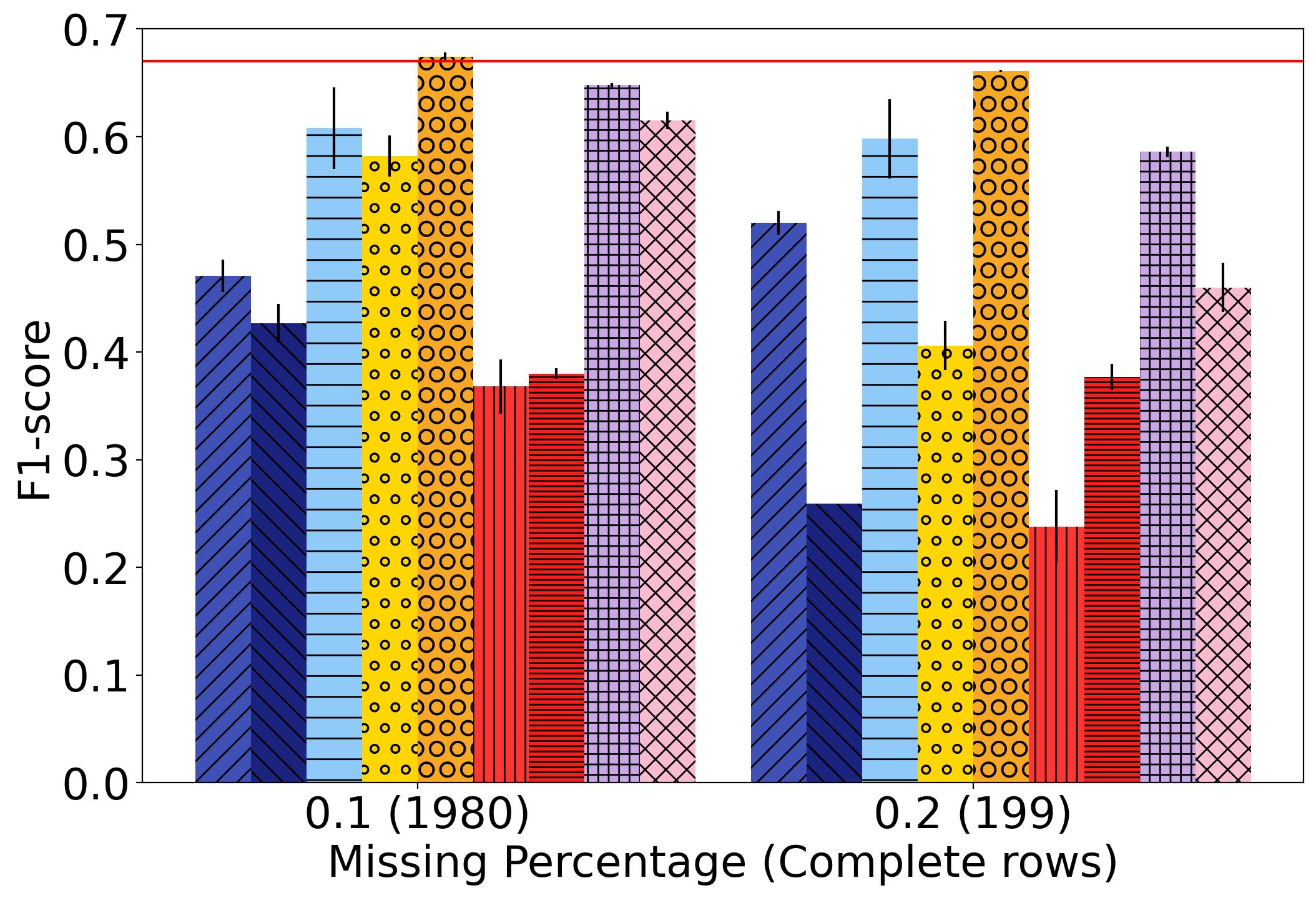}}

\subfigure{\includegraphics[width=\textwidth]{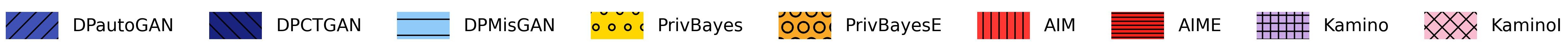}}

\vspace{-3mm}
\caption{
\revision{Adaptive methods (DPMisGAN, PrivBayesE, KaminoI) vs their respective baselines using the complete row approach (DPautoGAN, PrivBayes, Kamino)} 
with MCAR missingness at $\epsilon=1$. Algorithms of the same category are colored with the same shade. The adaptive methods result in better quality synthetic data. The red line denotes the best no missing baseline. }
\label{fig:exp2}
\end{figure*}

\stitle{\revision{Parameters.}} \revision{The utility of the algorithms}  
replies upon multiple hyperparameters, and we try our best to tune these parameters by running grid searches.
Although we don't consider privacy costs for tuning, it's crucial to tune these parameters privately in practice~\cite{mohapatra2021role, papernothyperparameter}. Deep learning methods typically require a larger privacy budget for meaningful results. Hence, we assign PrivBayes and Kamino $\epsilon=1$, \revision{and} GAN approaches  $\epsilon=3$ by default. Adaptive methods maintain the same budget as their non-adaptive counterparts. PrivBayes operates under pure DP ($\delta = 0$), while others employ approximate DP. The $\delta$ value is approximately set one magnitude lower than $1/|D|$ to the nearest exponent of 10.

\ifpaper
\else
\stitle{Missing data.}
We implement a pipeline that can generate different categories of missing data at different missing percentages, using the approach from Muzellec et al.~\cite{muzellec2020missing}. MCAR missing data is generated by masking the original dataset using a realization of a Bernoulli variable with a fixed parameter such that there is exactly the required number of missing values. To generate MAR data, we first use 50\% of the attributes in the dataset as features for a logistic regression model. The other attributes then have missing values according to random weights in the logistic model. A bias term is fitted using line search to attain the desired proportion of missing values. The MNAR approach works similarly to MAR, with the difference that the 50\% non-missing attributes are masked by an MCAR mechanism. This imposes the logistic model to depend potentially on missing values, hence enforcing MNAR missingness. As 50\% of the columns in the MAR mechanism has complete values, it comparatively has more complete rows in comparison to MCAR and MNAR. We run experiments on all kinds of missing types for every dataset and go up to 30\% missing values except the national dataset, where we stop at 20\% due to the lower number of rows in the original dataset. 
\fi

\stitle{Metrics.}
\ifpaper 
We use two utility metrics for evaluation: \revision{(i) the average variational distance between the $k$-way marginals of the ground truth data and that of the synthetic data;  and (ii) the average F1 scores of 9 classification models in classifying each attribute in the dataset.
All our experiments are repeated 3 times and reported with their the mean with standard deviation. For space constraints, we show limited metrics for some experiments and defer others and the details of the to full paper~\cite{full_paper}.}
\else
We use two utility metrics for evaluation.
The first metric is $k$-way marginal distance. 
We calculate all $k$-sized combinations of attributes in the dataset~\cite{DBLP:conf/sigmod/ZhangCPSX14, ge2021kamino}, and then report the average  variational distance between the true vs the synthetic marginals. For each $k$-sized combination set of attributes in $A \in \mathcal{A}$, we calculate the marginal, $h:\mathcal{D} \rightarrow  \mathbb{R}^{|\mathcal{D}(A)|}$ and report the average variational distance between as $\max_{a\in \mathcal{D}(A)} |h(D^\prime)[a]-h(\bar{D})[a]|$ where $D^\prime$ and $\bar{D}$ are the synthetic data and ground true data respectively~\cite{DBLP:conf/sigmod/ZhangCPSX14, ge2021kamino}. In our evaluation, we set $k=1,2$. Smaller values of this metric show more closeness between the true and synthetic data. 
Our second metric is model training. We consider 9 classification models (LogisticRegression, AdaBoost, GradientBoost, XGBoost, RandomForest, BernoulliNB, DecisionTree, Bagging, and MLP) to classify each attribute in the dataset (e.g., income is more than 50k or not, the loan should be given or not to user) using all other attributes as features. Each target attribute is processed to be a binary attribute, and we try to balance the two classes as much as possible. 
The quality of the learning task is represented by the average of all models across all attributes. The F1 score is reported for learning quality. Each model is trained using 70\% of the synthetic data, and the F1 score is evaluated using 30\% of the true data~\cite{ge2021kamino}. Larger values of this metric show better performance. For all our experiments, we run 3 times and report the mean with standard deviation. 
\fi 

\ifpaper
\begin{figure*} [htb]
    \centering
    \subfigure[Adult]{\includegraphics[width=.24\textwidth]{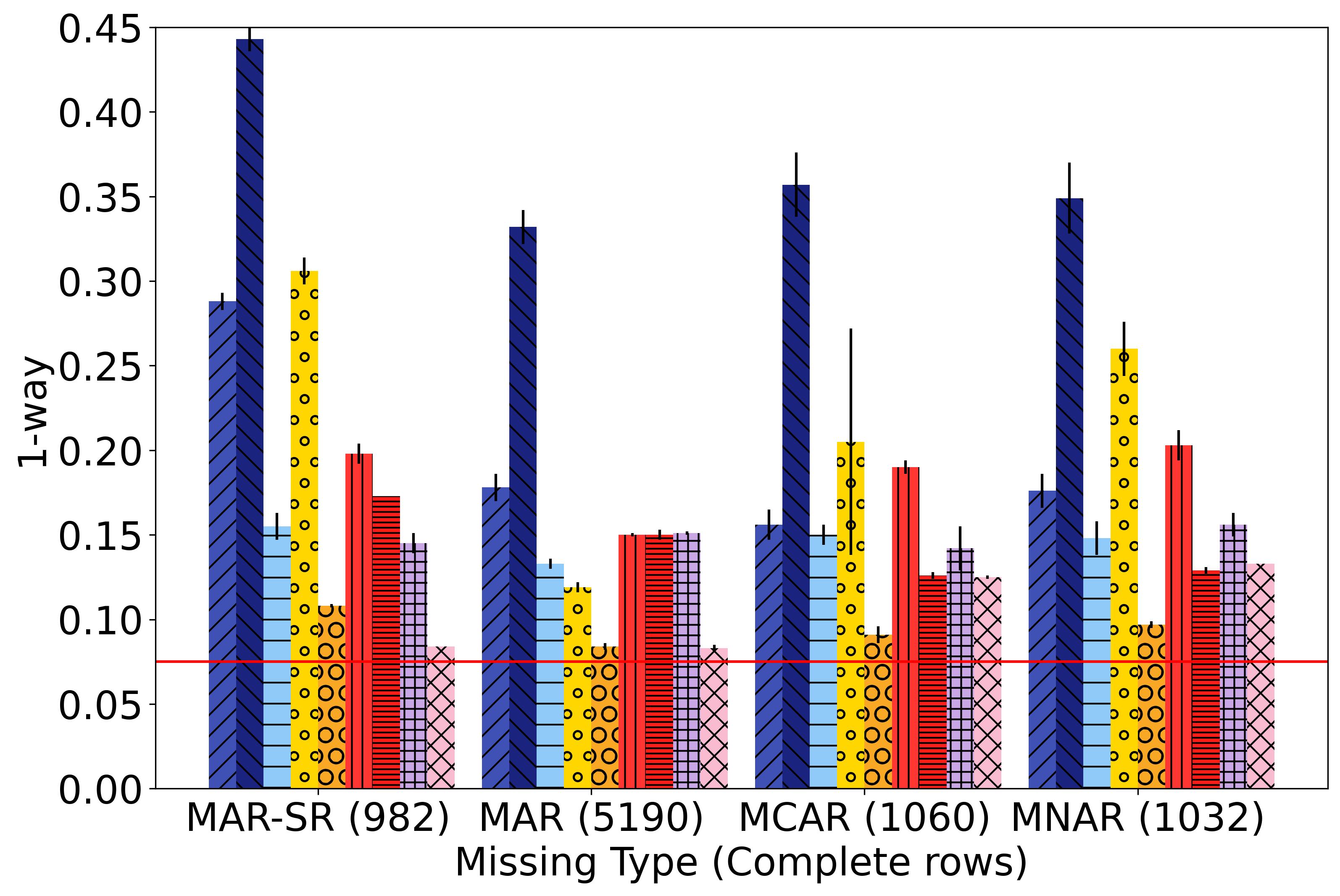}}
    \subfigure[Bank]{\includegraphics[width=.24\textwidth]{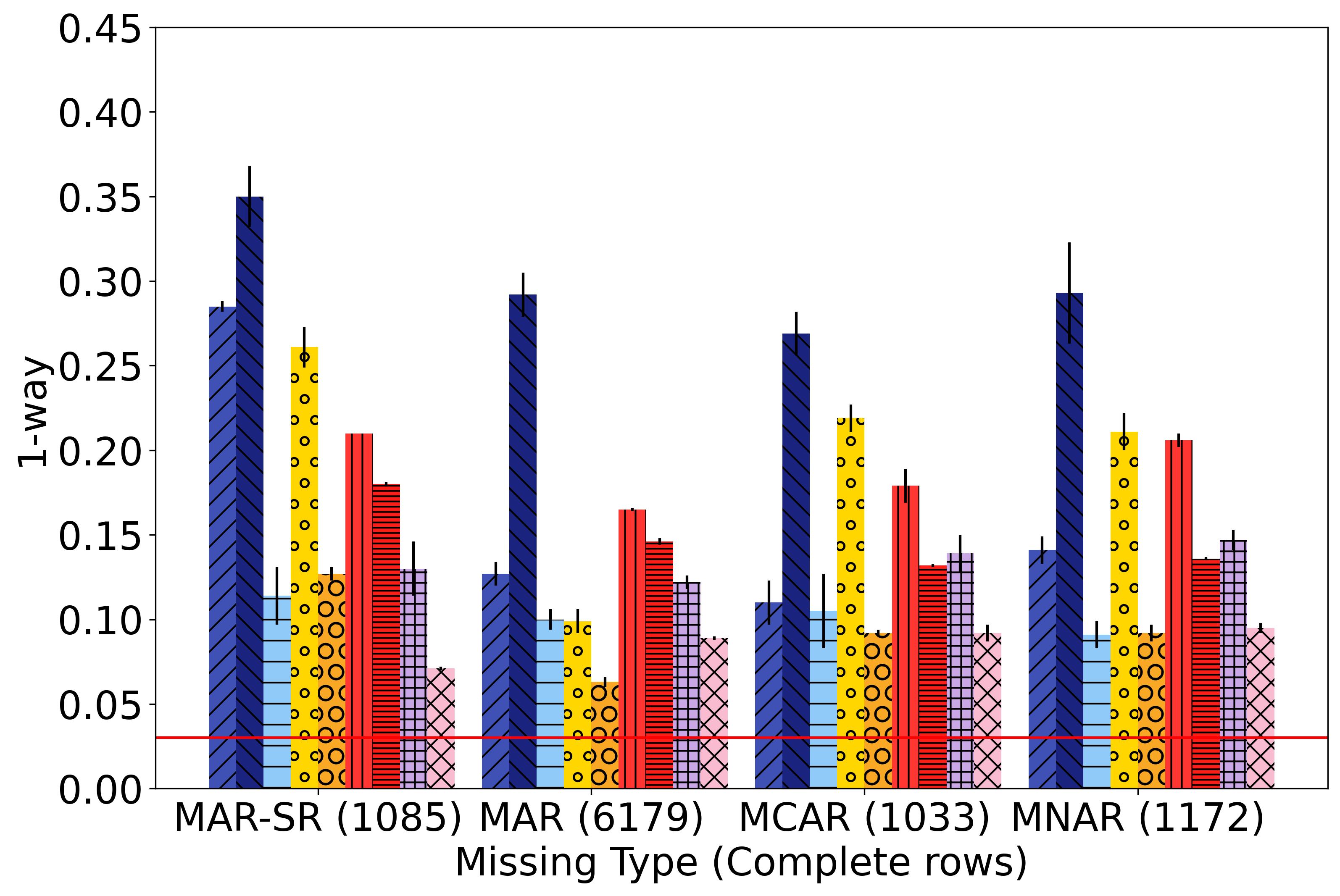}}
    \subfigure[BR2000]{\includegraphics[width=.24\textwidth]{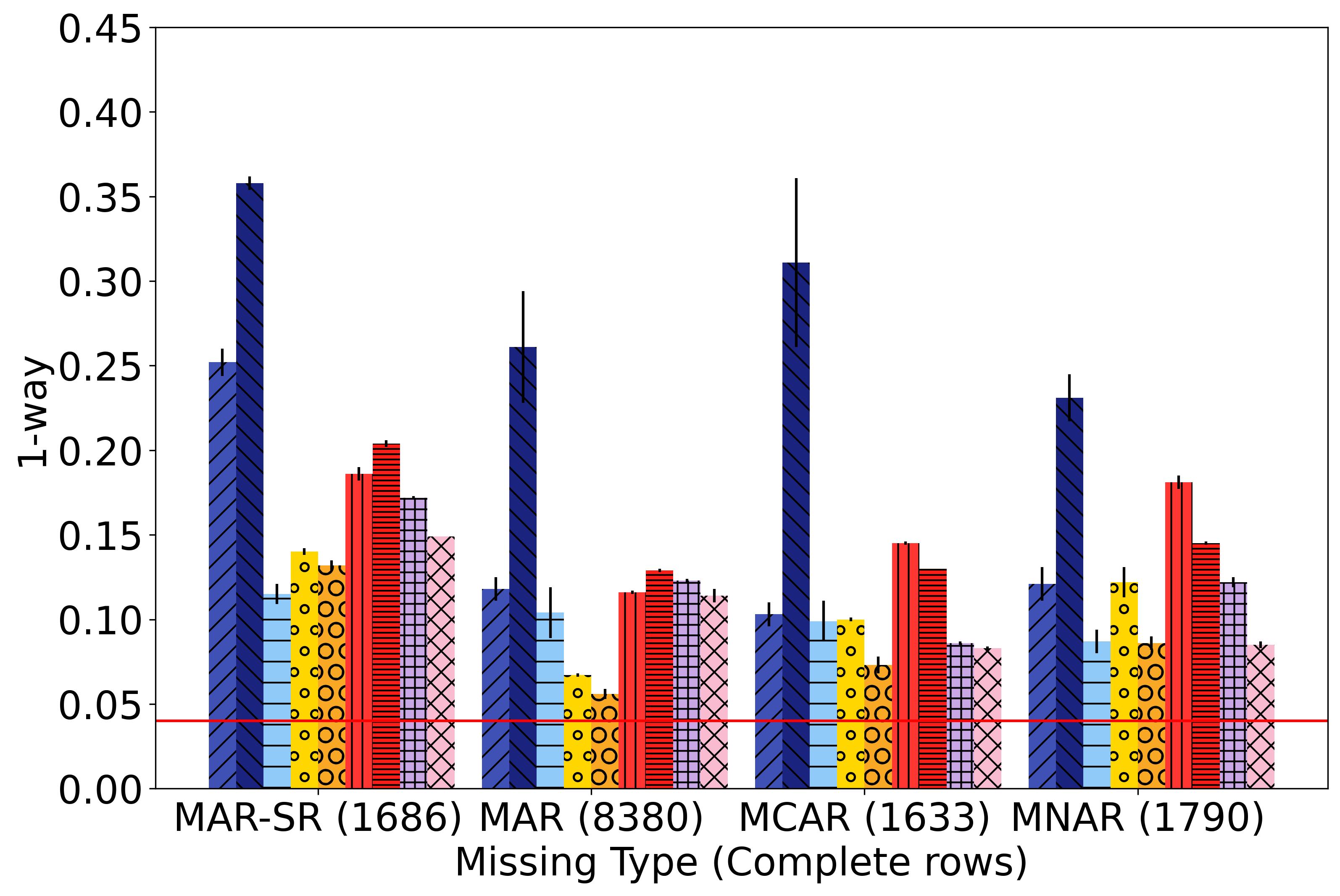}}
    \subfigure[National]{\includegraphics[width=.24\textwidth]{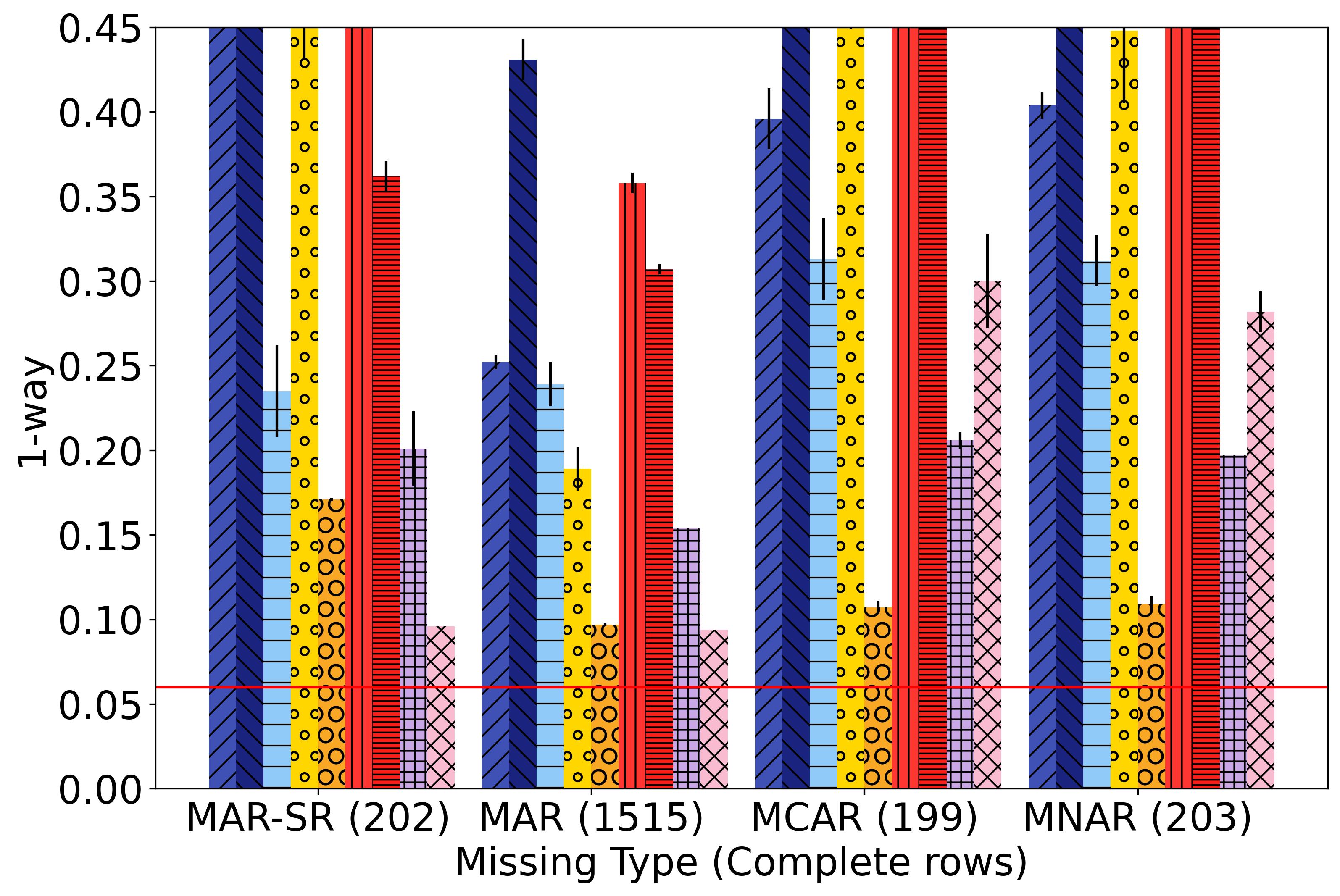}}
    \subfigure{\includegraphics[width=0.8\textwidth]{images/Exp3/Exp3-Legend.jpg}}
    \caption{\revision{Adaptive methods (DPMisGAN, PrivBayesE, KaminoI) vs their respective baselines using the complete row approach (DPautoGAN, PrivBayes, Kamino)} on different missing mechanisms. Red line denotes the no missing baseline.}
    \label{fig:exp3}
\end{figure*}
\else
\begin{figure*}
\centering
\subfigure[Adult 1-way ($\downarrow$)]{\includegraphics[width=.24\textwidth]{images/Exp2/adult-1-way.jpg}}
\subfigure[Bank 1-way ($\downarrow$)]{\includegraphics[width=.24\textwidth]{images/Exp2/bank-1-way.jpg}}
\subfigure[BR2000 1-way ($\downarrow$)]{\includegraphics[width=.24\textwidth]{images/Exp2/br2000-1-way.jpg}}
\subfigure[National 1-way ($\downarrow$)]{\includegraphics[width=.24\textwidth]{images/Exp2/national-1-way.jpg}}

\subfigure[Adult 2-way ($\downarrow$)]{\includegraphics[width=.24\textwidth]{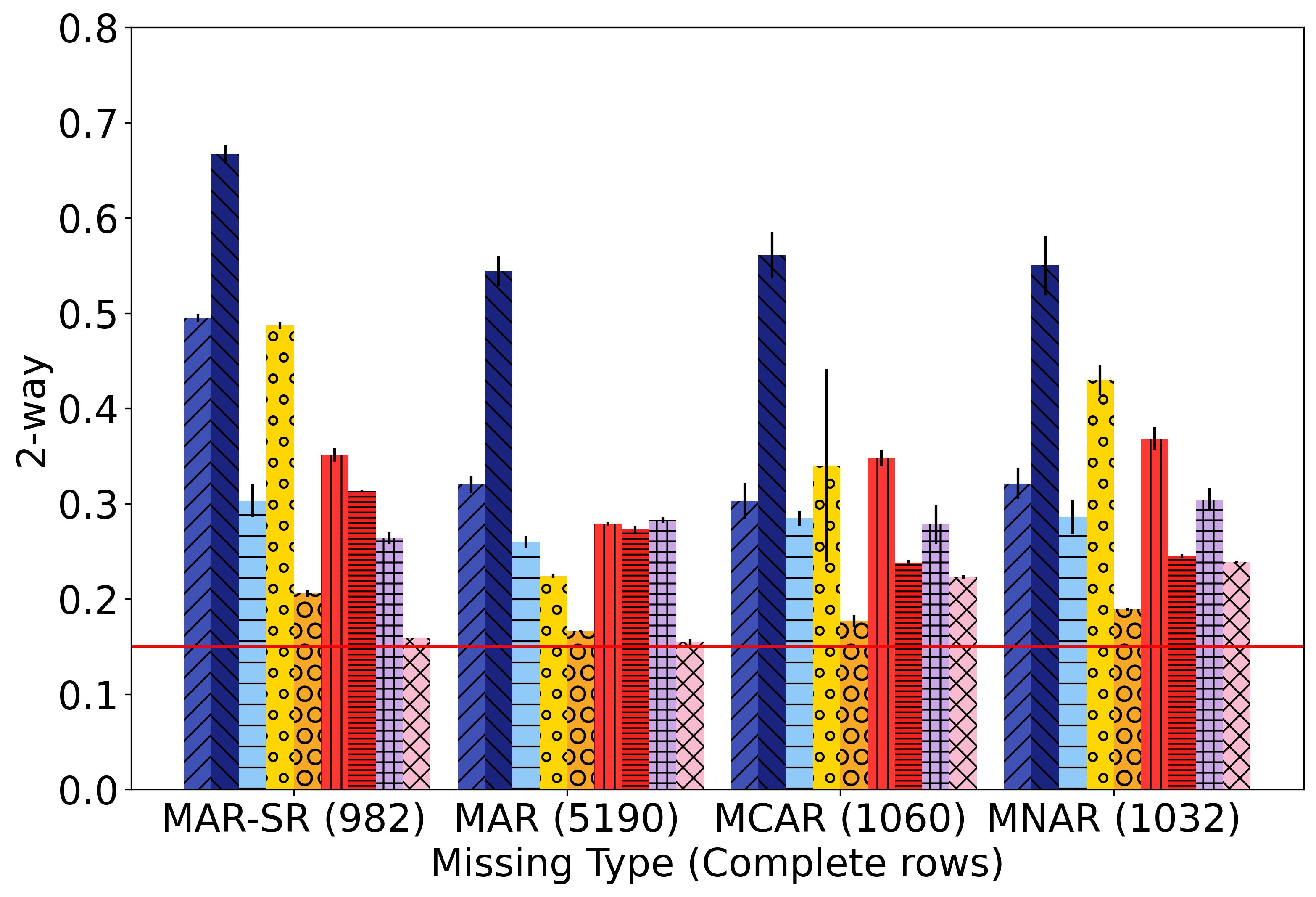}}
\subfigure[Bank 2-way ($\downarrow$)]{\includegraphics[width=.24\textwidth]{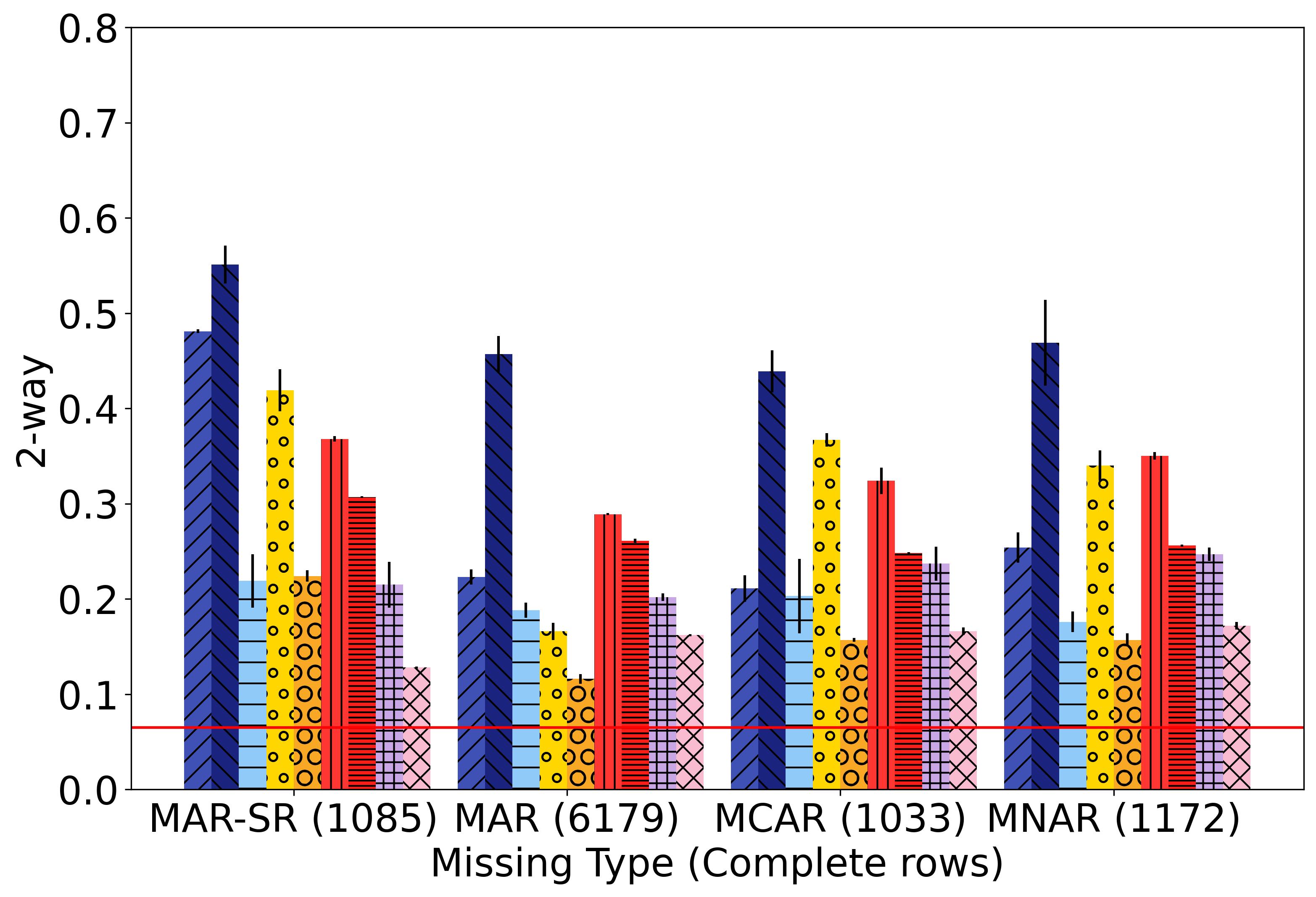}}
\subfigure[BR2000 2-way ($\downarrow$)]{\includegraphics[width=.24\textwidth]{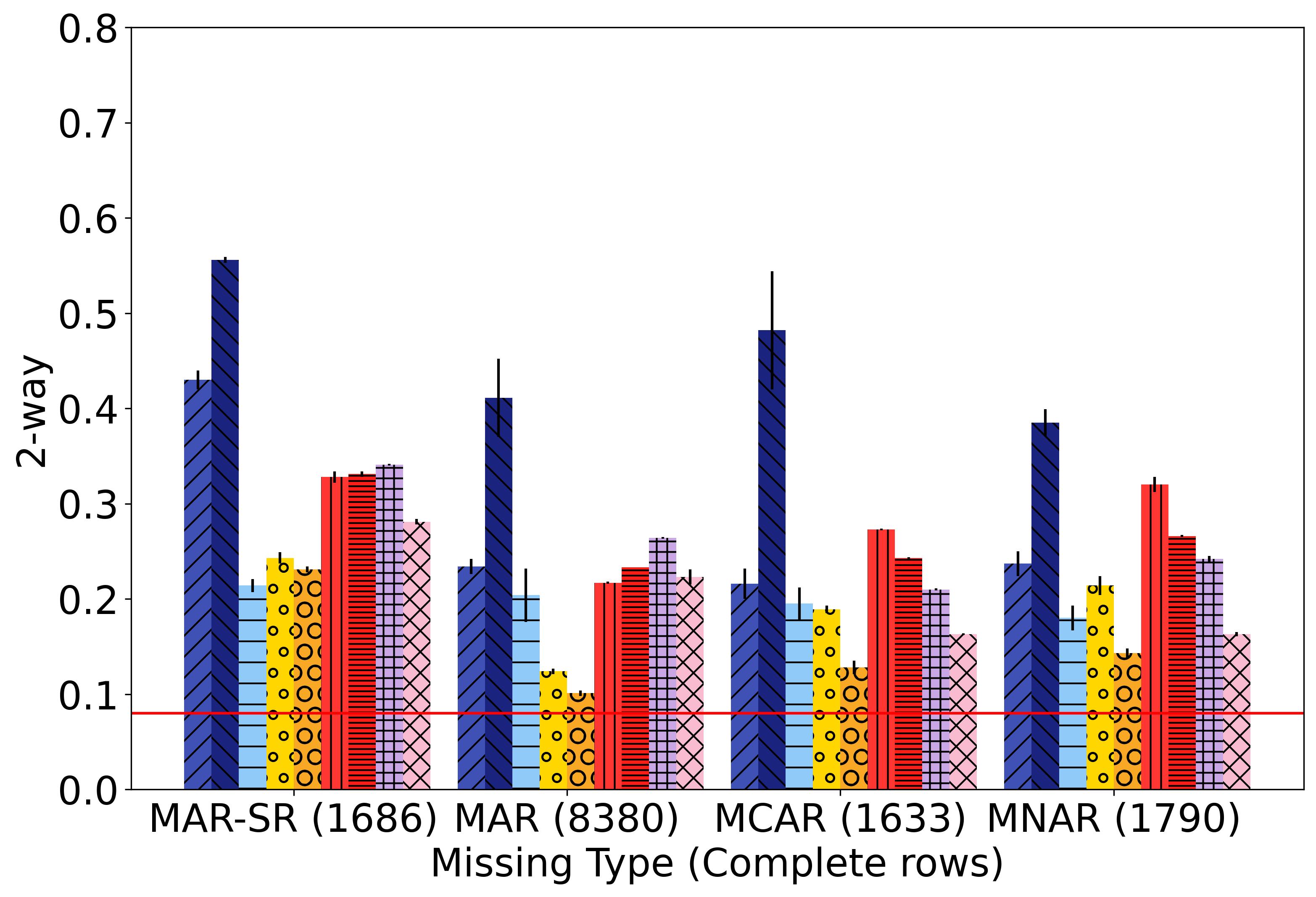}}
\subfigure[National 2-way ($\downarrow$)]{\includegraphics[width=.24\textwidth]{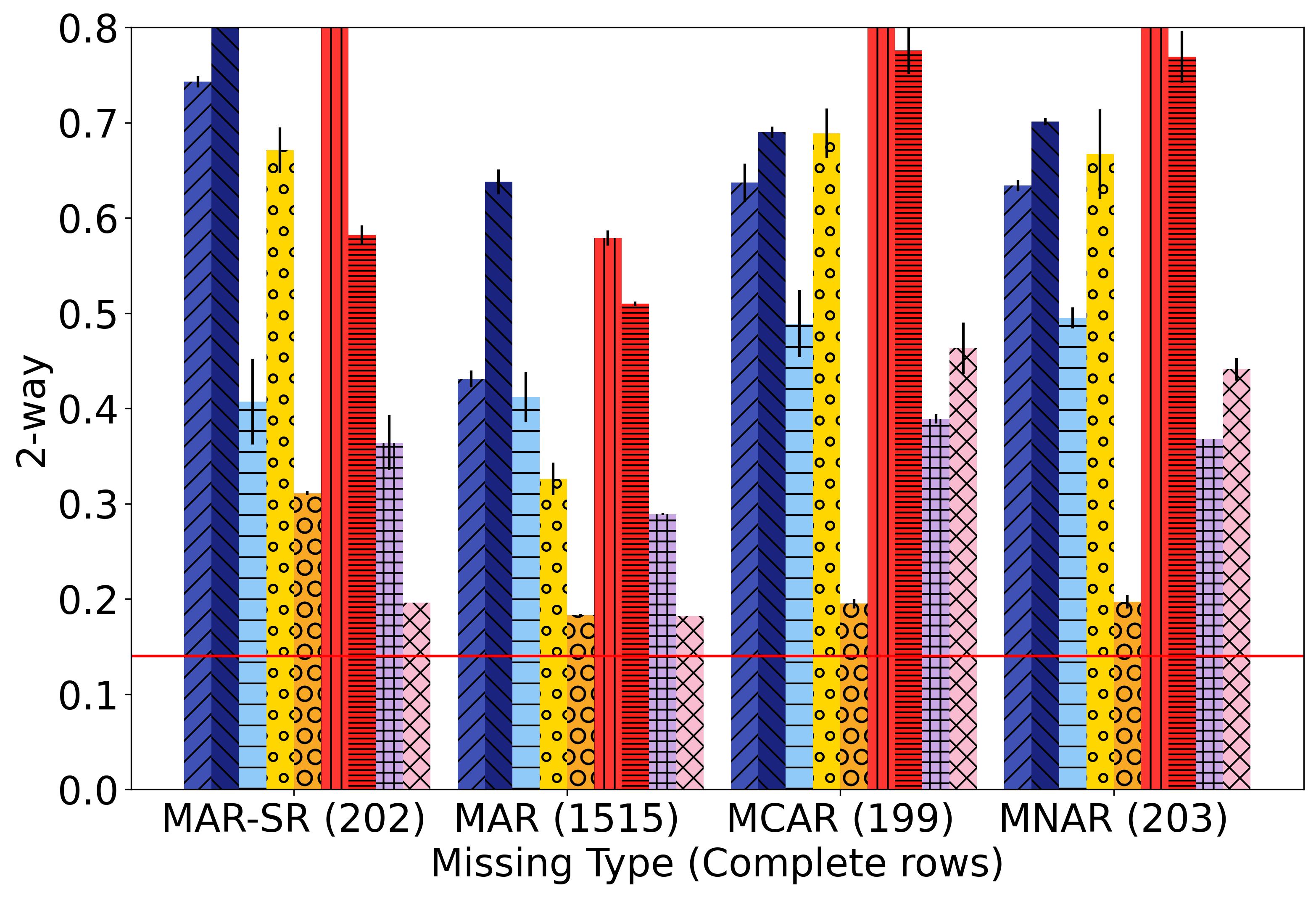}}

\subfigure[Adult F1-score ($\uparrow$)]{\includegraphics[width=.24\textwidth]{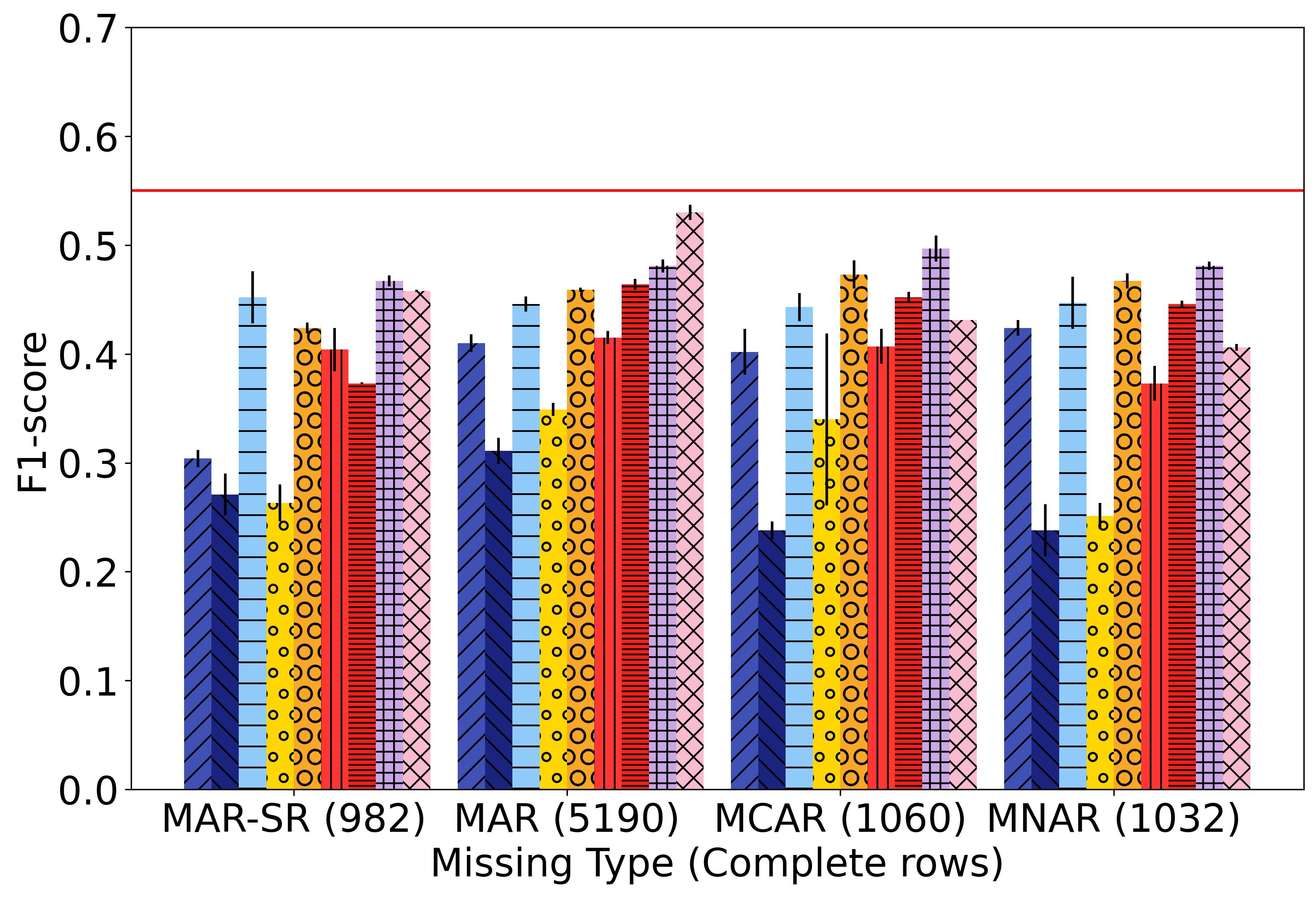}}
\subfigure[Bank F1-score ($\uparrow$)]{\includegraphics[width=.24\textwidth]{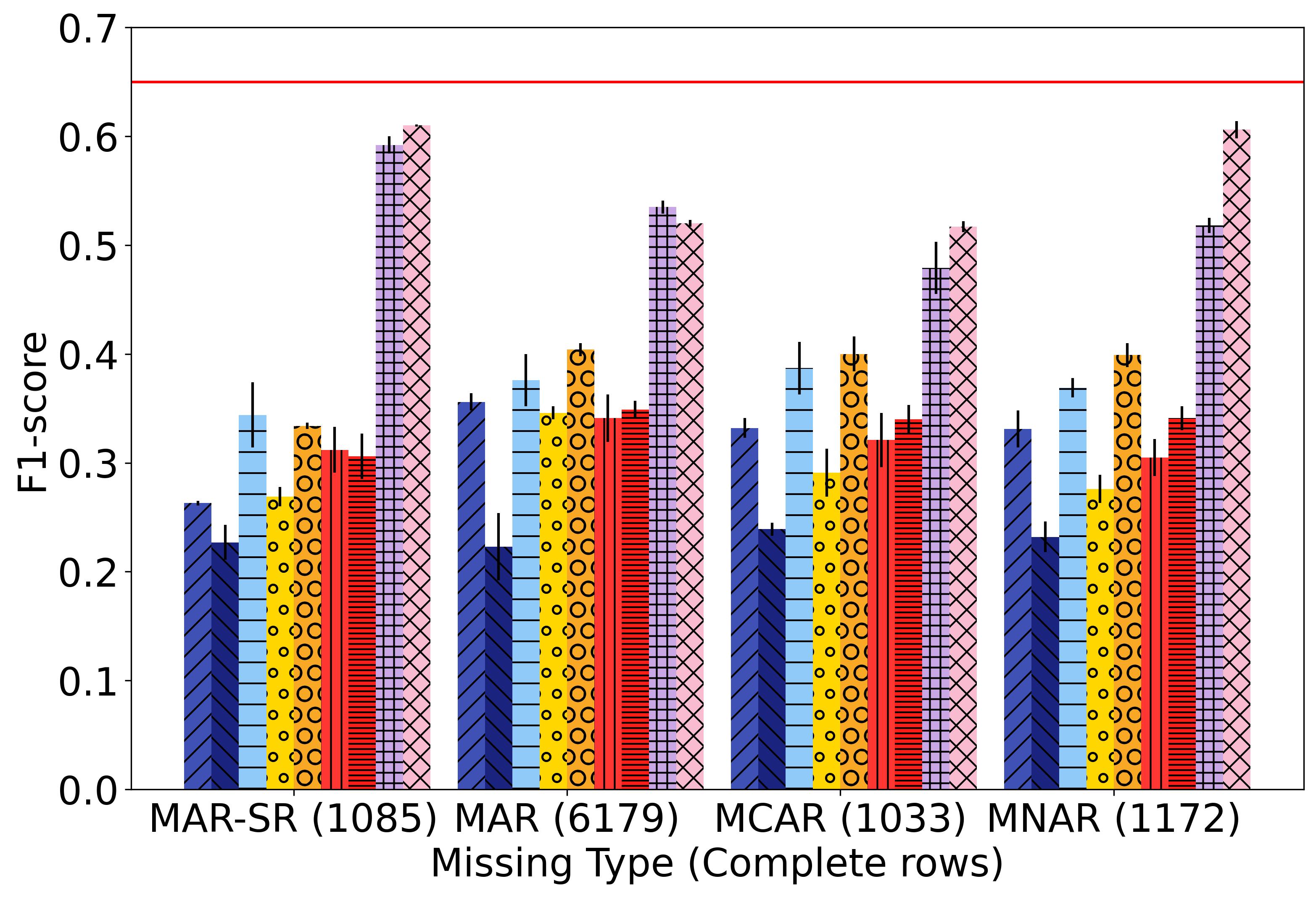}}
\subfigure[BR2000 F1-score ($\uparrow$)]{\includegraphics[width=.24\textwidth]{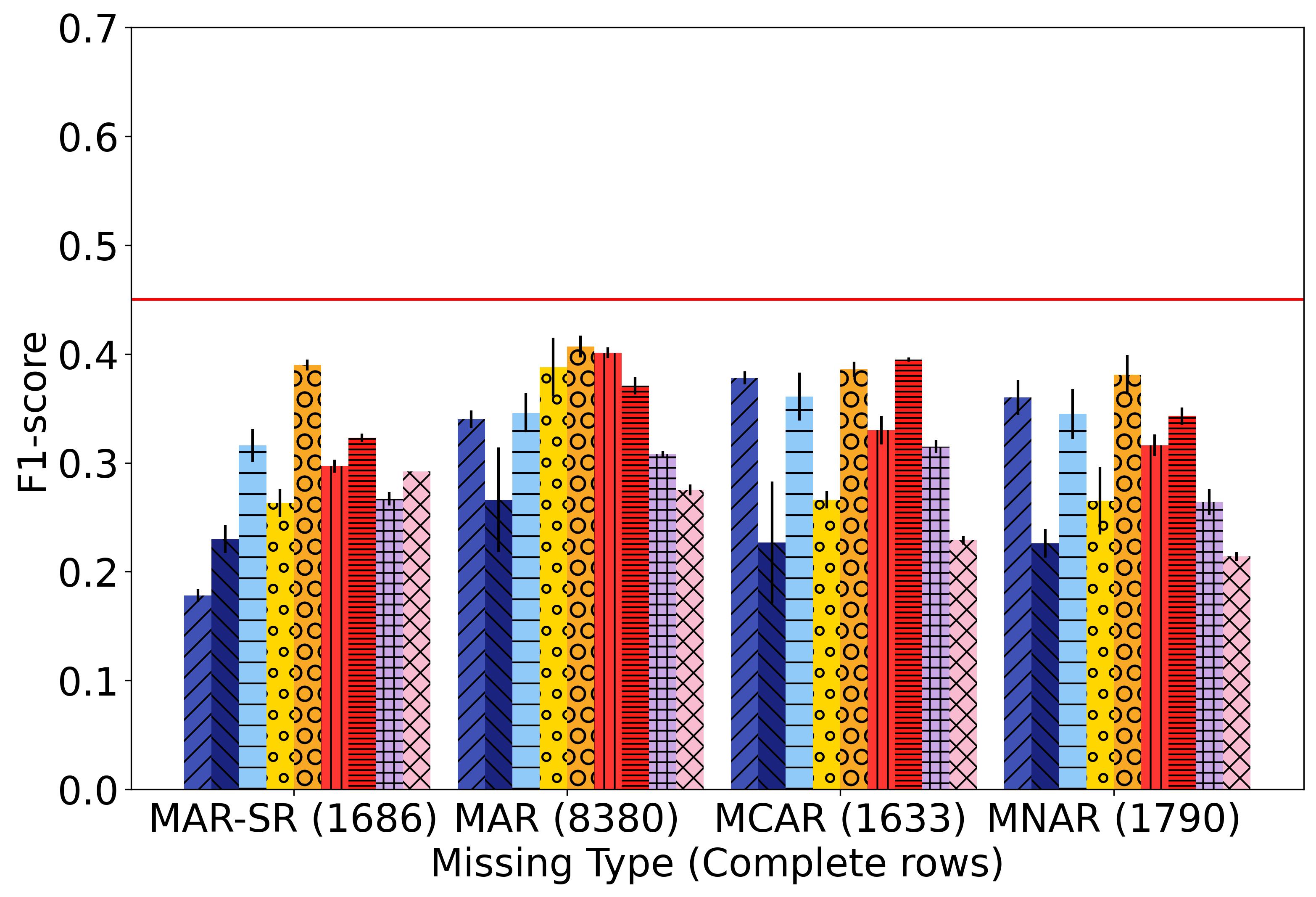}}
\subfigure[National F1-score ($\uparrow$)]{\includegraphics[width=.24\textwidth]{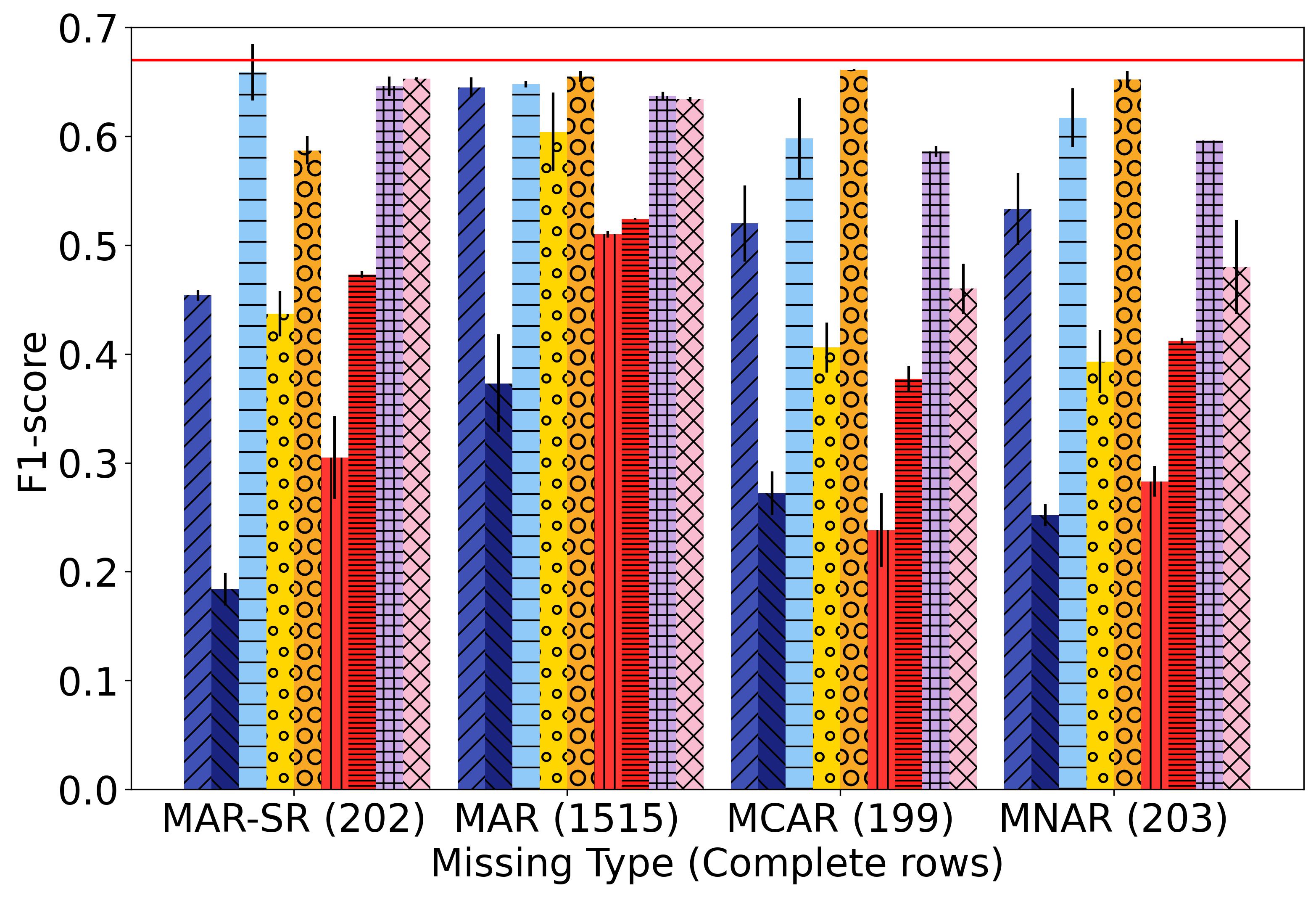}}

\subfigure{\includegraphics[width=\textwidth]{images/Exp3/Exp3-Legend.jpg}}
\caption{Adaptive methods vs non-adaptive methods on different missing mechanisms. The red line denotes the no missing baseline.}
\label{fig:exp3}
\end{figure*}
\fi

\subsection{Experimental Findings}
\revision{We evaluate the impact of missing values on DP synthetic data generation algorithms with different missing mechanisms and varying missing data from 1\% to 30\%. The quality of synthetic data generated by all algorithms degrades drastically, particularly when having >10\% of missing data. 
\ifpaper
We include the benchmark results in the full paper~\cite{full_paper} and show the following experiments for 10\%-30\% of missing data. 
We note that DPCTGAN has consistently poorer utility than DPautoGAN, and AIM shares similar observations as PrivBayes. Hence, we move their results to the full paper~\cite{full_paper}.} 
\else
\fi

\ifpaper
\else
\subsubsection{Benchmark Existing DP Methods on Missing Data}\label{sec:exp_benchmark}
We demonstrate the impact of missing values on DP synthetic data generation algorithms. 
Figure~\ref{fig:exp1} shows the performance of the baselines on the Adult dataset with varying levels of missing completely at random (MCAR) data (x-axis). Our results indicate that missing data up to 5\% minimally affects the quality of the synthetic data with 1-way and 2-way metrics experiencing 5-19\% impact, and F1-score experiencing 1-13\% impact. However, as missing data increases, the quality of synthetic data generated by all algorithms degrades drastically. With 20\% missing data, the 1-way metric is affected by 14-190\%, the 2-way by 18-147\%, and the F1-score is impacted by 6-28\%.
With high amounts of missing data ($ \geq 10\%)$, PrivBayes and DPCTGAN exhibit the most decrease in utility. However, Kamino and DPAutoGAN show a consistent level of performance even with more missing data indicating a high degree of stability in the output. The stability of these methods can be attributed to the fact that these methods tend to extract more information from the available data. Kamino learns the functional dependencies of the dataset and tries to preserve them while generating the synthetic data and DPAutoGAN trains an autoencoder as a preprocessing step to learn the low-dimensional statistics of the data. 
\fi
\vspace{-4mm}
\revision{
\subsubsection{Evaluate different strategies to deal with missing data}
\label{sec:expt_baselines}
For each existing DP algorithm for generating synthetic data, we compare its respective baseline methods and its adaptive recourse approach. Figure~\ref{fig:exp5} shows the performance of these methods across all four datasets, employing 20\% MCAR missing values. The red line denotes the baseline approach's performance without missing values.  For mean and Kamino imputation, we allocate three splits of the privacy budget (25\%, 50\%, and 75\%) for imputation and plot the average and standard deviation across all splits, reserving the rest for generating the synthetic dataset. We observe that no single imputation strategy consistently outperforms others, with the adaptive recourse approach generally yielding the best results. Hence, allowing learning and imputation to happen together is crucial. Additionally, the complete row approach often performs second best, serving as our baseline for subsequent experiments.
} 

\subsubsection{Evaluate Adaptive Recourse Approach}\label{sec:exp_adaptive_recourse}
We evaluate adaptive recourse approaches at various experimental configurations. 

\stitle{Varying missing percentage.}
In this experiment, we compare the best baseline from the previous experiment that uses the complete row approach with its corresponding adaptive recourse approach from Section~\ref{sec:adaptive}. We repeat our experiment on four different datasets with varying amount of MCAR data at two privacy levels ($\epsilon = 1, 3$). 
\ifpaper
We defer the results for $\epsilon =3$ to the full paper~\cite{full_paper} which have qualitatively similar results.
\else
\begin{figure*}
\centering
\subfigure[Adult 1-way ($\downarrow$)]{\includegraphics[width=.24\textwidth]{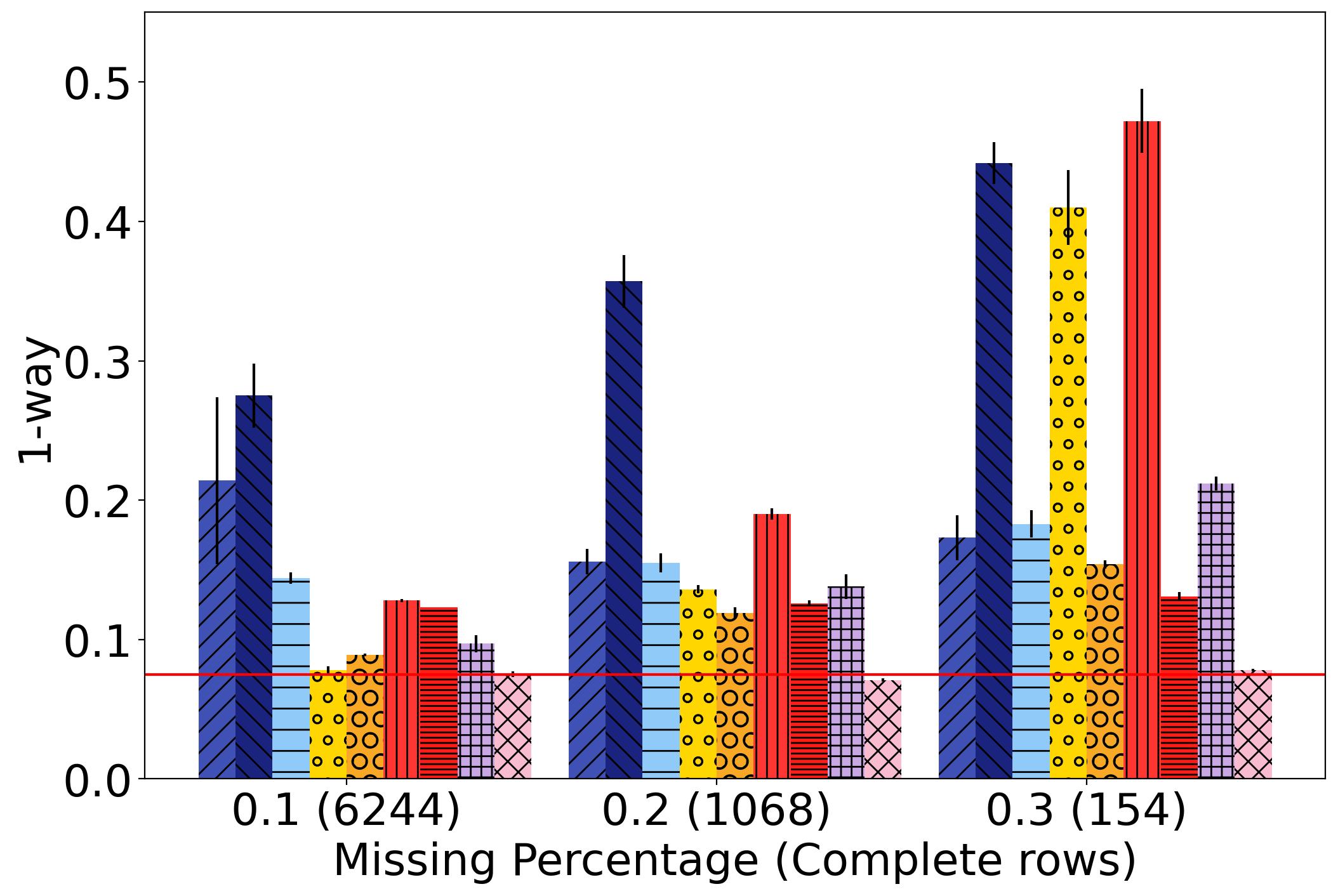}}
\subfigure[Bank 1-way ($\downarrow$)]{\includegraphics[width=.24\textwidth]{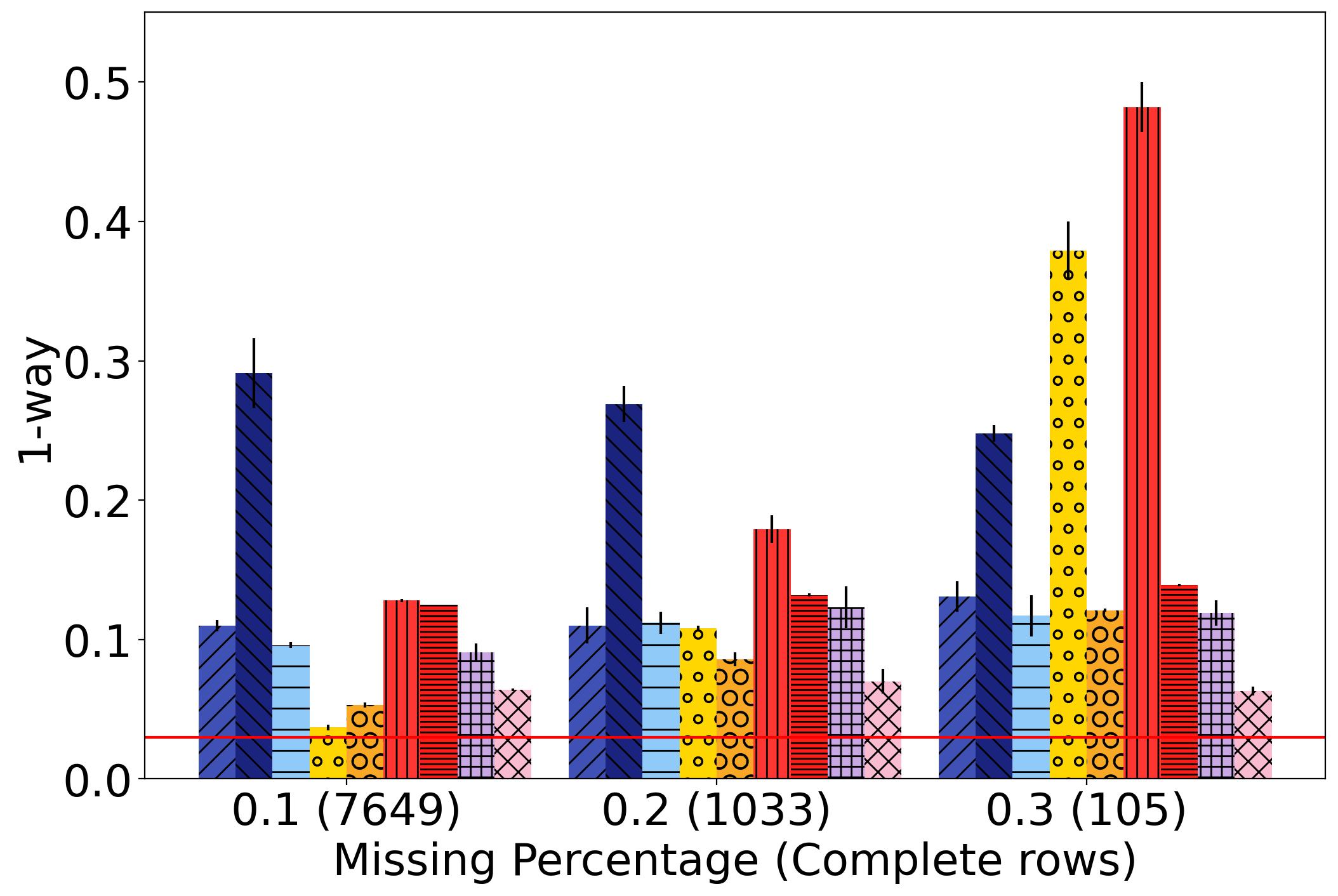}}
\subfigure[BR2000 1-way ($\downarrow$)]{\includegraphics[width=.24\textwidth]{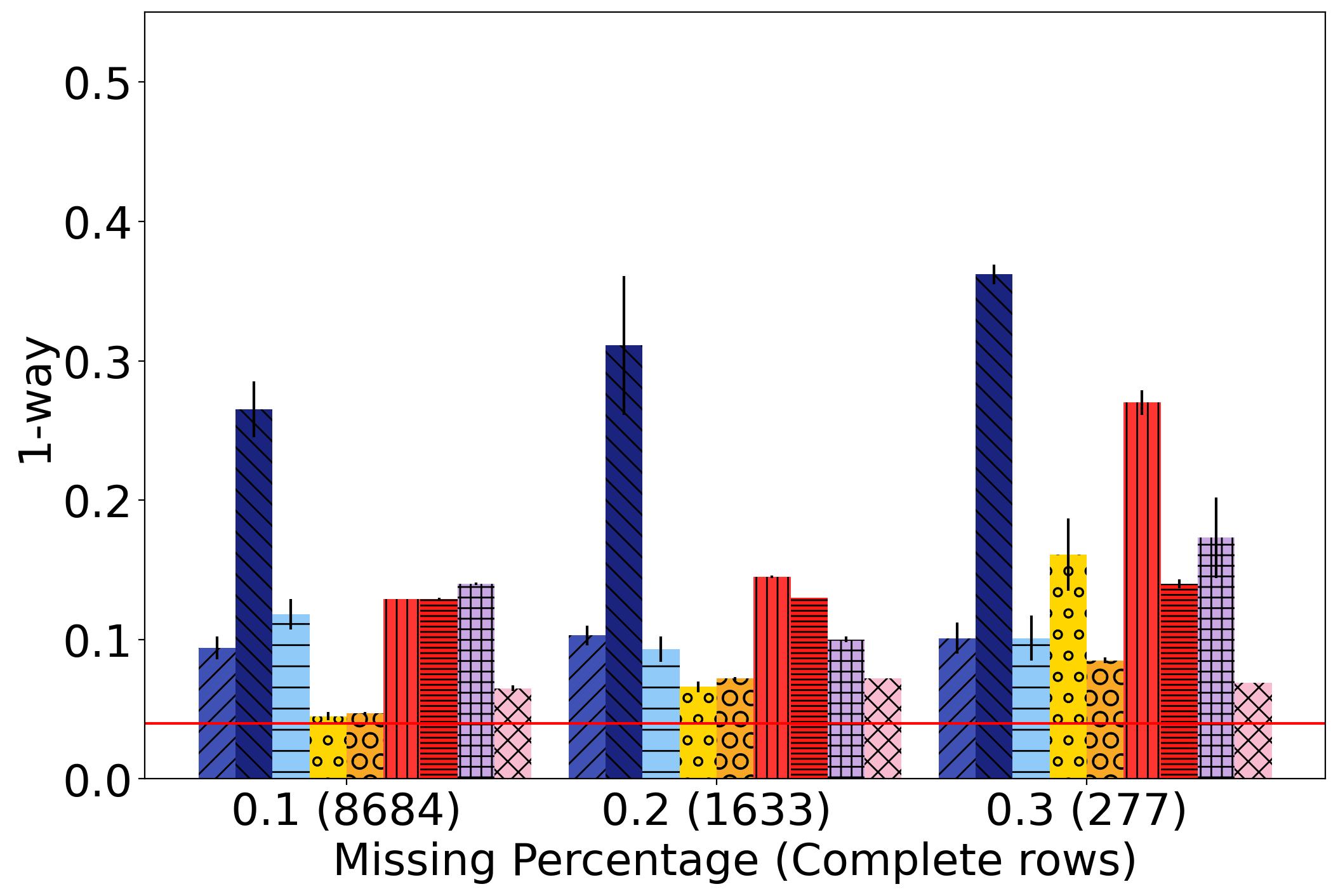}}
\subfigure[National 1-way ($\downarrow$)]{\includegraphics[width=.24\textwidth]{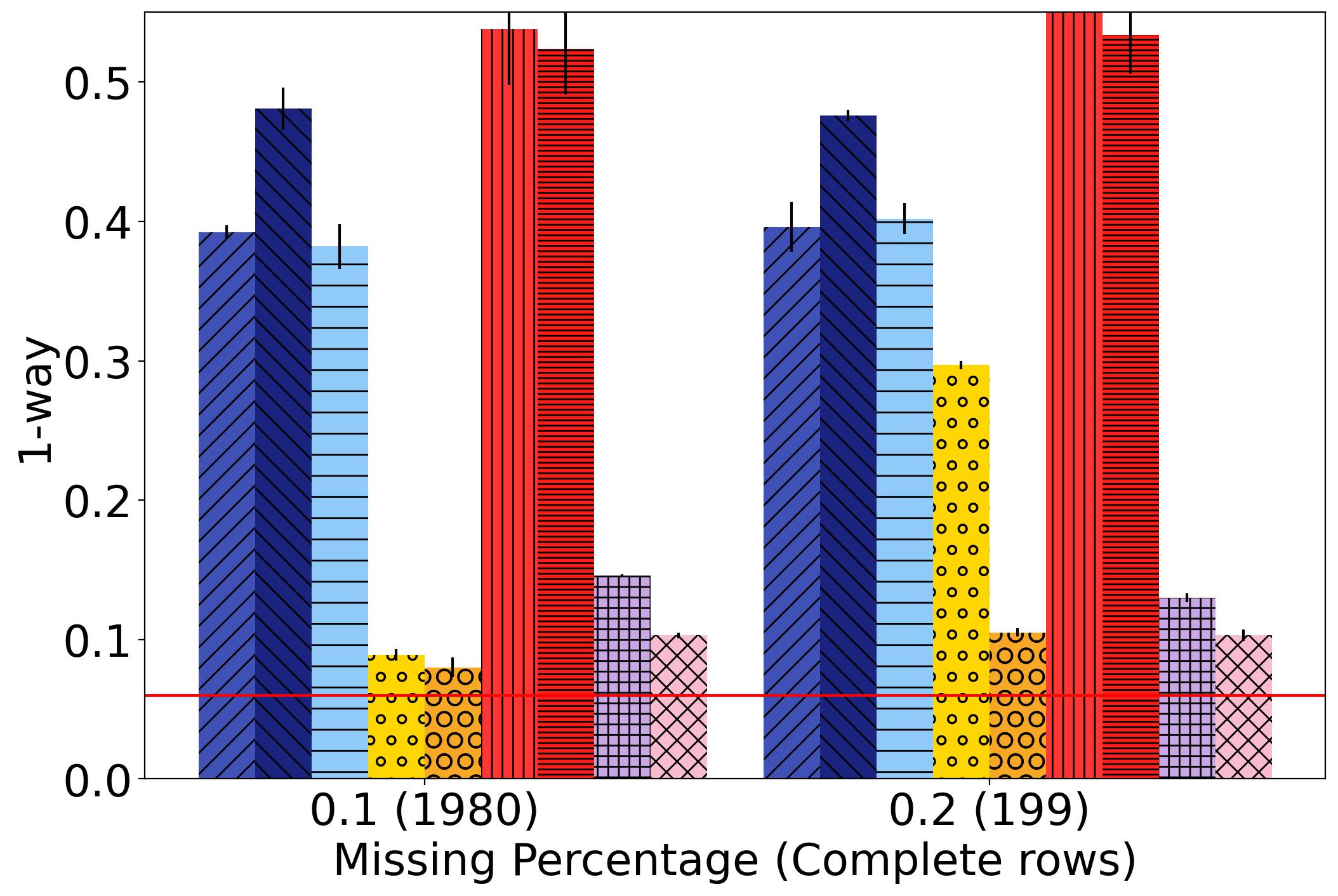}}

\subfigure[Adult 2-way ($\downarrow$)]{\includegraphics[width=.24\textwidth]{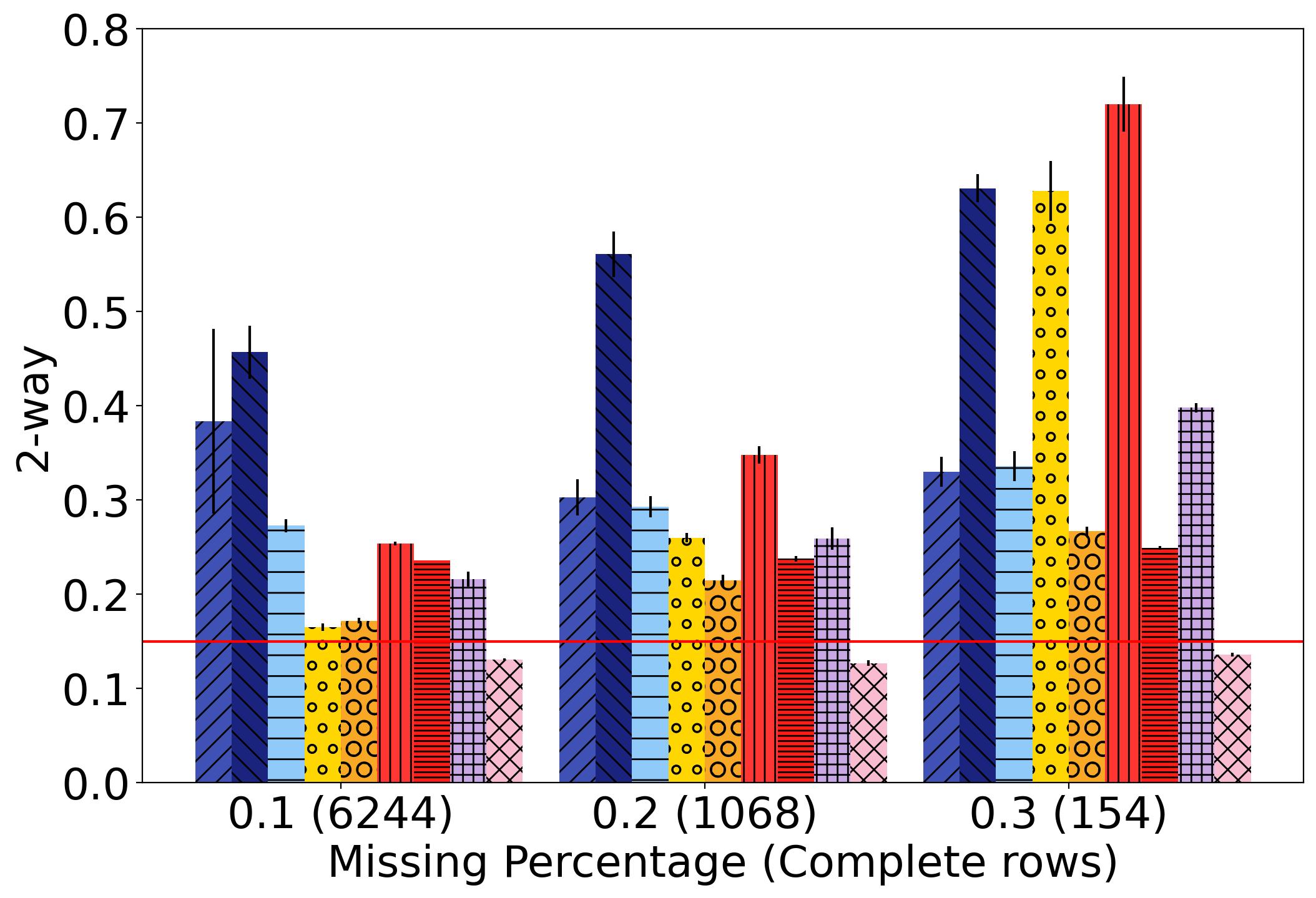}}
\subfigure[Bank 2-way ($\downarrow$)]{\includegraphics[width=.24\textwidth]{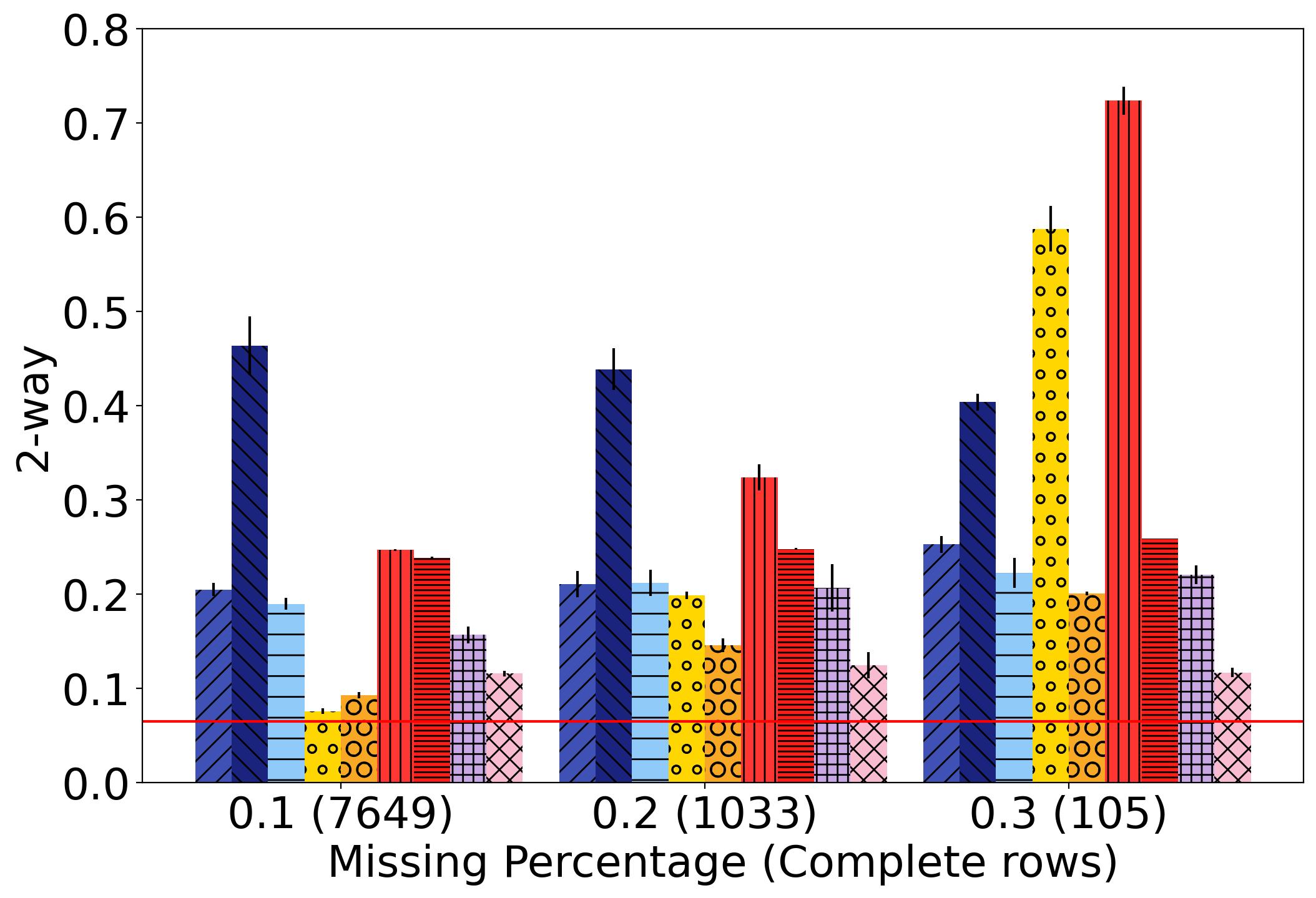}}
\subfigure[BR2000 2-way ($\downarrow$)]{\includegraphics[width=.24\textwidth]{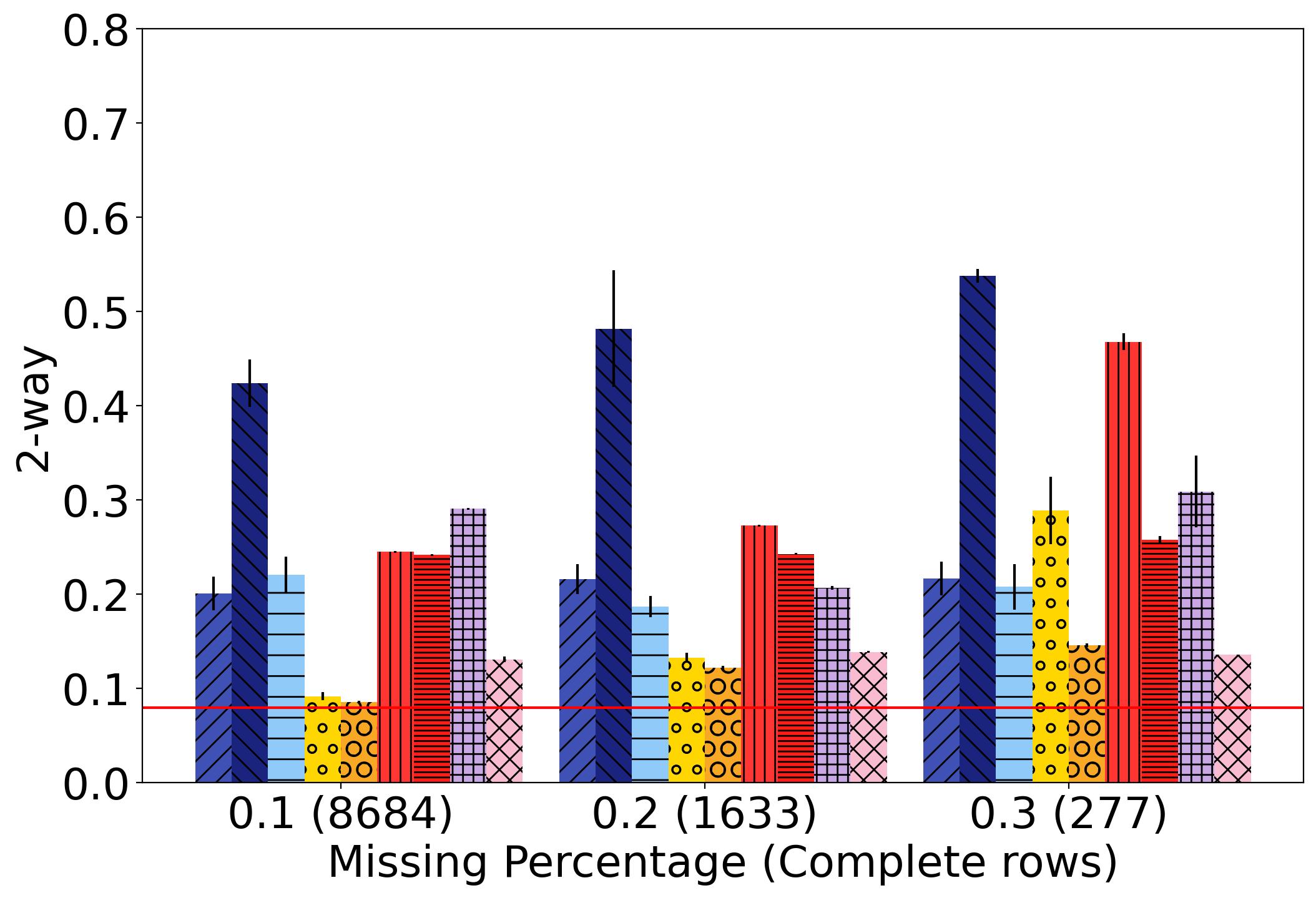}}
\subfigure[National 2-way ($\downarrow$)]{\includegraphics[width=.24\textwidth]{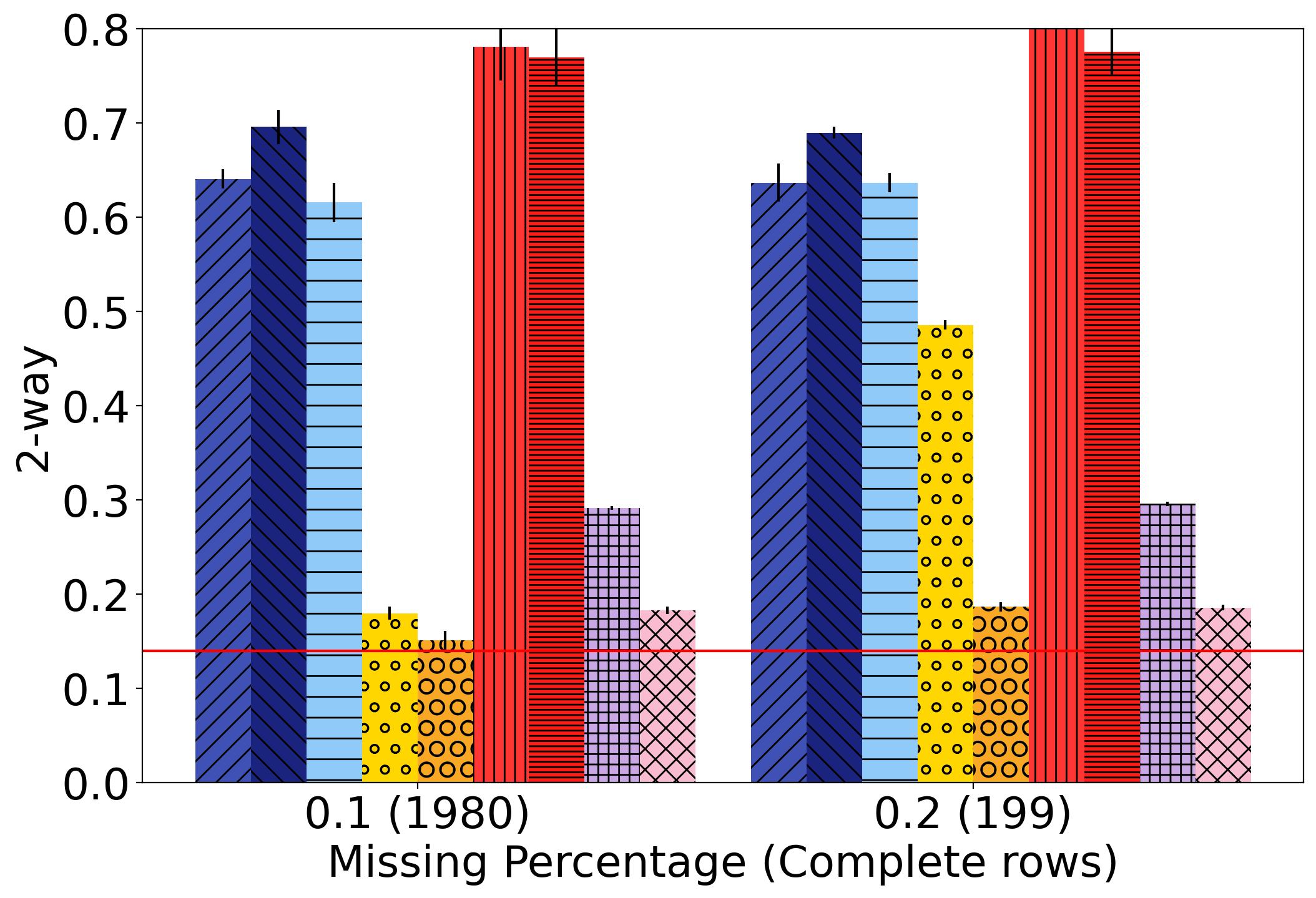}}

\subfigure[Adult F1-score ($\uparrow$)]{\includegraphics[width=.24\textwidth]{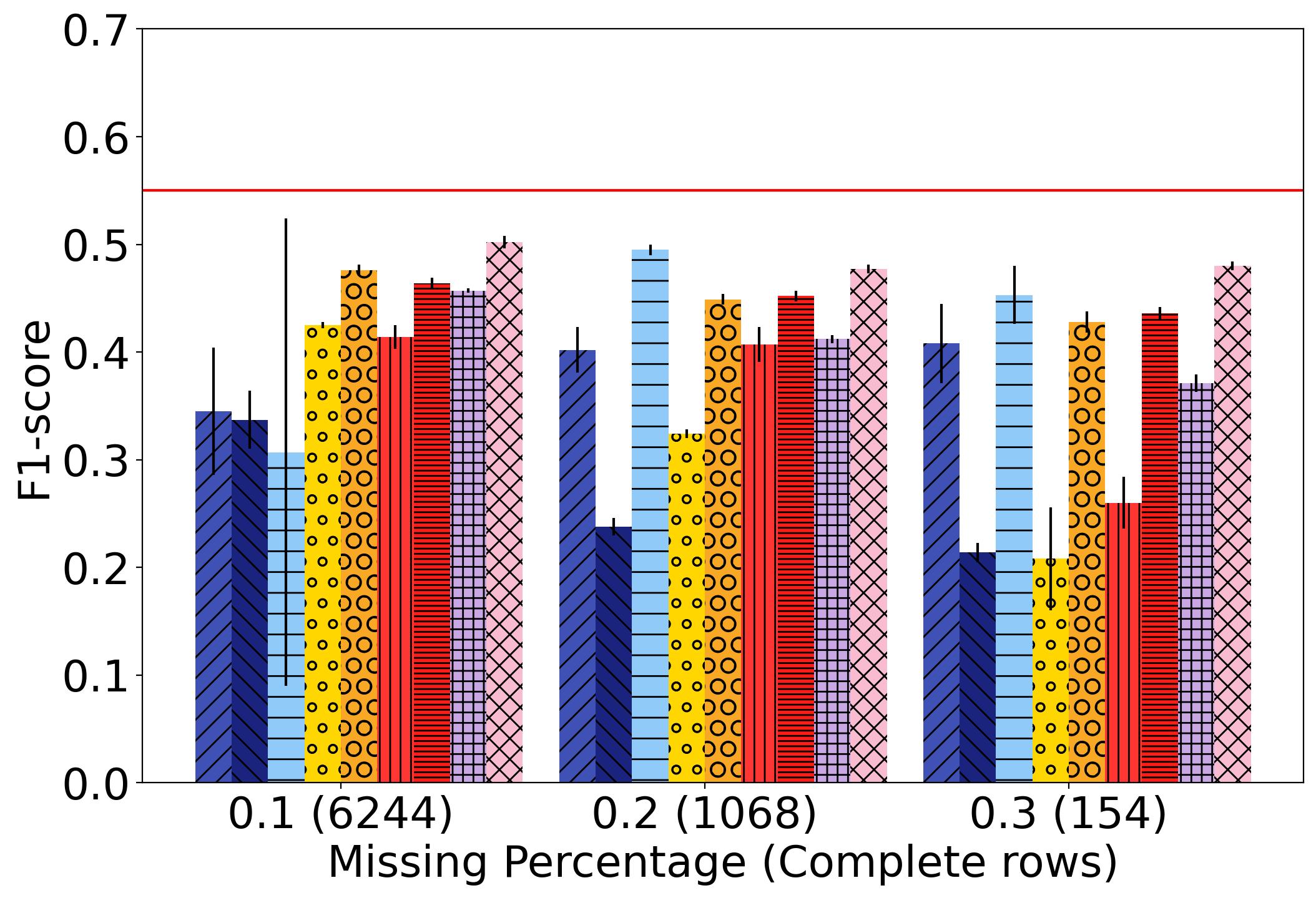}}
\subfigure[Bank F1-score ($\uparrow$)]{\includegraphics[width=.24\textwidth]{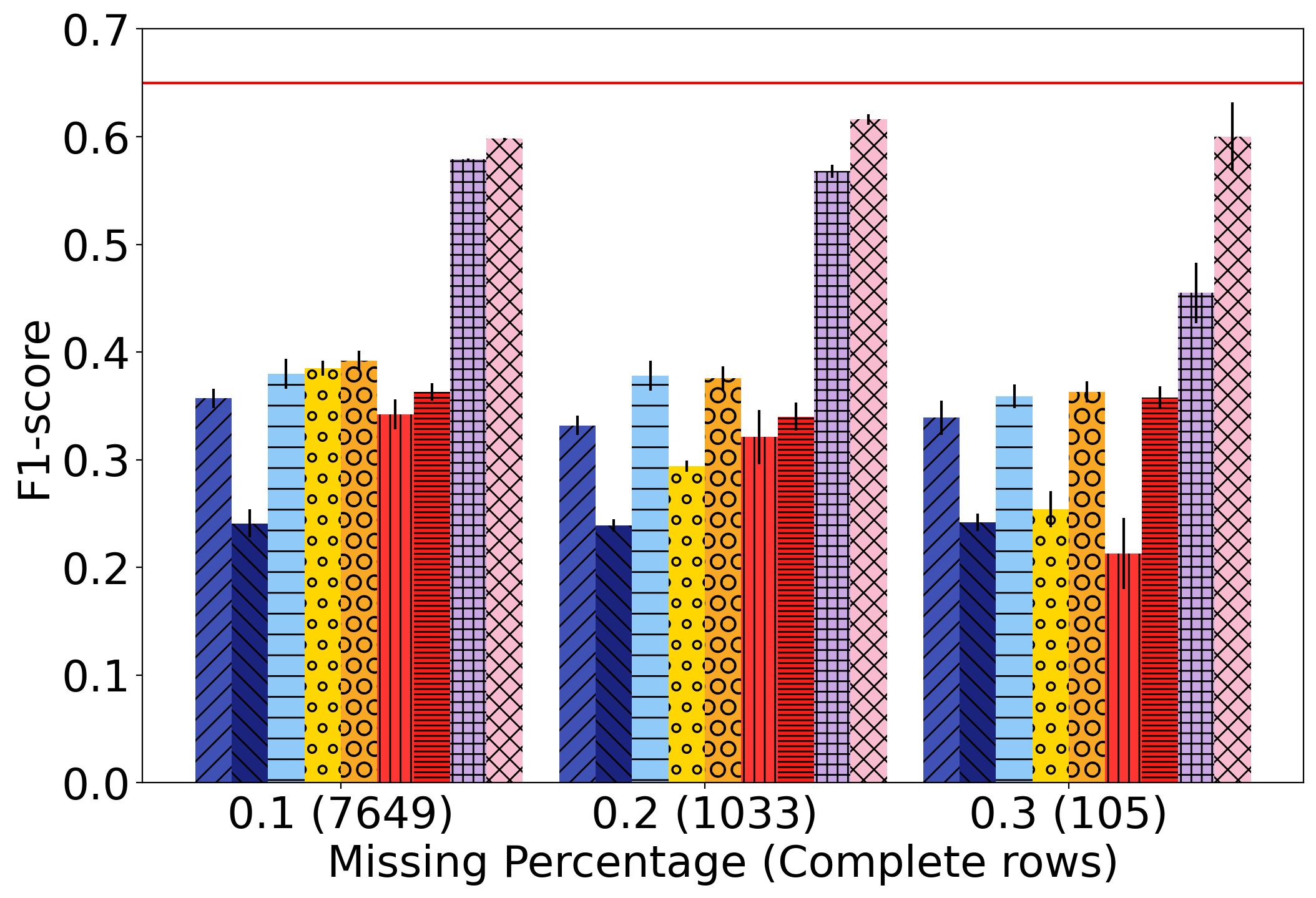}}
\subfigure[BR2000 F1-score ($\uparrow$)]{\includegraphics[width=.24\textwidth]{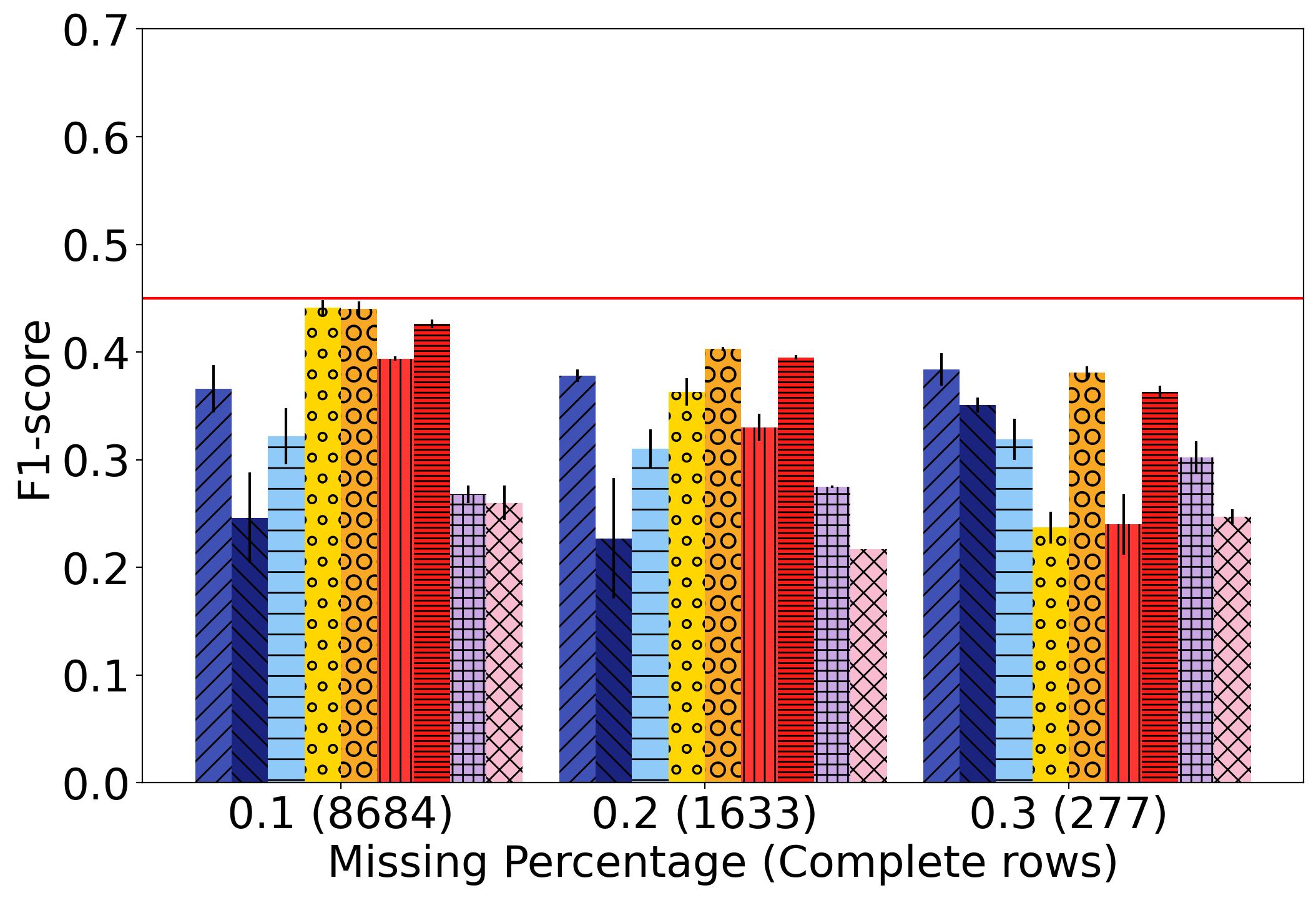}}
\subfigure[National F1-score ($\uparrow$)]{\includegraphics[width=.24\textwidth]{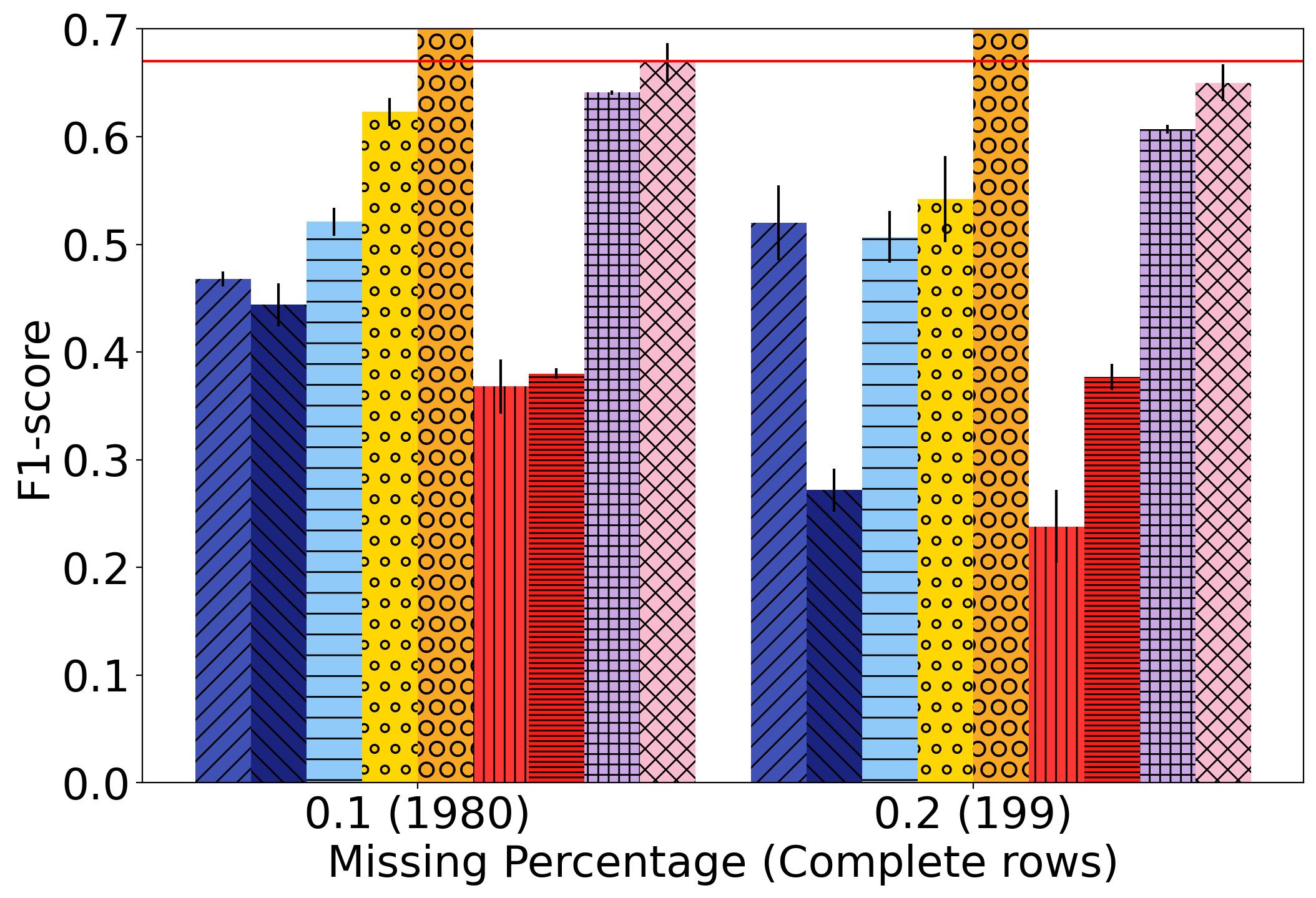}}

\subfigure{\includegraphics[width=\textwidth]{images/Exp3/Exp3-Legend.jpg}}
\label{fig:exp2eps3}
\caption{Baselines vs Adaptive methods at $\epsilon=3$. The adaptive methods perform better than their respective baselines.}
\end{figure*}

\fi
In Figure~\ref{fig:exp2}, we plot the algorithms in different shades of color depending on their type. \revision{We observe in general that adaptive recourse approaches (DPMisGAN, PrivBayesE, AIME, and KaminoI) result in significantly better quality synthetic data compared to their baseline with the complete row only approach (DPautoGAN, DPCTGAN, PrivBayes, AIM and Kamino). Across all datasets, the 1-way scores are improved by up to 68\%, 2-way by up to 66\% and F1-scores of up to 24\%. Furthermore, we observe that the adaptive methods often achieve the same utility as the no missing baseline with 10\% missing data (e.g., top left subfigure Adult 2-way for KaminoI and PrivBayesE at 10\% missing data).}

\stitle{Varying missing mechanisms.} In Figure~\ref{fig:exp3}, we repeat our experiment for all $4$ datasets at 10\% missing data for all missing mechanisms.  The adaptive recouse approaches(DPMisGAN, PrivBayesE, AIME and KaminoI) beat their non-adaptive complete-row baselines across all missing mechanisms. 
As missing values are added only to half of the attributes, missing at random (MAR) has more complete rows (the number in brackets on the x-axis labels) as compared to the other mechanisms. It is interesting to note that if we increase the missing percentage for MAR and plot it with same rows (MAR-SR), the algorithms start performing poorly. Hence, we make the conclusion that the number of complete of rows makes a more vital impact compared to the missing mechanism itself.

\ifpaper
\begin{figure}
    \centering
    \includegraphics[width=.49\linewidth]{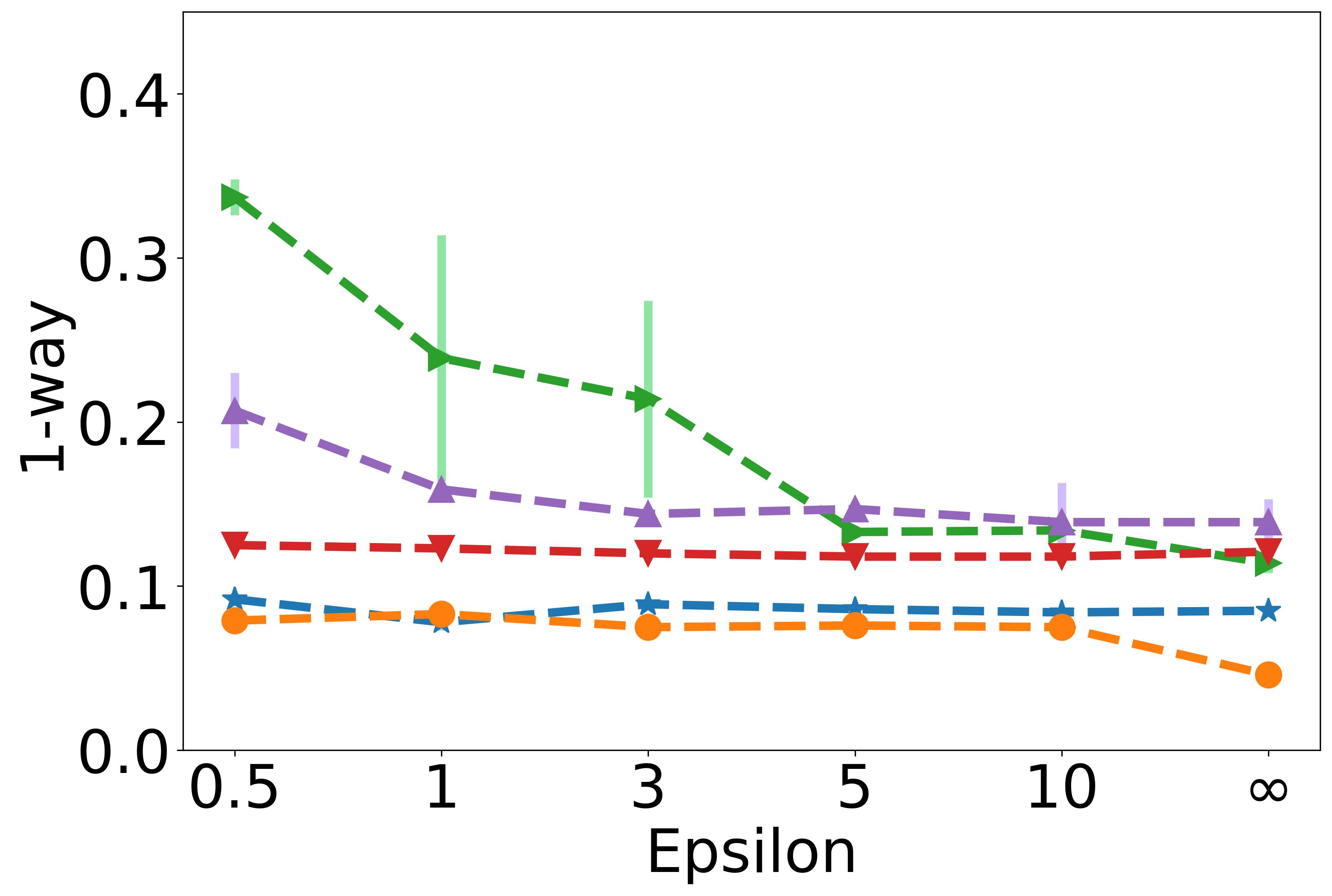}
    \includegraphics[width=.49\linewidth]{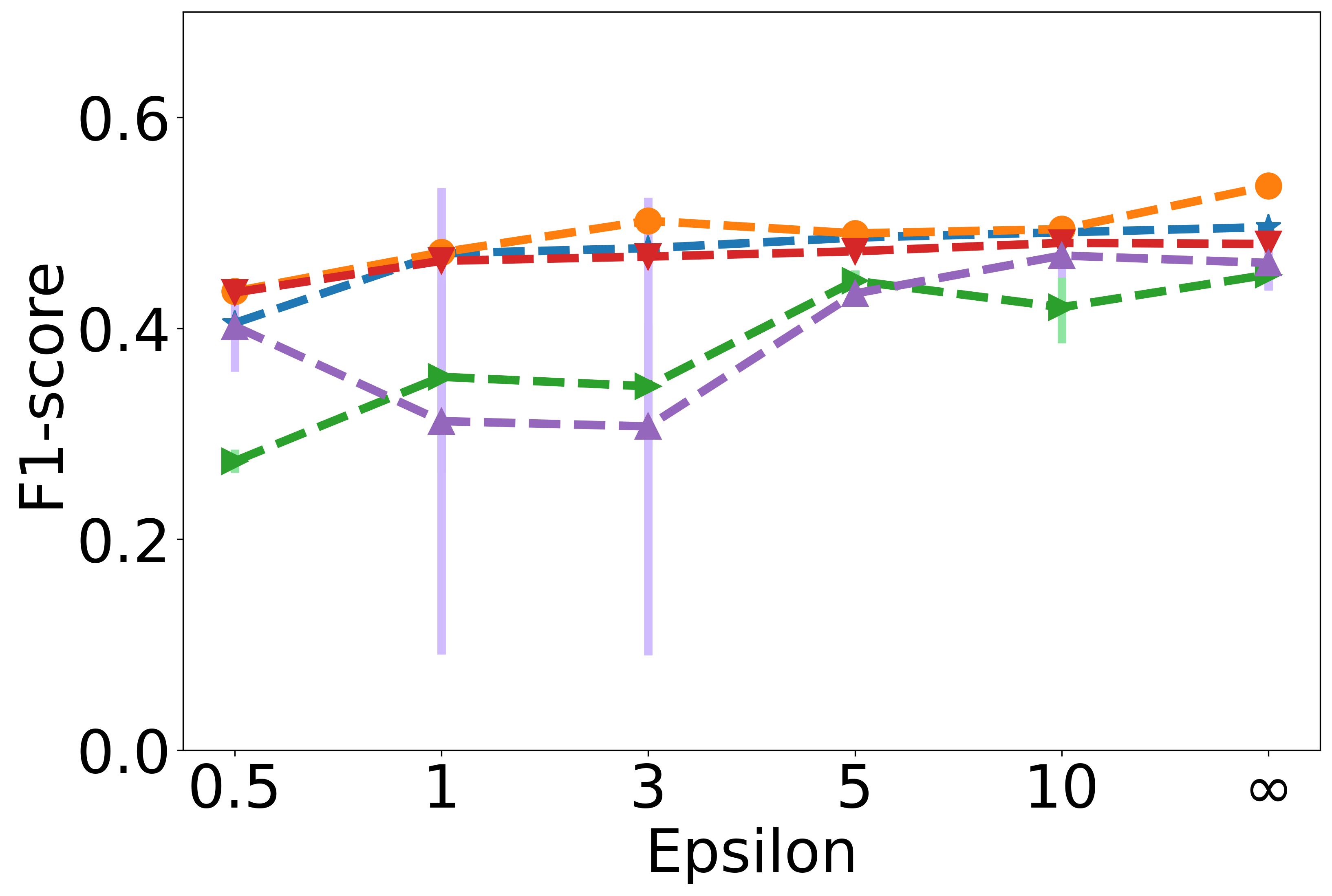}
    \includegraphics[width=\linewidth]{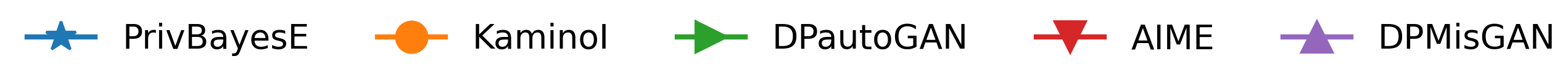}
    \caption{\revision{Adaptive methods vs baselines with MCAR missingness at varying privacy budget. }
    }
    \label{fig:exp4}
\end{figure}
\else
\begin{figure*} [htb]
    \centering
    \subfigure[1-way ($\downarrow$)]{\includegraphics[width=.32\textwidth]{images/Exp4/1-way.jpg}}
    \subfigure[2-way ($\downarrow$)]{\includegraphics[width=.32\textwidth]{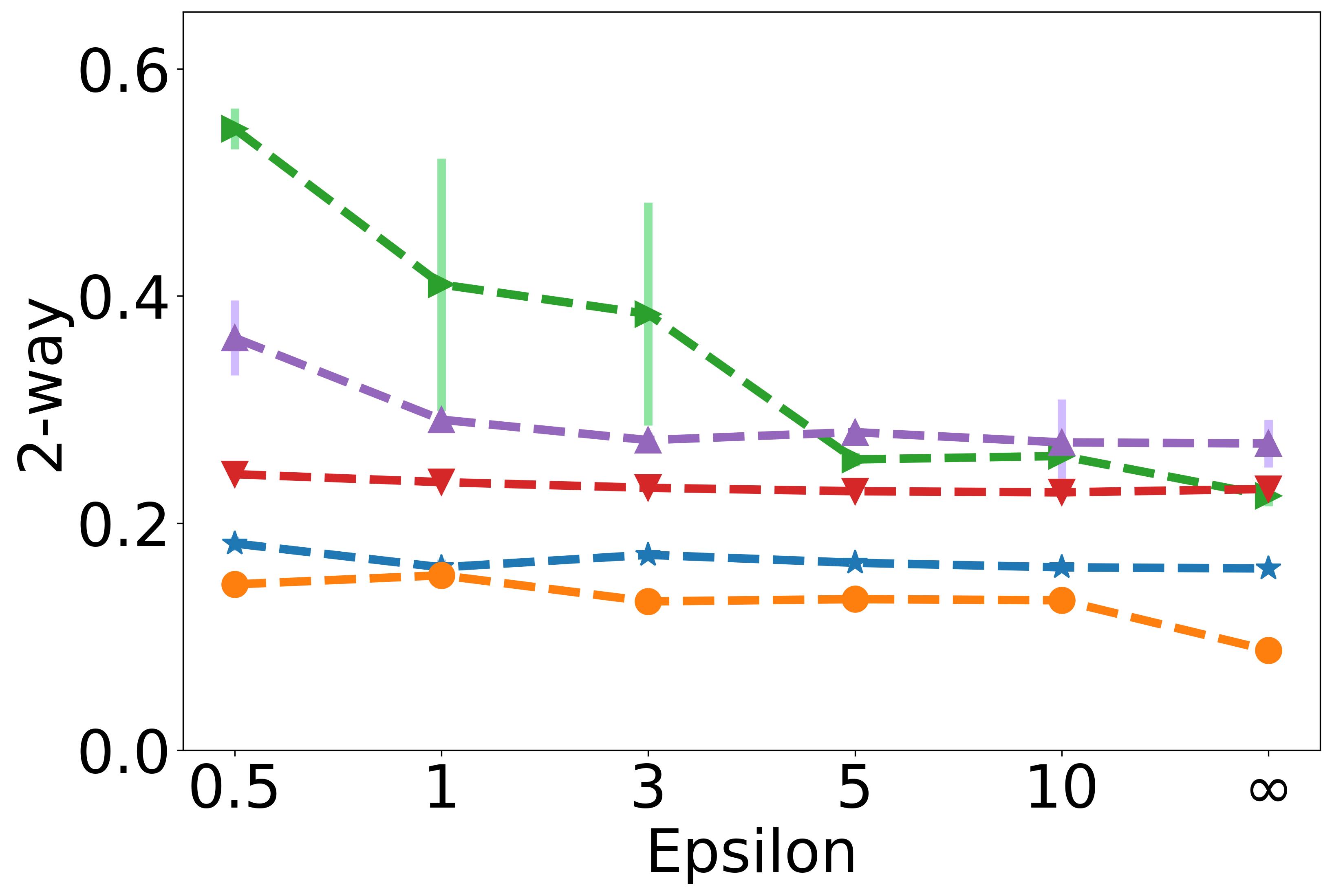}}
    \subfigure[F1-score ($\uparrow$)]{\includegraphics[width=.32\textwidth]{images/Exp4/F1-score.jpg}}
    \subfigure{\includegraphics[width=0.6\textwidth]{images/Exp4/Exp4-Legend.png}}
    \caption{Varying privacy budget and comparing utilities of our approaches}
    \label{fig:exp4}
\end{figure*}

\fi

\stitle{Varying privacy budget.} We show the impact of the privacy budget in Figure~\ref{fig:exp4} for adaptive methods with 10\% MCAR missing data on Adult. We vary the privacy budget $\epsilon \in [0.5, 1, 3, 5, 10, \infty]$ where $\epsilon = \infty$ refers to the non-private run. First, we note that increasing the privacy budget improves the utility of the synthetic dataset across all algorithms. \revision{Second, we observe that PrivBayesE and KaminoI outperformed GAN-based approaches at all privacy budgets. Similar observations are found in Figure~\ref{fig:exp2} and Figure~\ref{fig:exp3}. The poor performance of the GAN-based approaches is primarily due to the requirement for specific preprocessing steps and extensive hyperparameter tuning, which was challenging. Third, we find that there is no clear winner for all datasets, all utility metrics, and all privacy budgets, which is consistent with the observations from prior benchmarks on DP synthetic data generation~\cite{tao2021benchmarking}. However, we observe that PrivBayesE performs better at smaller epsilon values ($\epsilon \leq 1$), but KaminoI has better performance with higher epsilon ($\epsilon >1 $). We attribute this phenomenon to the fact that KaminoI trains multiple intermediate models, and these models generally require a higher privacy budget. Therefore, we recommend PrivBayesE at a low privacy regime ($\epsilon \leq  1$) and KaminoI at a higher privacy regime ($\epsilon > 1)$. }

\revision{In addition, we observe that PrivBayesE outperforms others at the 2-way tasks (Figure~\ref{fig:exp2} top row), and 
KaminoI demonstrates superior performance in the F1-score metric even at a low privacy regime ($\epsilon=1$, Figure~\ref{fig:exp2} bottom row, except BR2000). This shows that depending on the data sets and the downstream tasks different methods perform best. Hence, we also recommend to an end user to test all methods for their specific data sets and downstream tasks and assess result quality in the way we have done. Then they can select the best method for their case, but this ignores privacy cost of algorithm selection. Generally, we recommend marginal-based approach like PrivBayesE for k-way tasks and KaminoI for ML tasks.
}

\subsubsection{Amplification Due To Missingness}\label{sec:exp_amplification}
In this experiment, we show the amplified privacy budget for ground truth data. We experiment with PrivBayesE that runs on the incomplete dataset with 10 - 50\% MCAR missing data. For each run, we allocate a privacy budget of $\epsilon = 1$ and observe the marginals calculated by PrivBayesE. We assume a uniform privacy budget for each marginal calculated by PrivbayesE and run Algorithm~\ref{algo:mcar_amplification} to calculate the best valid partition of these marginals. In Table~\ref{tab:amplified_privacy}, we plot the amplified privacy cost $\bar{\epsilon}$ based on the best partition found by the algorithm. We repeat the experiment for two datasets -- Adult and BR2000. Our results show that the amplified privacy cost decreases almost linearly from 0.88x to 0.44x for Adult and 0.83x to 0.31x for the BR2000 dataset.

\begin{table}[t]
    \centering
    \caption{Amplified privacy for ground truth data.}
    \begin{tabular}{|*{6}{c|}}
         \hline
         Dataset & \multicolumn{5}{|c|}{MCAR missing \%} \\
         \cline{2-6}
         & 0.1 & 0.2 & 0.3 & 0.4 & 0.5 \\
         \hline
         Adult & 0.88 & 0.77 & 0.65 & 0.47 & 0.44 \\
         BR2000 & 0.83 & 0.68 & 0.55 & 0.41 & 0.31 \\
         \hline
    \end{tabular}
    \label{tab:amplified_privacy}
\end{table}

%% file: related.tex
\section{Related Work}\label{sec:related}
Differentially private synthetic data generation has been studied vastly in prior literature~\cite{Bowen_2020, fan_2020, DBLP:journals/tkde/ZhuLZY17, nist}. Prior works generate synthetic data via statistical approaches which estimate low-dimensional marginal distributions~\cite{DBLP:conf/sigmod/QardajiYL14, DBLP:journals/tkde/XiaoWG11}, deep learning approaches~\cite{DBLP:conf/sec/FrigerioOGD19, DBLP:conf/iclr/JordonYS19a}, or the combination of the two~\cite{ge2021kamino}. However, all of these algorithms focus on the no missing data setting. 
\revision{Some prior work look into missing data imputation for private datasets but are either not in the differential privacy setting~\cite{patki,  huang2018pacas, jagannathan2008privacy} or do not support generating synthetic data as a part of their work~\cite{DBLP:journals/tissec/CliftonHMM22, krishnan2016privateclean, DBLP:journals/corr/abs-2206-15063}. Patki et al.~\cite{patki} propose the synthetic data vault framework, which identifies and repairs inconsistencies in the generated synthetic data from incomplete data but does not consider any privacy guarantees. Their approach uses low-way marginals and learns them using Gaussian copulas that may be enhanced using our partial marginal observation approach if made private. Huang et al.~\cite{huang2018pacas} and Jagannathan et al.~\cite{jagannathan2008privacy} imputation of missing data as a cleaning algorithm for private datasets but consider privacy definitions of k-anonymity and cryptographic distributive computing, respectively. Other works focus on privately imputing missing values. The existing differentially private solutions~\cite{krishnan2016privateclean, DBLP:journals/tissec/CliftonHMM22,DBLP:journals/corr/abs-2206-15063} focus on data imputation and are difficult to adapt for the synthetic data generation problem. PrivateClean~\cite{krishnan2016privateclean} requires a human in the loop that can specify the user-defined specific imputation functions to clean the database. The k-means based word by Clifton et al.~\cite{DBLP:journals/tissec/CliftonHMM22} and imputation first approach in Das et al.~\cite{DBLP:journals/corr/abs-2206-15063} use part of the privacy budget to impute the missing values using OLS regression. We note that such OLS regression and k-means models cannot be applied to categorical attributes, and using them for imputation incurs a significant amount of privacy budget and leaves little budget for data synthesis. }

%% file: conclusion.tex
\section{Conclusion}\label{sec:conclusion}

In conclusion, our research paper presents a comprehensive study on differentially private synthetic data generation algorithms for private datasets with missing values. Our proposed adaptive recourse methods outperform classical approaches and strike a balance between privacy and utility. We also provide techniques for calculating privacy bounds and demonstrate the effectiveness of our methods through extensive experiments on real-world datasets. Our findings have important implications for privacy-preserving data sharing and analysis, and can facilitate the development of more effective methods for generating synthetic data in various practical applications.